\documentclass[aps,prx,twocolumn,amsmath,amssymb,superscriptaddress, notitlepage,nofootinbib,10pt]{revtex4-2}
\usepackage{graphicx,amssymb,amsmath,amsfonts,empheq,braket,bm}
\usepackage[dvipsnames]{xcolor}
\usepackage{dsfont}
\usepackage{dcolumn}
\usepackage{textcomp}
\usepackage[normalem]{ulem}
\usepackage{physics}
\usepackage{mathtools}
\usepackage{soul}
\usepackage{enumitem}
\usepackage{xr}
\usepackage{diagbox}
\usepackage{multirow}
\usepackage{tikz}
\usepackage{quantikz}
\usepackage{hyperref}
\usepackage{amsthm}
\usetikzlibrary{external}

\hypersetup{
colorlinks = true,
linkcolor = [rgb]{0,0,0}, 
citecolor = [rgb]{0.13,0.55,0.13},
urlcolor = [rgb]{0.25, 0.41, 0.88}}

\newtheorem{proposition}{Proposition}[section]
\newtheorem{lemma}{Lemma}

\newcommand{\eqnref}[1]{Eq.\,\eqref{#1}}
\newcommand{\figref}[1]{Fig.\,\ref{#1}}

\newcommand{\secref}[1]{Sec.\,\ref{#1}}
\newcommand{\appref}[1]{Appendix.\,\ref{#1}}

\renewcommand{\ket}[1]{| #1 \rangle}
\renewcommand{\bra}[1]{\langle #1 |}

\renewcommand{\leq}{\leqslant}

\newcommand{\bbF}{\mathbb{F}}

\newcommand{\calD}{\mathcal{D}}

\newcommand{\calH}{\mathcal{H}}

\newcommand{\calM}{\mathcal{M}}

\newcommand{\RNum}[1]{\uppercase\expandafter{\romannumeral #1\relax}}

\makeatletter
\renewcommand*\env@matrix[1][*\c@MaxMatrixCols c]{%
  \hskip -\arraycolsep
  \let\@ifnextchar\new@ifnextchar
  \array{#1}}

\newtheorem*{rep@theorem}{\rep@title}
\newcommand{\newreptheorem}[2]{%
\newenvironment{rep#1}[1]{%
 \def\rep@title{#2 \ref{##1}}%
 \begin{rep@theorem}}%
 {\end{rep@theorem}}}
\makeatother

\newreptheorem{theorem}{Theorem}
\newreptheorem{conclusion}{Conclusion}
\newreptheorem{lemma}{Lemma}
\newreptheorem{proposition}{Proposition}

\begin{document}

\title{
\textcolor{black}{
A Universal Circuit Set Using the \texorpdfstring{$S_3$}{S3} Quantum Double}
}

\author{Liyuan Chen}
\email{liyuanchen@fas.harvard.edu}
\affiliation{Department of Physics, Harvard University, Cambridge, Massachusetts 02138, USA}
\affiliation{John A. Paulson School of Engineering and Applied Sciences, Harvard University, Cambridge, Massachusetts 02138, USA}

\author{Yuanjie Ren}%
\email{yuanjie@mit.edu}
\affiliation{Department of Physics, Massachusetts Institute of Technology, Cambridge, Massachusetts 02139, USA
}

\author{Ruihua Fan}
\email{ruihua\_fan@berkeley.edu}
\affiliation{Department of Physics, University of California, Berkeley, CA 94720, USA}

\author{Arthur Jaffe}
\email{arthur\_jaffe@harvard.edu}
\affiliation{Department of Physics, Harvard University, Cambridge, Massachusetts 02138, USA}

\date{\today}

\begin{abstract}
One potential route toward fault-tolerant universal quantum computation is to use non-Abelian topological codes. 
In this work, we investigate how to achieve this goal with the quantum double model $\mathcal{D}(S_3)$---a specific non-Abelian topological code.
By embedding each on-site Hilbert space into a qubit-qutrit pair, we give an explicit construction of the circuits for creating, moving, and locally measuring all non-trivial anyons. 
We also design a specialized anyon interferometer to remotely measure the total charge of well-separated anyons; this avoids fusion, which would compromise fault tolerance. 
These protocols enable the implementation of a universal gate set proposed by Cui et al. and active quantum error correction of the circuit-level noise during the computation process.
To further reduce the error rate and facilitate error correction, we encode each physical degree of freedom of $\mathcal{D}(S_3)$ into a novel, quantum, error-correcting code, enabling fault-tolerant realization, at the logical level, of all gates in the anyon manipulation circuits. 
Our proposal offers a promising path to realize robust universal topological quantum computation in the NISQ era.

\end{abstract}

\maketitle

\tableofcontents

\section{\label{sec:introduction}Introduction}

One essential task of quantum information science is to realize fault-tolerant universal quantum computation~\cite{Shor:1996qc,Shor:1996qc,Gottesman:2009zvw,Terhal:2013vbm,Campbell2017Nature_FTUQC,Nielsen_Chuang_2010}. Topological quantum computation, leveraging the stability of topological orders under local perturbations, is one promising route toward this goal~\cite{Kitaev2003Annals_FT_anyons,freedman2003topological,Nayak2008RMP_TQC}. 
Currently, most of the theoretical and experimental efforts focus on topological stabilizer codes, such as the surface code and color code~\cite{Dennis2002JMP_TC,Fowler2012_SurfaceCode,bombin2015gaugecolor,Kubica2015PRA_3DColorCode,Lukin2021SpinLiquid,Satzinger_2021_toric_code,Bluvstein2022,fossfeig2023experimentaldemonstrationadvantageadaptive,Bluvstein_2023,acharya2024quantumerrorcorrectionsurface,reichardt2024demonstrationquantumcomputationerror,berthusen2024experiments4dsurfacecode}. 
However, to realize a transversal universal gate set within a single topological stabilizer code, additional resources, such as magic-state distillation or extra spatial dimensions, are indispensable~\cite{Eastin_2009,Bravyi2005PRA_MSD,Bravyi2012PRA_MSD_LowOverhead,Campbell2012PRX_MDS_PrimeDim,Haah2018codesprotocols,Hastings2018PRL_MSD_sublogrithmic}. 
More recently, several experimental groups have also reported on realizing non-Abelian topological codes with a high fidelity~\cite{Andersen2023Nature_Google_NonAbelian,Iqbal2024Nature_Quantinuum_non_Abelian,Iqbal2024_Feedforward,Xu2024_Fibonacci,minev2024realizingstringnetcondensationfibonacci,iqbal2024qutrittoriccodeparafermions}.
It is therefore natural to examine how to use non-Abelian topological code for robust universal quantum computation in the noisy intermediate-scale quantum (NISQ) era~\cite{Preskill2018quantumcomputingin}.

There are multiple options, with two well-known examples being the quantum double models $\calD(G)$ for a non-solvable group $G$, and Fibonacci topological orders, both of which enable universal computation through braiding~\cite{Kitaev2003Annals_FT_anyons,Freedman2002_Fibonacci_UQC}. 
However, they struggle to balance computational power with experimental feasibility.
The quantum double model for the minimal non-solvable group $A_5$ requires an impractically large Hilbert space dimension per site (equal to the group order of $A_5$). 
The (double) Fibonacci topological order, realizable by a string-net model based on qubits~\cite{Levin2005PRB_LevinWen,Koenig2010Anns_TVCode,Alexis2022PRX_Threshold_TVCode,Liu_2022_PRXQ_Simulating_String_Net,Xu2024_Fibonacci} involves high-weight syndrome measurement (up to 16) and two non-Clifford gates in the circuits, which makes realization and error control difficult~\cite{Alexis2022PRX_Threshold_TVCode,Xu2024_Fibonacci}.
It is both experimentally and theoretically interesting to identify a minimal setup capable of performing universal quantum computation. 

The quantum double model for the solvable group $S_3$ has been shown to have an exceptional computational power, which has a relatively small local Hilbert space simultaneously. 
It was first shown by Mochon that braiding and fusion of anyons enable a universal gate set, on the local degrees of freedom of fluxons~\cite{Mochon_2004}.
Cui, Hong, and Wang~\cite{Cui2015QIP_UQC_Weakly} extended the encoding to the fusion space of multiple anyons and showed that braiding, together with specific anyon charge measurements forms a universal gate set. This is called the $U$-model. 
Here, we focus on $\calD(S_3)$ and follow the computation scheme introduced by Cui et al. to explore an experimentally feasible realization that is also resilient to noise.

Specifically, in their construction, the logical information is stored non-locally in the degenerate fusion subspace of well-separated non-Abelian anyons, that is robust under local perturbation. 
The computation involves two basic components, moving a single anyon and measuring the total charge of well-separated anyons, both of which are tricky to realize in the presence of errors.
Unlike the Abelian anyons, a pair of well-separated non-Abelian anyons cannot be generated using a constant-depth unitary circuit~\cite{Shi:2018bfb}.
One has to use either a linear-depth unitary circuit or constant-depth adaptive circuit to create and manipulate them~\cite{bravyi2022adaptive,Liu_2022_PRXQ_Simulating_String_Net,lyons2024protocolscreatinganyonsdefects}. 
In the presence of noise, the non-transversal nature of these circuits can propagate a single error along the anyon string, which makes error control difficult. 
Furthermore, the naive anyon charge measurement is performed through fusion, which requires moving anyons close to each other and immediately lowers the resilience against noise.

In this work, we explicitly construct protocols to manipulate anyons that allow us to move anyons and measure anyon charges in a more robust fashion. 
First, we develop an adaptive protocol to coherently move non-Abelian anyons incrementally to avoid error propagation. 
This serves as the basic toolkit for both computation and quantum error correction. 
Second, we borrow ideas from anyon interferometry~\cite{Bonderson_2008_Interferometry} and invent an anyon charge measurement, which, in particular, allows us to realize all measurement protocols in Cui et al.'s computation model.
Importantly, anyons encoding the logical information are kept at a large distance during the entire process, making our method more resilient against noise than the naive fusion.
Furthermore, we provide explicit prescriptions to realize all these operations using quantum circuits for qubits and qutrits. 
Notably, our circuits only involve a single non-Clifford gate and weight-8 syndrome measurements, and no post-selection is required, which makes our scheme accessible to the current experimental platforms.

It is natural to inquire how to perform quantum error correction (QEC), especially at the circuit level, using the toolkits that we have provided. 
We show that we can convert the circuit-level noise into the phenomenological error model where errors are in the form of incoherent anyon pairs. 
Therefore, instead of handling the circuit-level noise directly, it suffices to apply the QEC algorithms at the phenomenological level~\cite{Wotton2014PRX_QEC_S3,Brell_2014_Ising_QEC,Wotton2016PRA_ActiveEC,Hutter_2016_Continuous_Ising, Burton_2017_Fibonacci_QEC,Dauphinais_2017_FT_QEC_non_cyclic,Alexis2022PRX_Threshold_TVCode,schotte2022faulttolerant}.
However, the non-Abelian nature of the error forbids us from correcting them instantaneously, which can lead to an extremely small error threshold~\cite{Wotton2016PRA_ActiveEC,Hutter_2016_Continuous_Ising,Dauphinais_2017_FT_QEC_non_cyclic,schotte2022faulttolerant}.
To circumvent this problem, we propose a concatenation architecture, in which we replace each local degree of freedom in the quantum double model with a small quantum error correcting code.
We show that all the gates used in the circuit have a fault-tolerant realization at the logical level of these local codes, which can help suppress the effective error rate efficiently.

To summarize, in this paper, we present a comprehensive blueprint for implementing fault-tolerant universal quantum computation using the quantum double model $\mathcal{D}(S_3)$. While its advantages over existing protocols may not yet be immediately apparent, we argue that $\mathcal{D}(S_3)$ represents a promising candidate for realizing large-scale quantum computation, for the following reasons:
\begin{itemize}[leftmargin=1.5em, itemsep=0em]
    \item First, the intrinsic universal computational power of $\mathcal{D}(S_3)$ eliminates the need for resource-intensive techniques such as magic-state distillation or the additional spatial dimensions required by protocols based on stabilizer codes. 
    \item Second, its straightforward circuit implementation ensures experimental feasibility and supports effective error correction. In particular, the concatenation scheme enabled by our novel error-correcting codes enhances the ability to suppress physical error rates below the threshold of $\mathcal{D}(S_3)$, offering an advantage over other approaches, such as those based on the Fibonacci code. 
    \item Finally, the potential for fault-tolerant constant-depth realizations of $\mathcal{D}(S_3)$ could significantly reduce the overhead typically associated with non-Abelian topological orders.
\end{itemize}
This work serves as a foundation for future exploration of these directions, which are crucial for demonstrating the advantages of topological quantum computation.

The remainder of this paper consists of five sections. 
\secref{sec:preliminaries} reviews the quantum double model. 
\secref{sec:UQC} presents our protocols for the manipulation and measurement of anyons in $\mathcal{D}(S_3)$ that are necessary for the universal gate set in the $U$-model. 
In \secref{sec:quantum_double_circuit}, we discuss the circuit realization for the basic operations of computation---ribbon operators and local anyon type measurements. 
\secref{sec:S3_QEC} gives remark on the quantum error correction. We explain how to incorporate active error correction for the circuit-level noise into computation and introduce a concatenation scheme for further suppressing the physical error rate. 
In \secref{sec:SummaryOutlook}, we conclude with a brief summary and an outlook.

\section{Review of quantum double models}
\label{sec:preliminaries}

The quantum double model is an exactly solvable lattice model that is designed based on the Drinfeld quantum double of a finite group $G$~\cite{Kitaev2003Annals_FT_anyons}. 
We abuse the notation and denote both the quantum model and the Drinfeld double by $\mathcal{D}(G)$. 
In this section, we review the Hamiltonian and anyonic content of $\mathcal{D}(G)$. 
We also provide a review of the Drinfeld quantum double algebra $\mathcal{D}(G)$ in \appref{appendix:Drinfeld_quantum_double} for the mathematically inclined reader.

For a finite group $G$, we denote a group element by $g \in G$ and its inverse by $\bar{g}$. 
The irreducible representations (irreps) of the group $G$ and the corresponding Drinfeld double $\mathcal{D}(G)$ are labeled by Greek letters $\alpha,\beta,\gamma,\dots$. 
For any $g\in G$, we label its representation matrix in an irrep $\alpha$ by $\Gamma^\alpha(g)$ and the character by $\chi_\alpha(g)$. The fusion rules for irreps are given by
\begin{equation} \label{eqn:fusion_irreps}
    \alpha \times \beta = \sum_\gamma N_{\alpha\beta}^\gamma \gamma\;,
\end{equation}
where $N_{\alpha\beta}^\gamma$ is the multiplicity of the irrep $\gamma$. 

\begin{figure*}[tb]
    \centering
    \includegraphics[width = \textwidth]{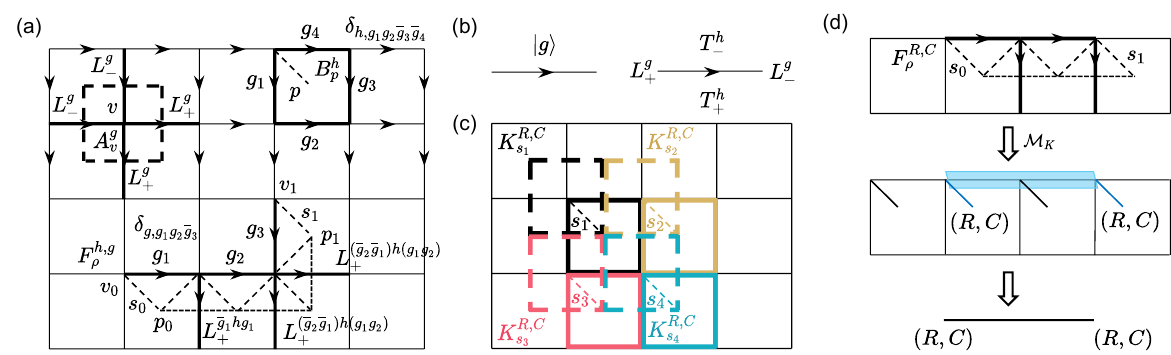}
    \caption{Operators of the quantum double model $\mathcal{D}(G)$. (a) The vertex operator $A^g_v$, plaquette operator $B^h_p$, and ribbon operator $F^{h,g}_\rho$ from $s_0 = (v_0,p_0)$ to $s_1 = (v_1,p_1)$, defined on a directed square lattice with horizontal/vertical edges pointing rightward/downward. (b) (Left) The orthonormal basis $\{\ket{g}\}$ of the Hilbert space on each directed edge, and (Right) the action of $L$ and $T$ operators on each edge. (c) The neighboring anyon type measurement operators $K^{R,C}_s$ for $s=s_1,\cdots,s_4$. (d) (Upper) The exicted state $F^{R,C}_\rho\ket{\Omega}$. (Middle) The configuration resulting from anyon type measurement $\mathcal{M}_K$, with blue lines indicating the presence of $(R,C)$ anyons at the ends of $\rho$ (light blue band). (Lower) The simplified representation using lines with anyon labels $(R,C)$ at their ends, omitting the square lattice.}
    \label{fig:ops_quantum_double}
\end{figure*}
 
Based on the group $G$, we define the quantum double model $\mathcal{D}(G)$ on a two-dimensional squared lattice as follows. 
We associate each edge with a $|G|$-dimensional Hilbert space $\mathcal{H}_G = \mathbb{C}G$ that has a orthonormal basis $\{\ket{g}:g\in G\}$ labeled by the group elements. 
As we will see later, it is generally convenient to consider directed edges to describe how the operator acts on the basis state, as shown in \figref{fig:ops_quantum_double}(a). Our convention is that, all horizontal and vertical edges point rightward and downward, respectively.
Specifically, there are four basic linear operators acting on $\mathcal{H}_G$
\begin{equation} \label{eqn:LT_ops}
\begin{aligned}
    & L^g_+ \ket{m} = \ket{gm}\;, \quad L^g_- \ket{m} = \ket{m\bar{g}}\;,\\
    & T^h_+\ket{m} = \delta_{h,m}\ket{m}\;,\quad T^h_-\ket{m} = \delta_{\bar{h},m}\ket{m}\;.
\end{aligned}
\end{equation}
Namely, $L^g$ does group multiplication and $T^h$ is the projection onto a specific basis state.
Pictorially, $L^g$ takes the $+$ ($-$) sign if it acts from the starting (ending) point of an edge $e$, while $T^h$ takes the $+$ ($-$) sign if it acts from the right (left) hand side of $e$. 
The right panel of \figref{fig:ops_quantum_double}(b) illustrates the conventions for the actions of $L_{\pm}^g$ and $T_{\pm}^h$ on a horizontal edge; the corresponding conventions for a vertical edge can be inferred by rotating the diagram clockwise by $90^\circ$. For a formal definition grounded in the language of direct and dual triangles, we refer readers to Appendix~\ref{appendix:ribbon_algebra} and to Ref.~\cite{Bombin2008PRB_NonAbelian}.

Based on these $L$ and $T$ operators, one can define the vertex operator $A^g_v$ and plaquette operator $B^h_p$, for any vertex $v$ and plaquette $p$, respectively, as illustrated in \figref{fig:ops_quantum_double}(a). 
In our choice of edge directions, these operators are defined as follows:
\begin{equation} \label{eqn:A_B_ops_DG}
\begin{aligned}
    A^{g}_v =\ & \begin{tikzpicture}[baseline = {(current bounding box.center)}]
        \draw (0,0) -- (0.8,0);
        \draw (0.4,0) node[inner sep=0.5pt, fill=white] {\scriptsize $L^g_-$};
        \draw (0.8,-0.8) -- (0.8,0);
        \draw (0.8,-0.4) node[inner sep=0.5pt, fill=white] {\scriptsize $L^g_+$};
        \draw (0.8,0) -- (1.6,0);
        \draw (1.2,0) node[inner sep=0.5pt, fill=white] {\scriptsize $L^g_+$};
        \draw (0.8,0) -- (0.8,0.8);
        \draw (0.8,0.4) node[inner sep=0.5pt, fill=white] {\scriptsize $L^g_-$};
    \end{tikzpicture}\;,\\
    B_p^{h} =\ & \sum_{\{g_i\}}\delta_{h,g_1g_2\bar{g}_3\bar{g}_4} \ket{\begin{tikzpicture}[baseline={(current bounding box.center)}]
        \draw (0,0) -- (0,0.6);
        \draw (0,0.3) node[inner sep=0.5pt, fill=white] {\scriptsize $g_1$};
        \draw (0,0.6) -- (0.6,0.6);
        \draw (0.3,0.6) node[inner sep=0.5pt, fill=white] {\scriptsize $g_4$};
        \draw (0.6,0.6) -- (0.6,0);
        \draw (0.6,0.3) node[inner sep=0.5pt, fill=white] {\scriptsize $g_3$};
        \draw (0,0) -- (0.6,0);
        \draw (0.3,0) node[inner sep=0.5pt, fill=white] {\scriptsize $g_2$};
    \end{tikzpicture}}
    \bra{\begin{tikzpicture}[baseline={(current bounding box.center)}]
        \draw (0,0) -- (0,0.6);
        \draw (0,0.3) node[inner sep=0.5pt, fill=white] {\scriptsize $g_1$};
        \draw (0,0.6) -- (0.6,0.6);
        \draw (0.3,0.6) node[inner sep=0.5pt, fill=white] {\scriptsize $g_4$};
        \draw (0.6,0.6) -- (0.6,0);
        \draw (0.6,0.3) node[inner sep=0.5pt, fill=white] {\scriptsize $g_3$};
        \draw (0,0) -- (0.6,0);
        \draw (0.3,0) node[inner sep=0.5pt, fill=white] {\scriptsize $g_2$};
    \end{tikzpicture}}\;.
\end{aligned}
\end{equation}
We define a site by a tuple $s = (v,p)$ consisting of a plaquette $p$ and its northwest vertex $v$, as illustrated by the thin dashed lines in \figref{fig:ops_quantum_double}(a). 
At different sites, $A$ and $B$ operators  commute. 
At the same site $s=(v,p)$, they satisfy the so-called Drinfeld double algebra (see Appendix~\ref{appendix:Drinfeld_quantum_double})
\begin{equation}
\begin{gathered}
    A^g_vA^{g^\prime}_v = A^{gg^\prime}_v\;,\quad B^{h}_pB^{h^\prime}_p = \delta_{h,h^\prime} B^h_p\;,\\
    A^g_vB^h_p = B^{gh\bar{g}}_p A^g_v\;.
\end{gathered}
\end{equation}
Therefore, the local Hilbert space at each site $s$ can be decomposed into a direct sum of different irreps of the Drinfeld double $\calD(G)$. Each irrep of $\calD(G)$ is labeled by a tuple $(R,C)$ where $C$ represents a conjugacy class of $G$, and $R$ is an irrep of the centralizer $Z(C)$ of $C$ (see Appendix~\ref{appendix:quantum_double}). The dimension of the irrep $(R,C)$ is $d_{R,C} = |R||C|$, where $|R|$ is the dimension of $R$ and $|C|$ is the number of elements in $C$.

Physically, $A^g_v$ implements local gauge transformations and $B^h_p$ represents the magnetic charge operators in the language of lattice gauge theory. 
Accordingly, we can define the projectors onto the zero charge and flux subspaces as follows:
\begin{equation} \label{eqn:A_B_commuting_projectors}
    A_v = \frac{1}{|G|}\sum_{g\in G}A^g_v,\quad B_p = B^e_p\;.
\end{equation}
The Hamiltonian of $\mathcal{D}(G)$ is the summation of these commuting projectors on each site, defined by
\begin{equation} \label{eqn:Hamiltonian_quantum_double}
    H_{G} = -\sum_{v}A_v - \sum_{p} B_p\;.
\end{equation}
On a planar lattice, the model has a unique ground state $\ket{\Omega}$ satisfying $A_v\ket{\Omega} = B_p\ket{\Omega} = \ket{\Omega}$ indicating zero charge and no flux everywhere.

The quantum double Hamiltonian \eqnref{eqn:Hamiltonian_quantum_double} has anyonic excitations that are in one-to-one correspondence with the irrep of Drinfeld double.
Given a ribbon $\rho$ that connects the starting site $s_0$ and the ending site $s_1$ as depicted in \figref{fig:ops_quantum_double}(a), we can create a pair of $(R,C)$ anyons localized at the two end sites by applying a set of ribbon operators $F^{R,C}_{\rho}$ to the ground state~\cite{Bombin2008PRB_NonAbelian}:
\begin{equation} \label{eqn:F_RC}
    F^{R,C;u,v}_{\rho} = \frac{|R|}{|Z(C)|}\sum_{n\in Z(C)} \Gamma^{R}_{jj^\prime}(n)F^{c,\tau_c n \bar{\tau}_{c^\prime}}_\rho\;.
\end{equation}
On the right-hand side, $\Gamma^{R}_{jj^\prime}(n)$ is the matrix of the irrep $R$, $\tau_c$ is the group element satisfying $\tau_c g_c \bar{\tau}_c = c$ with $g_c$ being a representative of $C$, and $F^{h,g}_\rho$ with $h,g \in G$ is an operator defined on the same support $\rho$, whose action is detailed in Appendix~\ref{appendix:ribbon_algebra}. 
On the left-hand side, $u = (c,j)$ and $v=(c^\prime,j^\prime)$ for $c,c^\prime \in C$ are the local degrees of freedom of $(R,C)$ at $s_0$ and $s_1$.
The quantum dimension of the anyon is given by the dimension of the irrep $(R,C)$.
When $(R,C)$ has a dimension greater than 1, the corresponding anyon is a non-Abelian anyon. The fusion of various non-Abelian anyons has multiple outcomes, which is a defining feature of non-Abelian anyons.

The anyon types can be measured locally using a set of closed ribbon operators. For a closed ribbon $\sigma$ that starts and ends at the same site $s$, the operator that measures the anyon type $(R,C)$ inside $\sigma$ is defined as:
\begin{equation} \label{eqn:K_RC}
    K^{R,C}_{\sigma} = \frac{|R|}{|Z(C)|} \sum_{n\in Z(C),c\in C} \bar{\chi}_R(n) F^{\tau_c n \bar{\tau}_c,c}_\sigma\;,
\end{equation}
where $\chi_R$ is the character of $R$. The set of $K^{R,C}_\sigma$ operators satisfy the following conditions \cite{Bombin2008PRB_NonAbelian,Beigi2011CMP_QDS3,bravyi2022adaptive,Cui2018TopologicalQC}:
\begin{equation}
\begin{aligned}
    K^{R_1,C_1}_\sigma K^{R_2,C_2}_\sigma &= \delta_{R_1,R_2}\delta_{C_1,C_2}K^{R_1,C_1}_\sigma\;, \\
    \sum_{(R,C) \in \mathrm{Irr}(\mathcal{D}(G))} K^{R,C}_\sigma &= 1\;,
\end{aligned}
\end{equation}
indicating that $K^{R,C}_\sigma$'s form a set of orthogonal projective measurements.
As a special case, we can define a single-site measurement at $s = (v,p)$ using the smallest closed ribbons $A^g_v$ and $B^h_p$ 
as~\cite{Bombin2008PRB_NonAbelian}
\begin{equation} \label{eqn:K_RC_single}
    K^{R,C}_{s} = \frac{|R|}{|Z(C)|} \sum_{n\in Z(C),c\in C} \bar{\chi}_R(n)A^{\tau_c n \bar{\tau}_c}_v B^c_p\;.
\end{equation}
For example, $K^{R,C}_{s}$ for the trivial irrep and conjugacy class $(R,C) = ([+],\{e\})$ is reduced to
\begin{equation}
    K^{[+],\{e\}}_s = A_v B_p\;,
\end{equation}
i.e., the projector onto the space of no charge and flux.
Notably, 
\begin{equation}
    [K^{R,C}_{s},K^{R,C}_{s^\prime}] = 0
\end{equation} 
for $s\neq s^\prime$. 
Therefore, by selecting all the northwest sites $\{s\}$, we can define a commuting set of measurement, as illustrated in \figref{fig:ops_quantum_double}(c). 
We define 
\begin{equation}
\begin{gathered}
    \mathcal{M}_K \equiv \{K^{R,C}_{s}\}, \\ \forall (R,C) \in \mathrm{Irr}(\calD(G)),\  \forall s = (v,p)
\end{gathered}
\end{equation}
as the set of such measurements covering every site of the lattice. 
In a planar lattice, such a projection is surjective, with degeneracies arising from non-Abelian anyons' local degrees of freedom. 
The measurement $\mathcal{M}_K$ will project a generic state onto a configuration of anyons, with an anyon label on each site. 
We refer to it as the anyon configuration picture. 
Consider an example as shown in \figref{fig:ops_quantum_double}(d). Given an excited state $F^{R,C}_\rho\ket{\Omega}$, performing $\mathcal{M}_K$ yields a configuration shown in the middle panel. The light blue band represents the ribbon $\rho$, and the blue lines at its two ends, $s_0$ and $s_1$, indicate the presence of a pair of $(R,C)$ anyons. We can also simply use a line with anyon labels at its ends to denote such a configuration, as shown in the lower panel. Each anyon configuration has a one-to-one correspondence with an anyon fusion tree. After performing $\calM_K$, we can always consider the subsequent operations by focusing on the corresponding anyon fusion tree without worrying about the details of the lattice model.

The basic topological operations of anyons are braiding and fusion. The braiding is the exchange of two anyons, resulting in a phase factor depending on the anyon types. The fusion is realized by bringing two anyons to a single site and measuring their total charge by $\mathcal{M}_K$. In the lattice model, we can use $F^{R,C}_\rho$ and $\mathcal{M}_K$ to realize these operations. 

We conclude this section with a concrete example $\calD(S_3)$, the quantum double model based on the permutation group of three elements. This group has the isomorphism $S_3 \cong \mathbb{Z}_3 \rtimes \mathbb{Z}_2$, with two generators $\mu \in \mathbb{Z}_3$ and $\sigma \in \mathbb{Z}_2$ satisfying the semidirect product relation $\sigma \mu \sigma = \bar{\mu}$. 
It has two 1-dimensional irreps $[+]$ and $[-]$ with:
\begin{equation}
    \Gamma^{[\pm]}(\mu) = 1\;,\quad \Gamma^{[\pm]}(\sigma) = \pm 1\;,
\end{equation}
and a 2-dimensional irrep $[2]$ with:
\begin{equation}
\begin{gathered}
    \Gamma^{[2]}(\mu) = \begin{pmatrix}
        \omega & 0 \\
        0 & \bar{\omega}
    \end{pmatrix}\;,\quad \Gamma^{[2]}(\sigma) = \begin{pmatrix}
        0 & 1 \\
        1 & 0
    \end{pmatrix}\;, \\
    \Gamma^{[2]}(\bar{\mu}) = \Gamma^{[2]}(\sigma)\Gamma^{[2]}(\mu)\Gamma^{[2]}(\sigma) = \begin{pmatrix}
        \bar{\omega} & 0 \\
        0 & \omega
    \end{pmatrix}\;.
\end{gathered}
\end{equation}
Here $\omega = \exp(i2\pi/3)$ denotes the third root of unity. 
They satisfy the following fusion rules
\begin{equation}
\begin{aligned}
    &[-]\times [-] = [+]\;,\quad [-]\times [2] = [2]\;,\\
    &[2]\times [2] = [+]+[-]+[2]\;,
\end{aligned}
\end{equation}
where the multiple fusion results in the last rule are a common feature in the fusion of high-dimensional (dimension $> 1$) irreps. 

The group $S_3$ has three conjugacy classes $C_1 = \{e\}, C_2 = \{\sigma,\mu\sigma,\mu^2\sigma\}$ and $C_3 = \{\mu,\mu^2\}$, where centralizers are $Z(C_1) = S_3, Z(C_2) = \{e,\sigma\}\cong \mathbb{Z}_2$, and $Z(C_3) = \{e,\mu,\mu^2\} \cong \mathbb{Z}_3$. Consequently, there are eight types of anyons in $\mathcal{D}(S_3)$, denoted by $A$ through $H$, as listed in Table~\ref{tab:conjugacy_irreps}. The $A$ and $B$ anyons are Abelian, while the others are non-Abelian anyons. These anyons have the fusion rules as listed in Table~\ref{tab:fusion_rules}. 

\begin{table}[h]
\centering
\begin{tabular}{|c||c@{}c@{}c|c@{}c|c@{}c@{}c|}
\hline
  \textbf{anyon type} & \textbf{A} & \textbf{B} & \textbf{C} & \textbf{D} & \textbf{E} & \textbf{F} & \textbf{G} & \textbf{H} \\ 
\hline
$\bm{C}$ & \(C_1\) & \(C_1\) & \(C_1\) & \(C_2\) & \(C_2\) & \(C_3\) & \(C_3\) & \(C_3\) \\ 
\hline
$\bm{R}$ & \([+]\) & \([-]\) & \([2]\) & \([+]\) & \([-]\) & \([1]\) & \([\omega]\) & \([\bar{\omega}]\) \\ 
\hline
$\bm{d_{R,C}}$ & \(1\) & \(1\) & \(2\) & \(3\) & \(3\) & \(2\) & \(2\) & \(2\) \\
\hline
\end{tabular}
\caption{Conjugacy classes \( C \), irreducible representations \( R \) of the corresponding centralizer subgroups \( Z(C) \), and the associated quantum dimensions \( d_{R,C} \) for anyons in the quantum double model \( \mathcal{D}(S_3) \).}

\label{tab:conjugacy_irreps}
\end{table}

\begin{table*}[tb]
\centering 
\begin{tabular}{|c||c@{\hspace{10pt}}c@{\hspace{10pt}}c|c@{\hspace{10pt}}c|c@{\hspace{10pt}}c@{\hspace{10pt}}c|}
\hline
\(\times\) & \textbf{A} & \textbf{B} & \textbf{C} & \textbf{D} & \textbf{E} & \textbf{F} & \textbf{G} & \textbf{H} \\ 
\hline
\textbf{A} & A & B & C & D & E & F & G & H \\ 
\textbf{B} & B & A & C & E & D & F & G & H \\ 

\textbf{C} & C & C & A \(+ \) B \(+ \) C & D \(+ \) E & D \(+ \) E & G \(+ \) H & F \(+ \) H & F \(+ \) G \\ 
\hline
\textbf{D} & D & E & D \(+ \) E & A \(+ \) C \(+ \) F \(+ \) G \(+ \) H & B \(+ \) C \(+ \) F \(+ \) G \(+ \) H & D \(+ \) E & D \(+ \) E & D \(+ \) E \\ 
\textbf{E} & E & D & D \(+ \) E & B \(+ \) C \(+ \) F \(+ \) G \(+ \) H & A \(+ \) C \(+ \) F \(+ \) G \(+ \) H & D \(+ \) E & D \(+ \) E & D \(+ \) E \\ 
\hline
\textbf{F} & F & F & G \(+ \) H & D \(+ \) E & D \(+ \) E & A \(+ \) B \(+ \) F & H \(+ \) C & G \(+ \) C \\ 
\textbf{G} & G & G & F \(+ \) H & D \(+ \) E & D \(+ \) E & H \(+ \) C & A \(+ \) B \(+ \) G & F \(+ \) C \\ 
\textbf{H} & H & H & F \(+ \) G & D \(+ \) E & D \(+ \) E & G \(+ \) C & F \(+ \) C & A \(+ \) B \(+ \) H \\ 
\hline
\end{tabular}
\caption{Fusion rules for anyons in the quantum double model \(\mathcal{D}(S_3)\).}
\label{tab:fusion_rules}
\end{table*}

\section{\label{sec:UQC} Implementing basic gadgets for universal computation with \texorpdfstring{$\mathcal{D}(S_3)$}{D(S3)}}

In the quantum double model, excitations created by the ribbon operators are characterized by local degrees of freedom and anyon types. While the former are susceptible to local noise, the latter are more robust~\cite{Bombin2008PRB_NonAbelian}.
Cui et al. proposed a so-called $U$-model that employs the logical encoding based on the anyon types and enables a universal gate set~\cite{Cui2015QIP_UQC_Weakly}. 
In this section, we review the basic gadgets of the $U$-model and discuss how to implement them in a potentially fault-tolerant manner.

\subsection{Review of the \texorpdfstring{$U$}{U}-model}

Consider the fusion of four $D$ anyons into a single $G$ anyon. This process has a nine-dimensional fusion space $V_G^{DDDD}$ with one complete set of basis states shown by the following fusion tree
\begin{equation} \label{eqn:state_xy_U_model}
    \ket{xy} = 
    \begin{tikzpicture}[baseline = {(current bounding box.center)}, scale = 0.7]
        \draw (0,0) -- (0,0.5);
        \draw (0,0.5) -- (-0.6,1.1);
        \draw (-0.6,1.1) -- (-1,1.5);
        \draw (-0.6,1.1) -- (-0.2,1.5);
        \draw (0,0.5) -- (0.6,1.1);
        \draw (0.6,1.1) -- (1,1.5);
        \draw (0.6,1.1) -- (0.2,1.5);
        \draw (-0.5,0.7) node {\scriptsize$x$};
        \draw (0.5,0.7) node {\scriptsize$y$};
        \draw (0.2,0) node {\scriptsize$G$};
        \draw (-1,1.75) node {\scriptsize$D$};
        \draw (-0.3,1.75) node {\scriptsize$D$};
        \draw (0.3,1.75) node {\scriptsize$D$};
        \draw (1,1.75) node {\scriptsize$D$};
    \end{tikzpicture}\;.
\end{equation}
Physically, we can prepare each basis state by creating the ($D$ and $G$) anyons at different locations on the lattice with the designated total charges $x$ and $y$.
Owing to their spatial separation, the fusion space is robust against any local process and thus can serve as a code space.
Concretely speaking, the information stored in $\ket{xy}$ is only affected by operators with support sizes larger than the spatial separation of that anyons~\cite {Kitaev2003Annals_FT_anyons,Bombin2008PRB_NonAbelian}.
We organize these basis states into two distinct groups, forming two orthogonal subspaces defined as follows
\begin{equation}
\begin{aligned}
U =& \{\ket{AG},\ket{GG},\ket{GA}\} \\
U^\perp = & \{\ket{FC},\ket{CF},\ket{FH},\ket{HF},\ket{CH},\ket{HC}\}
\end{aligned}
\end{equation}
As pointed out in \cite{Cui2015QIP_UQC_Weakly}, we can use the three-dimensional subspace $U$ to encode a logical qutrit and realize a universal gate set, which explains the term ``$U$-model" for this construction.

More specifically, as stated in the Theorem~1 of \cite{Cui2015QIP_UQC_Weakly}, the universal gate set consists of: (1) single-qutrit classical gates, (2) a generalized qutrit Hadamard gate $h = (1/\sqrt{3})\sum_{i,j=1}^3\omega^{ij}\ket{i}\bra{j}$, (3) the measurement to the qutrit $\ket{0}$ and $\ket{0}^\perp = \mathrm{span}\left\{\ket{1},\ket{2}\right\}$ subspaces, (4) and a two-qutrit generalized CNOT gate defined by $\ket{i,j}\mapsto \ket{i,i+j\mathrm{\ mod\ }3}$. We review the necessary fusion tree operations that enable the implementation of this universal gate set.

The single-qutrit gates (1) and (2) can be implemented through braiding of $D$ anyons.
However, braiding may temporarily move the logical state $\ket{xy}$ out of the computational subspace $U$,
necessitating a measurement that projects the system either back into $U$ or onto its orthogonal complement $U^\perp$ in a coherent manner.
To realize the measurement (3) distinguishing between $\ket{0}$ and $\ket{0}^\perp$ within $U$, we require an additional anyon charge measurement that detects whether the left pair of $D$ anyons have a trivial total charge $(A)$ or not $(A' = G)$.
We denote these two additional measurements by $\mathcal{M}_{U} = \{\Pi_U,\Pi_{U^\perp}\}$ and $\mathcal{M}_{A} = \{\Pi_{A},\Pi_{A^\prime}\}$ respectively, where $\Pi_{\bullet}$ represents the projector onto the corresponding subspace.
The protection of the code space stems from the spatial separation between the anyons, and any reliable implementation of these operations must at least preserve this separation throughout the process, which we discuss in \secref{sec:adaptive_movement} and \ref{sec:fault_tolerant_meas}.

The two-qutrit generalized CNOT gate (4) can be constructed by composing the Hadamard gate with a two-qutrit controlled-$Z$ gate.
Specifically, consider two logical qutrits, we need to braid $D$ anyons in the following basis
\begin{equation*}
    \begin{tikzpicture}[baseline = {(current bounding box.center)}, scale = 0.8]
        \draw (0.9,-0.4) -- (0,0.5);
        \draw (0,0.5) -- (-0.4,0.9);
        \draw (-0.4,0.9) -- (-0.6,1.1);
        \draw (-0.4,0.9) -- (-0.2,1.1);
        \draw (0,0.5) -- (0.4,0.9);
        \draw (0.4,0.9) -- (0.6,1.1);
        \draw (0.4,0.9) -- (0.2,1.1);
        \draw (0.2,-0.1) node {\scriptsize $G$};
        \draw (-0.6,1.3) node {\scriptsize $D$};
        \draw (-0.2,1.3) node {\scriptsize $D$};
        \draw (0.2,1.3) node {\scriptsize $D$};
        \draw (0.6,1.3) node {\scriptsize $D$};
        \draw (0.9,-0.9) -- (0.9,-0.4);
        \draw (0.9,-0.4) -- (1.8,0.5);
        \draw (1.8,0.5) -- (1.4,0.9);
        \draw (1.4,0.9) -- (1.2,1.1);
        \draw (1.4,0.9) -- (1.6,1.1);
        \draw (1.8,0.5) -- (2.2,0.9);
        \draw (2.2,0.9) -- (2.0,1.1);
        \draw (2.2,0.9) -- (2.4,1.1);
        \draw (1.6,-0.1) node {\scriptsize $G$};
        \draw (1.2,1.3) node {\scriptsize $D$};
        \draw (1.6,1.3) node {\scriptsize $D$};
        \draw (2.0,1.3) node {\scriptsize $D$};
        \draw (2.4,1.3) node {\scriptsize $D$};
        \draw (1.2, -0.9) node {\scriptsize $G$};
    \end{tikzpicture}\;,
\end{equation*}
where the two logical qutrits have a definite total anyon charge $G$.
Since each logical qutrit is prepared independently, we must be able to merge them into this larger fusion tree.
When there are more than two logical qutrits, we often need to apply the controlled-Z gate to any pair of them.
Therefore we also have to be able to split this larger fusion tree into small ones.
In \secref{sec:multi_qutrit_gate}, we present explicit algorithms for merging two logical qutrits and splitting this larger fusion tree.

\begin{figure*}
    \centering
    \includegraphics[width = \textwidth]{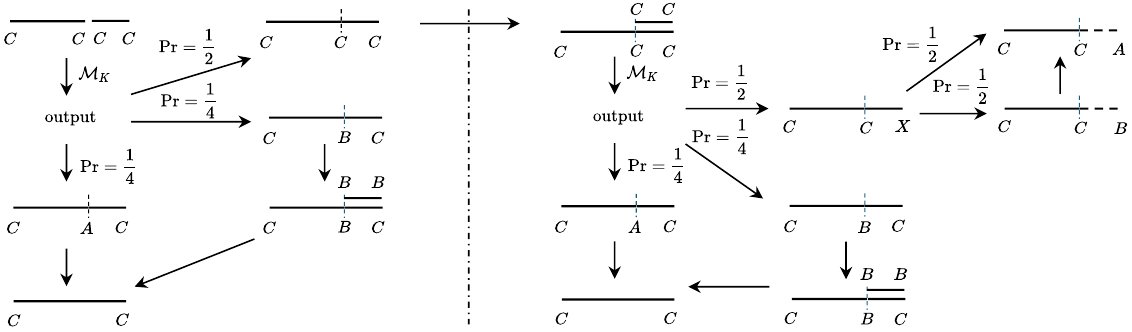}
    \caption{Protocol for moving a type $C$ (non-Abelian) anyon one site to the right.}
    \label{fig:anyon_movement}
\end{figure*}

\subsection{Braiding via adaptively moving anyons}
\label{sec:adaptive_movement}

When braiding two anyons, it is straightforward to maintain their large separation. 
The major source of errors is the process of moving a non-Abelian anyon itself.
Naively, to move an anyon $\alpha$ from the position $x$ to $y$, we first create a pair of $\alpha$ and its anti-particle $\bar{\alpha}$ at $y$ and $x$ respectively, and then fuse the two anyons at $x$ into the vacuum, as is depicted below
$$
\begin{tikzpicture}
\filldraw[black] (2,0) circle (0.05);
\filldraw[black] (2,0.5) circle (0.05);
\filldraw[black] (6,0.5) circle (0.05);
\draw[thick] (1,0) node[left]{$\cdots$} -- (2,0) node[below]{\scriptsize $\alpha$};
\draw[thick] (2,0.5) node[above]{\scriptsize $\bar\alpha$} -- (6,0.5) node[above]{\scriptsize $\alpha$};
\draw[thick] (2,0.25) ellipse (0.35 and 0.8);
\end{tikzpicture}
$$
For a non-Abelian $\alpha$, this procedure faces two obstacles.
First, the fusion process may fail due to multiple possible outcomes.
Quantitatively, the probability of fusing anyons $\alpha$ and $\beta$ into $\gamma$ is given by~\cite{Preskill2004LectureNF}:
\begin{equation} \label{eqn:fusion_prob}
    \mathrm{Pr}(\alpha\times \beta \to \gamma) = \frac{N_{\alpha\beta}^\gamma d_{\gamma}}{d_\alpha d_\beta}\;,
\end{equation}
where $N^{\gamma}_{\alpha\beta}$ is the fusion multiplicity, and $d_{\alpha}$, $d_\beta$, $d_\gamma$ are the anyon quantum dimensions.
In the case of $\mathcal{D}(S_3)$, the relevant quantum dimensions $d_{R,C}$ are listed in Table~\ref{tab:conjugacy_irreps}.
Therefore, this method succeeds only with a finite probability even in the absence of errors.
Second, creating a pair of well-separated non-Abelian anyons requires either a linear-depth unitary circuit or a constant-depth adaptive circuit. 
Consequently, any local error in the early stage of the circuit will propagate to the entire anyon string, rendering the error control impossible.

To suppress the error propagation, we move the non-Abelian anyon incrementally, step by step, rather than over a long distance in a single shot.
Specifically, to move an $\alpha$ anyon by one step, we first create an $\alpha-\bar{\alpha}$ pair using a short anyon string and then fuse $\bar{\alpha}$ with the original $\alpha$ anyon, aiming to produce a vacuum.
The fusion succeeds only with a finite probability, and we must repeat the process adaptively until we get the desired outcome. We will show in concrete examples that the success probability converges to 1 exponentially fast with the number of repetitions.
After finishing each step, we proceed to the next until moving $\alpha$ to its final destination.
In this incremental movement protocol, any local error remains confined to where we apply the operators, which makes the quantum error correction possible, as detailed in \secref{sec:S3_QEC}.

In the context of the quantum double model, we move the non-Abelian anyons by one lattice site in each step.
Namely, we only apply the shortest ribbon operator in the horizontal or vertical direction:
\begin{equation}
\begin{aligned}
    \label{eqn:shortest_F} F^{R,C;u,v}_{\rho_h} &= \frac{|R|}{|Z(C)|} \sum_{n\in Z(C)} \Gamma^R_{jj^\prime}(n) (T^{\tau_c n \bar{\tau}_{c^\prime}}_{+})_1 (L^{c^\prime}_{+})_2\;,\\
    F^{R,C;u,v}_{\rho_v} &= \frac{|R|}{|Z(C)|}\sum_{n\in Z(C)} \Gamma^{R}_{jj^\prime}(n) (L^{c}_{+})_1 (T^{\tau_c n\bar{\tau}_{c^\prime}}_{-})_2\;,
\end{aligned}
\end{equation}
where each operator acts on two edges, labeled $1$ and $2$, whose positions are specified by the corresponding ribbons $\rho_h$ and $\rho_v$ as follows: 
\begin{equation*}
    \rho_h = \ \begin{tikzpicture}[baseline={(current bounding box.center)}, scale = 0.85]
    \foreach \i in {0,1}
        {\foreach \j in {0,1}{
            \draw[thick] (\i,\j+1) -- (\i,\j);
            \draw[thick] (\i,\j+1) -- (\i+1,\j+1);
            \draw[thick] (\i,\j) -- (\i+1,\j);
            \draw[thick] (\i+1,\j+1) -- (\i+ 1,\j);
            }}
    \draw (0.5,0.5)--(1,1);
    \draw (0.5,0.5)-- (1.5,0.5);
    \foreach \i in {0,1}{           
        \draw (\i,1) -- (0.5+\i,0.5);
    }
    \node at (0.2,0.5) {\scriptsize $s$};
    \node at (1.6,0.8) {\scriptsize $s^\prime$};
    \node at (0.8,0.3) {\scriptsize $2$};
    \node at (0.47,1.2) {\scriptsize $1$};
    \end{tikzpicture} \;,\quad
    \rho_v = \ \begin{tikzpicture}[baseline={(current bounding box.center)}, scale = 0.85]
    \foreach \i in {0,1}
        {\foreach \j in {0,1}{
            \draw[thick] (\i,\j+1) -- (\i,\j);
            \draw[thick] (\i,\j+1) -- (\i+1,\j+1);
            \draw[thick] (\i,\j) -- (\i+1,\j);
            \draw[thick] (\i+1,\j+1) -- (\i+ 1,\j);
            }}
    \draw (0,1)--(0.5,1.5);
    \draw (0.5,0.5)-- (0.5,1.5);
    \foreach \i in {0,1}{           
        \draw (0,\i+1) -- (0.5,\i+0.5);
    }
    \node at (0.2,0.5) {\scriptsize $s$};
    \node at (0.55,1.8) {\scriptsize $s^\prime$};
    \node at (0.17,1.51) {\scriptsize $2$};
    \node at (0.7,0.75) {\scriptsize $1$};
    \end{tikzpicture}\;,
\end{equation*}
where the triangle markers indicate whether the operators act on the sides or endpoints of the edges. Following the conventions in \figref{fig:ops_quantum_double}(b), one can directly identify the corresponding $L^{g}_{\pm}$ and $T^{h}_{\pm}$ operators appearing in \eqnref{eqn:shortest_F}. For a more detailed and precise definition, we refer the reader to Appendix~\ref{appendix:ribbon_algebra}.

Compared with the general discussion, one difference is that the probability of the fusion outcome does not always follow \eqnref{eqn:fusion_prob}, but depends on the local degrees of freedom.~\footnote{An analogous example is the fusion of two spin-$1/2$ particles. If both spins are in the $\ket{0}$ (spin up) state, the fusion outcome must be a spin-$1$.}
In principle, one can leverage such a dependence to optimize the fusion probability by choosing the proper internal states. However, the local degrees of freedom of non-Abelian anyons are susceptible to local errors, making it impractical to track them in the presence of noise. 
Here, we only use the ribbon operators creating non-Abelian anyons with the maximally mixed local degrees of freedom. 
Let $F^{R,C}_{\rho_h}$ and $F^{R,C}_{\rho_v}$ denote such ribbon operators. Their action on a density matrix $\rho$ is given by the following quantum channel
\begin{equation} 
\label{eqn:maximally_mixed_F_RC}
    F^{R,C}_{\rho_{h/v}}[\rho] = \frac{|G|^2}{|R|^3 |C|^3} \sum_{u,v}F^{R,C;u,v}_{\rho_{h/v}}\rho \left(F^{R,C;u,v}_{\rho_{h/v}}\right)^\dagger\;.
\end{equation}
For such two independent anyons, the probability of different fusion outcomes does follow the general result \eqnref{eqn:fusion_prob} so that we can systematically analyze the success probability.
In practice, if the local noise changes local degrees of freedom, the probability will be modified though we expect our analysis still holds qualitatively in the generic case.

As an illustrative example, let us consider a long type-$C$ ribbon and move the $C$ anyon located at its right end to one site away, as shown in Figure~\ref{fig:anyon_movement}. 
The first step is to apply the shortest horizontal type $C$ ribbon operator and apply a $\mathcal{M}_K$ measurement to fuse the two intermediate $C$ anyons.
In this case, there are three possible fusion outcomes occurring with the following probabilities
\begin{equation} \label{eqn:fusion_prob_C}
    \mathrm{Pr}(C\times C \to X)= \begin{cases}
        1/4,& \mathrm{for}\ X= A,B\\
        1/2, & \mathrm{for}\ X=C
    \end{cases}\;.
\end{equation}
Depending on the outcome, the remaining procedure consists of two subroutines.

If the fusion outcome is the vacuum $A$, we have succeeded. 
If the outcome is the Abelian anyon $B$, we can remove the middle $B$ anyon by applying the shortest horizontal type-$B$ ribbon according to the fusion rules $B\times B = A$ and $B\times C = C$. These are illustrated in the lower part of the left panel.

If the fusion outcome is $C$, we need to proceed to the second subroutine (right panel of \figref{fig:anyon_movement}).  
First, we apply the shortest horizontal $C$ ribbon to try to annihilate the intermediate $C$ anyon.
If the left pair of $C$ anyons fuses to $A$ or $B$, we are back to the scenarios with outcomes $A$ or $B$ in the first subroutine, allowing us to repeat the same procedure as stated earlier. 
If the left $C$ anyon pair fuses to $C$, the right pair of $C$ anyons has two possible outcomes (obtained by applying $F$-move, see Appendix~\ref{appendix:UQC}), indicated by the $X$ anyon. The outcomes may be $X=A$ or $X=B$, each with probability $1/4$, so we return to the start of the first subroutine and restart the entire process.

Non-Abelian anyons of types $F$, $G$, and $H$ have the same fusion rules as $C$, specifically $X\times X = A+B+X$, for $X=F,G$ and $H$. Therefore, we can apply the same protocol to move any of them. 
The protocol for moving $D$ and $E$ anyons is different and is detailed in Appendix~\ref{appendix:movement_D_E}.
By mathematical induction, one can prove that the probability of successfully moving the anyon after $n$ rounds of $\mathcal{M}_K$ measurement is:
\begin{equation*} 
    \mathrm{Pr}(\mathrm{moving}\ X) = \begin{cases}
    1\;, & X=A,B\\
        1-\left(\frac{1}{2}\right)^n\;,& X=C,F,G,H\\
        1-\frac{8}{9}\left(\frac{1}{2}\right)^{n-1}\;, & X=D,E
    \end{cases}\;.
\end{equation*}
In other words, the success probability of this adaptive protocol approaches $1$ exponentially as $n$ increases.

\subsection{Measurements via anyon interferometry}
\label{sec:fault_tolerant_meas}

\begin{figure*}[tb]
    \centering
    \includegraphics[width = \textwidth]{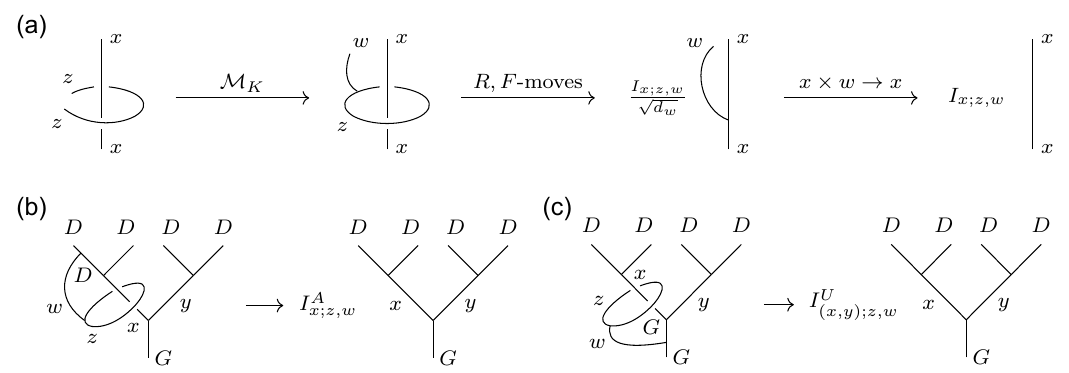}
    \caption{The remote measurement protocol using non-Abelian anyon interferometer. (a) The process of measuring an $x$ anyon: A pair of $z$ anyons is created, with one circling around $x$. The two $z$ anyons are then fused into a $w$ anyon, which is subsequently fused with $x$. The amplitude $I_{x;z,w}$ associated with this process is derived using category theory. The interferometric setups for the measurements $\mathcal{M}_A$ (b) and $\tilde{\mathcal{M}}_U$ (c), with amplitudes $I^{A}_{x;z,w}$ and $I^U_{(x,y);z,w}$, respectively.}
    \label{fig:FT_measurements}
\end{figure*}

A typical method for measuring the total anyon charge is fusion. 
However, this operation requires bringing anyons close to each other, which reduces their separation and the robustness of the code space. 
Therefore, we must realize the two measurements $\mathcal{M}_U$ and $\mathcal{M}_A$ in a remote manner, without using this naive fusion. 
In this section, we explain how to achieve this remote measurement via anyon interferometry.

We illustrate the basic idea by considering the process of measuring a single $x$ anyon (see \figref{fig:FT_measurements}(a)).
We create a pair of $z$ anyons from the vacuum and braid one of them around the $x$ anyon. We then fuse the two $z$ anyons and get the fusion outcome, e.g., a $w$ anyon, with a probability $|I_{x;z,w}|^2$ determined via a sequence of $R$-moves and $F$-moves (see Appendix~\ref{appendix:Anyon_Interferometry}). 
Next, we fuse the $w$ anyon with the $x$ anyon to return to the original fusion tree. The net effect of these operations is a diagonal gate on the anyonic state.
Moreover, if the measured state is in a linear superposition of anyonic states $\ket{\psi} = \sum_{x} c_x \ket{x}$, braiding and fusing the $z$ anyon yields a $w$ with the probability $\sum_x |c_x I_{x;z,w}|^2$ and the final state (after fusing $w$ with the original anyon)
\begin{equation} \label{eqn:measure_state_psi}
    \mathcal{M}_z (\ket{\psi}) \xrightarrow{w} \frac{1}{\sqrt{\sum_{x} |c_x I_{x;z,w}|^2}} \sum_x c_x I_{x;z,w} \ket{x}\;.
\end{equation}
Generally, we need to repeat this process sufficiently many times in order to read off the type of $x$ anyon from the probability distribution of the fusion outcome. 
For our specific purpose, this is not necessary as we can properly choose $z$ anyons so that a single shot fusion outcome provides us with enough information. 
As is shown below, we can adaptively realize the two measurements with an arbitrarily high accuracy.

This process is robust against various undesired local braiding and fusion. 
For instance, if a local error braids the two $z$ anyons when they are close to each other, the result is a trivial global phase $R^{zz}_{w}$.
Another scenario is where a pair of $a$ anyons is locally created, braided with one of the $z$ anyons, and fused into $b$. 
Then, the pair of $z$ anyons is fused into a $c$ anyon, while the $b$ and $c$ anyons are finally fused into a $w$ anyon. After removing the intermediate anyon loops, we have
\begin{equation} \label{eqn:local_interferometry}
    \begin{tikzpicture}[baseline = {(current bounding box.center)}, scale = 0.9]
        \draw (0,0) -- (0,0.21); 
        \draw (0,0.36) -- (0,2.0); 
        \draw (-0.5,0.88).. controls (-0.75,0.83) and (-0.9,0.6) .. (-0.7,0.45) ..controls (-0.4,0.22) and (0.4, 0.22) .. (0.7,0.45) .. controls (0.9,0.6) and (0.75,0.9) .. (0.08,0.95);
        \draw (-0.08,0.95) .. controls (-0.15,0.95) and (-0.3,0.94) .. (-0.38,0.91);
        \draw (-0.7,0.86) .. controls (-0.65,1.22) and (-0.45,1.05) .. (-0.43,0.85) .. controls (-0.45, 0.55) and (-0.62,0.63).. (-0.65,0.73);
        \draw (-0.59,1.09) .. controls (-0.6,1.20) and (-0.65,1.36) .. (-0.8,1.45);
        \draw (-0.80,0.54) .. controls (-1.1,0.9) and (-1,1.3) .. (-0.8,1.45);
        \draw (-0.8,1.45) .. controls (-0.8,1.60) and (-0.75,1.8) .. (-0.7,1.9);
        \node at (0.2,0) {\scriptsize \(x\)};
        \node at (0.2,2.0) {\scriptsize \(x\)};
        \node at (-0.5,0.15) {\scriptsize \(z\)};
        \node at (-0.32,0.60) {\scriptsize \(a\)};
        \node at (-0.47,1.27) {\scriptsize \(b\)};
        \node at (-1.13,1.1) {\scriptsize \(c\)};
        \node at (-0.5,2) {\scriptsize \(w\)};
    \end{tikzpicture} = I_{z,(a,b,c),w}\begin{tikzpicture}[baseline = {(current bounding box.center)}, scale = 0.9]
        \draw (0,0) -- (0,0.21); 
        \draw (0,0.36) -- (0,2.0); 
        \draw (-0.08,0.95).. controls (-0.75,0.9) and (-0.9,0.6) .. (-0.7,0.45) ..controls (-0.4,0.22) and (0.4, 0.22) .. (0.7,0.45) .. controls (0.9,0.6) and (0.75,0.9) .. (0.08,0.95);
        \draw (-0.65,0.81) .. controls (-0.9,1.00) and (-0.95,1.6) .. (-0.7,1.9);
        \node at (0.2,0) {\scriptsize\(x\)};
        \node at (0.2,2.0) {\scriptsize\(x\)};
        \node at (-0.5,0.15) {\scriptsize\(z\)};
        \node at (-0.5,2) {\scriptsize\(w\)};
    \end{tikzpicture}\;,
\end{equation}
where $I_{z,(a,b,c),w}$ is a global phase independent of $x$ (see Appendix~\ref{appendix:Phase_Factor_Interferometry}). This implies that the process returns to the original one as shown in \figref{fig:FT_measurements}(a), and the undesired local braiding and fusion do not induce logical errors.

To implement $\mathcal{M}_A = \{\Pi_A,\Pi_{A^\prime}\}$ in the computational subspace, we can apply the anyon interferometry to the left pair of $D$ anyons followed by fusing the $w$ anyon (the fusion outcome of $z$) with the first $D$ anyon of the fusion tree (see \figref{fig:FT_measurements}(b)).
By selecting $z=D$, and through an iterative procedure as described in Appendix~\ref{appendix:FT_Measurements}, the protocol reliably returns the fusion tree to $V^{DDDD}_G$. Notably, one can show that the internal state $x$ must be the same as the fusion outcome $w$, allowing for the remote implementation of $\mathcal{M}_A$.

To implement $\mathcal{M}_U = \{\Pi_U,\Pi_{U^\perp}\}$, we need an intermediate anyon interferometry $\tilde{\mathcal{M}}_U$, as depicted in \figref{fig:FT_measurements}(c). 
Here, the fusion outcome $w$ of the $z$ anyons is fused with the $G$ anyon of the fusion tree, yielding an amplitude $I_{(x,y);z,w}^U$ that depends on the state $\ket{xy}$ (see Appendix~\ref{appendix:FT_Measurements}). 
Importantly, the measurement must preserve coherence within the $U$ and $U^\perp$ subspaces, i.e., the amplitudes of the basis states in each subspace must remain unchanged after the measurement.
It requires us to choose $z = H$ for the intermediate anyon interferometry $\tilde{\mathcal{M}}_U$. 
Specifically, the amplitudes $I^{U}_{(x,y);H,A}$ of the measurement result $w = A$ for different computational basis states $\ket{xy}$ are:
\begin{equation} \label{eqn:I_U_H_A}
I^{U}_{(x,y);H,A} = \begin{cases}
     1\;, & (x,y) \in U\\
     -\frac{1}{2}\;, & (x,y) \in U^\perp
\end{cases}\;,
\end{equation}
where $(x,y) \in U$ means that the basis state $\ket{xy}$ lies within the subspace $U$. 
For the measurement outcome $w=B$, we have
\begin{equation} \label{eqn:I_U_H_B}
    I^{U}_{(x,y);H,B} = \begin{cases} 0\;, & (x,y) \in U\\
        -\frac{i\sqrt{3}}{2}\;, & (x,y) \in U^\perp_1\\
        \frac{i\sqrt{3}}{2}\;, & (x,y) \in U^\perp_2
    \end{cases}\;,
\end{equation}
where $U^\perp_1 = \mathrm{span}\{\ket{FC}, \ket{HF}, \ket{CH}\}$ and $U^\perp_2 = \mathrm{span}\{\ket{CF}, \ket{FH}, \ket{HC}\}$.
Note that the amplitudes for the states in $U_1^\perp$ and $U_2^\perp$ differ by a sign.
Consequently, if the outcome $w=B$ is obtained, we can repeat $\tilde{\mathcal{M}}_U$ until a subsequent $w=B$ result cancels the relative sign, thereby preserving the coherence within $U^\perp$.

Building upon $\tilde{\mathcal{M}}_U$, we can realize $\mathcal{M}_U$ by applying $\tilde{\mathcal{M}}_U$ repeatedly with an arbitrarily high accuracy.
Consider a general state 
\begin{equation*}
\begin{aligned}
\ket{\psi_0} & = \sum_{(x,y)\in U} \alpha_{xy}\ket{xy} \\
& + \sum_{(x,y)\in U^\perp_1} \beta_{xy}\ket{xy} + \sum_{(x,y)\in U^\perp_2} \gamma_{xy}\ket{xy}\;.
\end{aligned}
\end{equation*}
After performing $\tilde{\mathcal{M}}_U$, we ignore the overall normalization factor and write the resulting state as
\begin{equation*}
    \begin{cases}
        \sum \alpha_{xy}\ket{xy} - \frac{1}{2}\left(\sum \beta_{xy}\ket{xy} + \sum \gamma_{xy}\ket{xy} \right)\;, & w = A\\
        -\sum\beta_{xy}\ket{xy}+\sum\gamma_{xy} \ket{xy}\;,& w = B 
    \end{cases}\;
\end{equation*}
each of which occurs with the probability
\begin{equation*}
    \begin{cases}
        \sum|\alpha_{xy}|^2+\frac{1}{4}\sum|\beta_{xy}|^2+\frac{1}{4}\sum|\gamma_{xy}|^2\;,&w=A\\
        \frac{3}{4}\sum|\beta_{xy}|^2+\frac{3}{4}\sum|\gamma_{xy}|^2\;,&w=B
    \end{cases}
\end{equation*}
Here the summations are taken over the corresponding subspaces. 
After $n$ rounds of $\tilde{\mathcal{M}}_U$ measurements, if all the measurement outcomes are $w=A$, the state reads
\begin{equation*}
    \sum \alpha_{xy} \ket{xy} + \left(-\frac{1}{2}\right)^n\left( \sum \beta_{xy} \ket{xy} + \sum\gamma_{xy} \ket{xy} \right)\;,
\end{equation*}
which is exponentially close to $\Pi_U \ket{\psi_0}$, effectively realizing $\Pi_U$. 
This event happens with a probability exponentially close to $\sum |\alpha_{xy}|^2$, as what we expect for projecting $\ket{\psi_0}$ onto $U$.

On the other hand, a single occurrence of $w=B$ projects the state onto $U^\perp$ but with a relative sign between $U^\perp_1$ and $U^\perp_2$. 
Subsequently, we continue performing $\tilde{\mathcal{M}}_U$ until another $w=B$ outcome is obtained, thereby canceling out the relative minus sign. 
For a state in $U^\perp$, each $\tilde{\mathcal{M}}_U$ has a $\frac{1}{4}$ or $\frac{3}{4}$ probability of yielding $w=A$ or $w=B$, respectively. 
The probability of eventually obtaining another $w=B$ result is exponentially close to $1$ with respect to the number of measurement rounds.

We close this subsection with a remark on the potential alternative approach of using a large closed ribbon operator $K^{R,C}_\sigma$ for the measurement, where the support $\sigma$ forms a large closed ribbon enclosing the measured anyons.
Using the anyon interferometry or $K^{R,C}_\sigma$ is based on the same physics.
Their main difference is in their noise resilience properties.
It follows from \eqref{eqn:K_RC} and the entangling structure of $F^{h,g}_{\rho}$ that any error in the early stages of applying $K^{R,C}_\sigma$ propagates through the entire ribbon, making it infeasible in practice.

\subsection{Merging independent logical qutrits}
\label{sec:multi_qutrit_gate}

\begin{figure*}[tb]
    \centering
    \includegraphics[width = \textwidth]{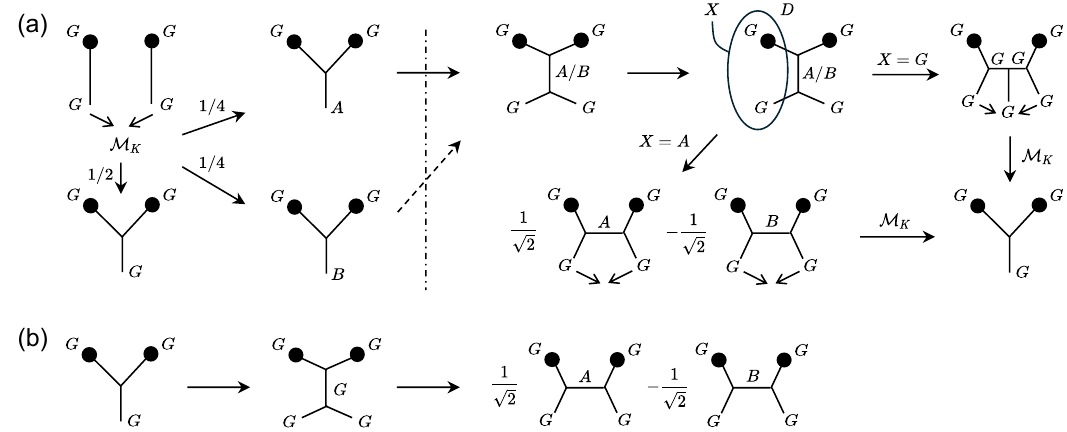}
    \caption{(a) Protocol for fusing the $G$ anyons from two fusion trees, both in $V^{DDDD}_G$, into a single $G$ anyon. (b) Protocol for splitting the $G$ anyon, in the two-qutrit fusion tree, into two $G$ anyons. This splitting process allows for the transition into two fusion trees in $V^{DDDD}_G$.}
    \label{fig:fusion_G_anyons}
\end{figure*}

In addition to the single-qutrit operations, implementing the two-qutrit controlled-$Z$ is essential for achieving a universal gate set. This gate is realized through specific braiding operations in the two-qutrit fusion tree basis:
\begin{equation}\label{eqn:two_qutrit_basis}
    \begin{tikzpicture}[baseline = {(current bounding box.center)}]
        \draw (0.9,-0.4) -- (0,0.5);
        \draw (0,0.5) -- (-0.4,0.9);
        \draw (-0.4,0.9) -- (-0.6,1.1);
        \draw (-0.4,0.9) -- (-0.2,1.1);
        \draw (0,0.5) -- (0.4,0.9);
        \draw (0.4,0.9) -- (0.6,1.1);
        \draw (0.4,0.9) -- (0.2,1.1);
        \draw (-0.35,0.6) node {\scriptsize $x_1$};
        \draw (0.44,0.6) node {\scriptsize $y_1$};
        \draw (0.2,-0.1) node {\scriptsize $G$};
        \draw (-0.6,1.3) node {\scriptsize $D$};
        \draw (-0.2,1.3) node {\scriptsize $D$};
        \draw (0.2,1.3) node {\scriptsize $D$};
        \draw (0.6,1.3) node {\scriptsize $D$};
        \draw (0.9,-0.9) -- (0.9,-0.4);
        \draw (0.9,-0.4) -- (1.8,0.5);
        \draw (1.8,0.5) -- (1.4,0.9);
        \draw (1.4,0.9) -- (1.2,1.1);
        \draw (1.4,0.9) -- (1.6,1.1);
        \draw (1.8,0.5) -- (2.2,0.9);
        \draw (2.2,0.9) -- (2.0,1.1);
        \draw (2.2,0.9) -- (2.4,1.1);
        \draw (1.45,0.6) node {\scriptsize $x_2$};
        \draw (2.24,0.6) node {\scriptsize $y_2$};
        \draw (1.6,-0.1) node {\scriptsize $G$};
        \draw (1.2,1.3) node {\scriptsize $D$};
        \draw (1.6,1.3) node {\scriptsize $D$};
        \draw (2.0,1.3) node {\scriptsize $D$};
        \draw (2.4,1.3) node {\scriptsize $D$};
        \draw (1.2, -0.9) node {\scriptsize $G$};
    \end{tikzpicture}\;,
\end{equation}
where both $\ket{x_1y_1}$ and $\ket{x_2y_2}$ are elements of $V^{DDDD}_G$. 
Since each logical qutrit is prepared independently, we must be able to merge them into this large fusion tree and also split them.
All operations in this section will not affect the internal states, so we represent each individual logical qutrit as
\begin{equation}
    \begin{tikzpicture}[baseline = {(current bounding box.center)}, scale=0.95]
        \filldraw (0,1.2) circle (2pt);
        \draw (0,0) -- (0,1.2);
        \draw (0.2,0) node {\scriptsize $G$};
        \draw (0.3,1.2) node {\scriptsize $G$};
    \end{tikzpicture}
    \equiv
    \begin{tikzpicture}[baseline = {(current bounding box.center)}, scale = 0.8]
        \draw (0,0) -- (0,0.5);
        \draw (0,0.5) -- (-0.6,1.1);
        \draw (-0.6,1.1) -- (-1,1.5);
        \draw (-0.6,1.1) -- (-0.2,1.5);
        \draw (0,0.5) -- (0.6,1.1);
        \draw (0.6,1.1) -- (1,1.5);
        \draw (0.6,1.1) -- (0.2,1.5);
        \draw (-0.5,0.7) node {\scriptsize $x$};
        \draw (0.5,0.7) node {\scriptsize $y$};
        \draw (0.2,0) node {\scriptsize $G$};
        \draw (-1,1.75) node {\scriptsize $D$};
        \draw (-0.3,1.75) node {\scriptsize $D$};
        \draw (0.3,1.75) node {\scriptsize $D$};
        \draw (1,1.75) node {\scriptsize $D$};
    \end{tikzpicture}\;.
\end{equation}
Below we describe a protocol for fusing and splitting the $G$ anyons at the roots
\begin{equation}
    \begin{tikzpicture}[baseline = {(current bounding box.center)}, scale=0.95]
        \filldraw (-0.8,1) circle (2pt);
        \draw (-0.8,0) -- (-0.8,1);
        \draw (-1,0) node {\scriptsize $G$};
        \draw (-1.1,1) node {\scriptsize $G$};
        \filldraw (0,1) circle (2pt);
        \draw (0,0) -- (0,1);
        \draw (0.2,0) node {\scriptsize $G$};
        \draw (0.3,1) node {\scriptsize $G$};
    \end{tikzpicture}\  \xleftrightarrow{\quad}\   \begin{tikzpicture}[baseline = {(current bounding box.center)}, scale=0.95]
        \filldraw (-1,1) circle (2pt); 
        \filldraw (0,1) circle (2pt);
        \draw (-1,1) -- (-0.5,0.5);
        \draw (0,1) -- (-0.5,0.5);
        \draw (-0.5,0.5) -- (-0.5,0);
        \draw (-1.3,1) node {\scriptsize $G$};
        \draw (0.3,1) node {\scriptsize $G$};
        \draw (-0.3,0) node {\scriptsize $G$};
    \end{tikzpicture}\;,
\end{equation}
which facilitates the implementation of the controlled-$Z$ gate.

The fusion protocol, as shown in \figref{fig:fusion_G_anyons}(a), has two subroutines that are separated by the vertical dashed line. In the first subroutine (left panel), we begin by bringing the two $G$ anyons together at a single site and performing a $\mathcal{M}_K$ measurement. This results in a fusion outcome of $A, B$ or $G$, with probabilities given by \eqref{eqn:fusion_prob}. If the outcome is $G$, as shown in the lower part of the first subroutine, the protocol is complete. If the outcome is $A$ or $B$, we proceed to the second subroutine (right panel). 

In the second subroutine, the $A/B$ anyon is first split into two $G$ anyons by applying a shortest $G$ ribbon operator with one end at the $A/B$ anyon, followed by a $\mathcal{M}_K$ measurement. According to the fusion rule $A/B \times G \to G$, this process annihilates the $A/B$ anyon and creates two $G$ anyons. The two $G$ anyons are then separated far apart using the protocol described in \secref{sec:adaptive_movement}. 

Next, as shown in the upper middle part of the second subroutine, an anyon interferometer is performed on the quasi-$G$ anyon and the $G$ anyon located on the left part of the configuration, using a pair of $D$ anyons. This interferometer yields a fusion outcome $X$, which can be $A$ or $G$. 

If $X=A$, we obtain a linear superposition of fusion trees with four (quasi) $G$ anyons and an internal $A$ or $B$ anyon, as depicted in the lower left part of the right panel. By fusing the lower two $G$ anyons via $\mathcal{M}_K$, we obtain a single $G$ anyon, completing the process. If $X=G$, in addition to the two quasi-$G$ anyons, the configuration includes three $G$ anyons, as shown in the right part of the second subroutine. Fusing these three $G$ anyons will also result in a single $G$ anyon, thereby achieving the desired outcome. A detailed analysis of these fusion processes and their associated probabilities is provided in Appendix~\ref{appendix:multi_qutrit_gate}.

The splitting protocol, as illustrated in Figure~\ref{fig:fusion_G_anyons}(b), begins with splitting the $G$ anyon at the root of the two-qutrit fusion tree into two $G$ anyons. This is achieved by iteratively applying the shortest $G$ ribbon, with one end anchored to the $G$ anyon, followed by a $\mathcal{M}_K$ measurement. According to the fusion rule $G\times G \to A+B+G$, the splitting process succeeds with a probability that exponentially approaches $1$, with respect to the number of iterations of the shortest $G$ ribbon application and the $\mathcal{M}_K$ measurement. Once the two $G$ anyons are generated, they are separated using the protocol described in \secref{sec:adaptive_movement}.
We can apply an $F$-move to view it as a linear superposition of fusion trees involving four (quasi) $G$ anyons with an internal $A$ or $B$ anyon (see Appendix~\ref{appendix:multi_qutrit_gate}). 
As demonstrated in Appendix~\ref{appendix:multi_qutrit_gate}, since $B$ is an Abelian anyon and indistinguishable from $A$ in all steps necessary for realizing the universal gate set, the splitting protocol is deemed successful.

All operations in these protocols take place far from the single-qutrit fusion trees, which ensures robustness against local errors.
We can also apply them to any pair of qutrits in $V^{DDDD}_G$ by swapping fusion trees to position target qutrits adjacently, followed by the fusion and splitting steps. Consequently, the controlled-$Z$ gate can be applied to any pair of qutrits. Combined with the implementation of single-qutrit operations, we can achieve a multi-qutrit universal gate set in a remote manner.

\section{Circuits for anyon manipulation}
\label{sec:quantum_double_circuit}

It is challenging to directly work with a six-dimensional local Hilbert space and realize the quantum double algebra on the current experimental platforms.
A more accessible approach is to replace each local Hilbert space with a pair of qutrit and qubit.
In this realistic setup, it has been shown how to prepare the ground state of $\mathcal{D}(S_3)$ quantum double efficiently~\cite{verresen2022efficiently,Tantivasadakarn2023PRXQuantum_Hierarchy,bravyi2022adaptive}.  
Here we focus on the explicit implementation of the shortest ribbon operator $F^{R,C}_{\rho_{h/v}}$ and the local anyon charge measurement $\mathcal{M}_K$ using qutrits and qubits, and thus all the basic operations for the $U$-model at the circuit level. The resulting circuit architecture is compatible with several existing experimental platforms, as we discuss at the end of this section.

Let us specify our notations of qubits and qutrits for the rest of the paper. Symbols for qubits are the standard ones, such as the basis states $\ket{l}$, $l=0,1$ and the Pauli's $\mathcal{P} = \{I, X, Z, XZ\}$. 
Symbols for qutrits are written with a hat, e.g., the basis states $\ket{\hat{k}}$, $k = 0,1,2$ and the qutrit Pauli's $\hat{\mathcal{P}}=\{\hat{I}, \hat{X}, \hat{Z}, \ldots\}$. We use $\ket{\hat{k}_+}$ as the eigenstate of $\hat{X}$ with eigenvalue $e^{2\pi i k/3}$.

\subsection{The realization of the \texorpdfstring{$L$}{L} and \texorpdfstring{$T$}{T} operators}

Our first step is to explicitly construct circuits for the basic units of all quantum double operators---the $L$ and $T$ operators, whose definitions are copied below for the reader's convenience
\begin{equation*} 
\begin{aligned}
    & L^g_+ \ket{m} = \ket{gm}\;, \quad L^g_- \ket{m} = \ket{m\bar{g}}\;,\\
    & T^h_+\ket{m} = \delta_{h,m}\ket{m}\;,\quad T^h_-\ket{m} = \delta_{\bar{h},m}\ket{m}\;.
\end{aligned}
\end{equation*}
Essentially, $L^g$ follows the group multiplication rule, and $T^h$ is the projection onto the element $h$.
Their representations manifest in the qutrit-qubit basis by noticing the group isomorphism
\begin{equation}
\begin{gathered}
    S_3 \cong \mathbb{Z}_3 \rtimes \mathbb{Z}_2 \\
    g \in S_3 \mapsto \mu^k \sigma^l \in \mathbb{Z}_3 \rtimes \mathbb{Z}_2
\end{gathered}
\end{equation}
Accordingly, we encode the group basis $\{\ket{g}: g \in S_3\}$ into a qutrit and a qubit by
\begin{equation}
    \ket{g} \in \calH_{S_3} \mapsto \ket{\hat{k},l} \in \mathbb{C}^3 \otimes \mathbb{C}^2\,.
\end{equation}
This allows us to write down $L^{\mu,\sigma}$ and $T^{\mu,\sigma}$ in the qutrit-qubit basis and generate all $L^g_{\pm}$ via the group isomorphism.

Specifically, one can show that $L^{\mu}$ and $L^\sigma$ are given by the following operators
\begin{equation} \label{eqn:L_circuits}
    \begin{aligned}
        &L^{\mu}_+ = \hat{X} \otimes I,\quad L^{\mu}_- = \hat{X}^{-Z}\;,\\
        &L^{\sigma}_+ = \hat{C}\otimes X\;,\quad L^{\sigma}_- = \hat{I}\otimes X\;,
    \end{aligned}
\end{equation}
where $\hat{C}$ is the qutrit charge conjugation gate acting as $\hat{C}\ket{\hat{k}} = \ket{-\hat{k}}$, and $\hat{X}^{-Z}$ is a qubit-to-qutrit controlled gate which acts as $\hat{X}^\dagger$ or $\hat{X}$ on the qutrit when the qubit is $\ket{0}$ or $\ket{1}$. 
Introducing the controlled gate is really inevitable as it arises from the nontrivial algebra $\sigma \mu \sigma =\bar{\mu}$ between the generators $\sigma$ and $\mu$ of $\mathbb{Z}_2$ and $\mathbb{Z}_3$ in $S_3$.
To understand this, let us verify the circuit for $L^{\mu}_-$, for which we have 
\begin{equation}
    \hat{X}^{-Z}\ket{\hat{k},l} = \ket{\hat{k} - (-1)^l,l}\;.
\end{equation}
On the other hand, we have
\begin{equation}
    L^{\mu}_- \ket{g} = \ket{g\bar{\mu}} = \ket{\mu^{k}\sigma^l \mu^{-1}}\;.
\end{equation}
By the identity $\sigma \mu \sigma = \bar{\mu}$, we reach the equality between them through the group isomorphism. 
Note that $\hat{X}^{-Z}$ is a non-Clifford gate and can be converted to another non-Clifford gate introduced below.

Similarly, using the fact that operators $T^h_+$ and $T^{h}_-$ are projectors onto $h = \mu^k \sigma^l$ and $\bar{h} = \mu^{-(-1)^l k } \sigma^l$, we have
\begin{equation} \label{eqn:T_circuits}
    T^{h}_+ = P^{k}_{\hat{Z}} P^{l}_{Z}\;,\quad  T^{h}_- = P^{k}_{\hat{Z}^{-Z}} P^{l}_Z\;,
\end{equation}
where $P^{k}_{\hat{O}}$ is the projector onto the eigenstate of the operator $\hat{O}$ with eigenvalue $\omega^{k}$, and $P^{l}_{O}$ is the projector onto the eigenstate  of the operator $O$ with eigenvalue $(-1)^l$.
We can see that $T^{h}_+$ is a measurement in the computational basis.
To recast $T^{h}_-$ in the computational basis in a more transparent fashion, we introduce the qubit-to-qutrit controlled charge conjugation gate $C\hat{C}$~\cite{verresen2022efficiently}
\begin{equation}
    C\hat{C} \ket{\hat{k},l} =  \ket{(-1)^l\hat{k},l}\;.
\end{equation}
It is a non-Clifford gate and has the following algebra with Pauli operators:
\begin{equation} \label{eqn:algebra_CC}
\begin{aligned}
    &C\hat{C} (\hat{X}\otimes I) C\hat{C} = \hat{X}^Z\;,\quad C\hat{C} (\hat{Z}\otimes I) C\hat{C} = \hat{Z}^Z\;,\\
    &C\hat{C}(\hat{I} \otimes X) C\hat{C} = \hat{C}\otimes X\;,\quad C\hat{C}(\hat{I} \otimes Z) C\hat{C} = \hat{I}\otimes Z\;.
\end{aligned}
\end{equation}
Thus, we can obtain $T^{h}_-$ by applying $(\hat{C}\otimes I)C\hat{C}$ followed by a standard computational basis measurement.

In summary, we can construct all $L$ and $T$ operators using Pauli operators, $\hat{C}$, $C\hat{C}$ and the computational basis measurements, where $C\hat{C}$ is the only necessary non-Clifford gate.
Intuitively, at least one non-Clifford gate should be present in the circuits for universal computation. 
On the other hand, having only one non-Clifford gate simplifies the task of suppressing the effective physical error rate, as we will discuss in \secref{sec:ECC}.

\subsection{The realization of \texorpdfstring{$F^{R,C}_{\rho_{h/v}}$}{FRCphv} and \texorpdfstring{$K^{R,C}_s$}{MK}}

As discussed in \secref{sec:adaptive_movement}, we want to apply ribbon operators that create anyons with maximally mixed local degrees of freedom.
This requires us to introduce extra ancillas and entangle them with the local degrees of freedom, followed by discarding the ancillas. 
As an illustrative example, we show the circuit for the horizontal ribbon that creates the chargeons, i.e. $F^{R,C}$ with $C = C_1 = \{e\}$, and refer the reader to Appendices~\ref{appendix:circuit_F_RCrhoh} and \ref{appendix:circuit_F_RCrhov} for the circuit for other anyons.

Recall that the shortest horizontal ribbon written in terms of the $L$ and $T$ operators reads
\begin{equation*}
\begin{gathered}
F^{R,C;u,v}_{\rho_h} = \frac{|R|}{|Z(C)|} \sum_{n\in Z(C)} \Gamma^R_{jj^\prime}(n) (T^{\tau_c n \bar{\tau}_{c^\prime}}_{+})_1 (L^{c^\prime}_{+})_2 \,,\\
\rho_h = \ \begin{tikzpicture}[baseline={(current bounding box.center)}, scale = 0.85]
    \foreach \i in {0,1}
        {\foreach \j in {0,1}{
            \draw[thick] (\i,\j+1) -- (\i,\j);
            \draw[thick] (\i,\j+1) -- (\i+1,\j+1);
            \draw[thick] (\i,\j) -- (\i+1,\j);
            \draw[thick] (\i+1,\j+1) -- (\i+ 1,\j);
            }}
    \draw (0.5,0.5)--(1,1);
    \draw (0.5,0.5)-- (1.5,0.5);
    \foreach \i in {0,1}{           
        \draw (\i,1) -- (0.5+\i,0.5);
    }
    \node at (0.2,0.5) {\scriptsize $s$};
    \node at (1.6,0.8) {\scriptsize $s^\prime$};
    \node at (0.8,0.3) {\scriptsize $2$};
    \node at (0.47,1.2) {\scriptsize $1$};
    \end{tikzpicture} \,.
\end{gathered}
\end{equation*}
For $C= C_1 = \{e\}$,  $L^{c^\prime}_{+}$ becomes the identity operator and the operator only acts nontrivially on the edge 1.
If $R$ is the one-dimensional irreps $[+]$ or $[-]$, the ribbon creates the trivial anyon or $B = ([-],C_1)$. We have
\begin{equation}
\begin{aligned}
    F^{[+],C_1}_{\rho_h} &= (\hat{I}\otimes I)_1 \otimes (\hat{I} \otimes I)_2\;,\\
    F^{[-],C_1}_{\rho_h} &= (\hat{I}\otimes Z)_1 \otimes (\hat{I}\otimes I)_2\;.
\end{aligned}
\end{equation}
When $R$ is the two-dimensional irrep $[2]$, the ribbon creates the non-Abelian anyon $C$ and reads
\begin{equation} \label{eqn:F_2_C1_uv_circuit}
    F^{[2],C_1;u,v}_{\rho_h} = (\hat{Z}^{u+1} \otimes \ket{u+v}\bra{u+v})_1\otimes (\hat{I}\otimes I)_2\;.
\end{equation}
The $F^{[2],C_1;u,v}_{\rho_h}$ is non-unitary because the representation matrix $\Gamma^{[2]}$ has zero elements.
One can verify that we can realize $F^{[2],C_1}_{\rho_h}$ by the following circuit:
\begin{equation}
    \begin{gathered}   \includegraphics{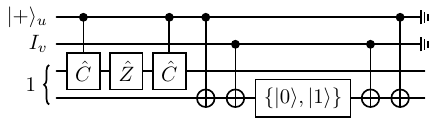}
\end{gathered}
\end{equation}
where the top two lines are the ancilla qubits that are prepared in the $\ket{+}$ and the maximally mixed state $I_v$, respectively.
The bottom two lines are the qutrit and qubit for the edge 1. The $\{\ket{0},\ket{1}\}$ denote the measurement in the corresponding bases.

\begin{figure*}[t]
    \centering
    \includegraphics[width = 0.95\textwidth]{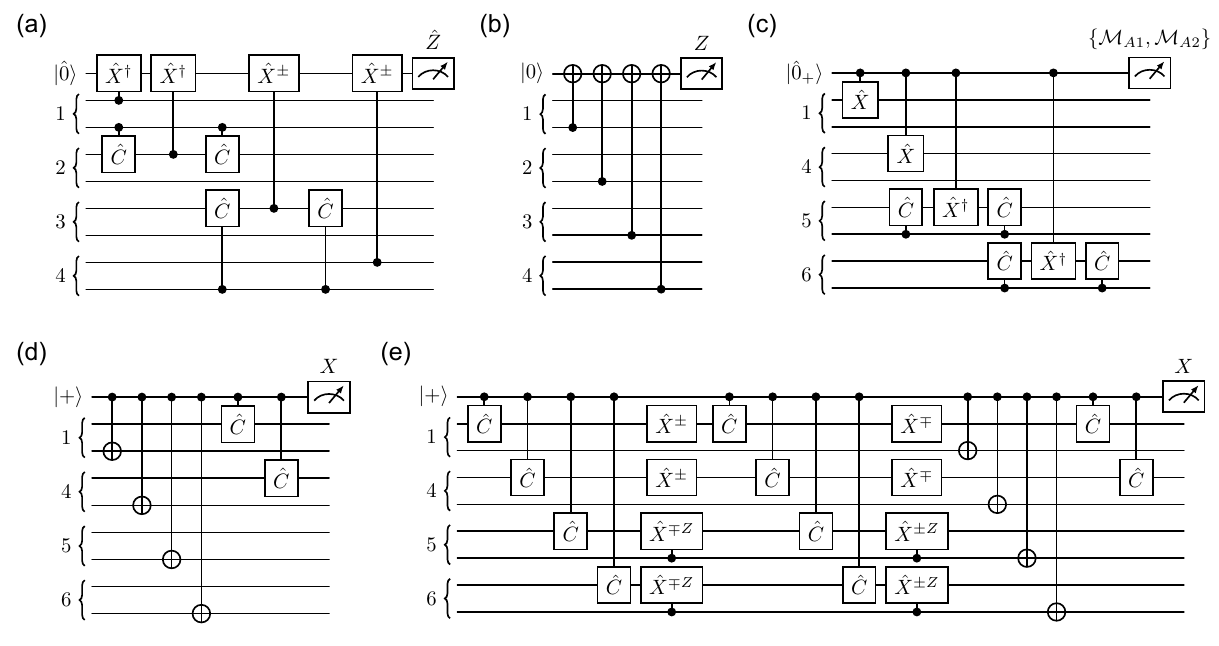}
    \caption{Quantum circuits for the $K^{R,C}_{s}$ measurement. The measurement of $B^h_p$ is implemented by adaptively measuring the operators $\hat{Z}^{\pm}_p$ and $Z_p$. (a) $\hat{Z}^{\pm}_p$ measurement, where the $+$ or $-$ sign denotes $\hat{Z}^+_p$ or $\hat{Z}^-_p$ respectively. (b) $Z_p$ measurement. (c) $A^\mu_v$ measurement, with $\mathcal{M}_{A1} = \{\ket{\hat{0}_+},\ket{\hat{0}_+}^\perp\}$ or $\mathcal{M}_{A2} = \{\ket{\hat{0}_+},\ket{\hat{1}_+},\ket{\hat{2}_+}\}$ on the ancilla qutrit corresponding to measurements in the $C_1$ or $C_3$ class, respectively. (d) $A^{\sigma}_v$ measurement. (e) $A^{\mu \sigma}_v$ or $A^{\mu^2\sigma}_v$ measurement, with the $+$ or $-$ sign indicating $A^{\mu\sigma}_v$ or $A^{\mu^2\sigma}_v$, respectively.}
    \label{fig:KRC_circuit}
\end{figure*}

To figure out the circuits for local anyon type measurement, let us recall the expression for $K_s^{R,C}$
\begin{equation*}
\begin{gathered}
    K^{R,C}_{s} = \frac{|R|}{|Z(C)|} \sum_{n\in Z(C),c\in C} \bar{\chi}_R(n)A^{\tau_c n \bar{\tau}_c}_v B^c_p \\
    \begin{tikzpicture}[baseline={(current bounding box.center)}, scale=0.9]
    \foreach \i in {0,1}
        {\foreach \j in {0,1}{
            \draw[] (\i,\j+1) -- (\i,\j);
            \draw[] (\i,\j+1) -- (\i+1,\j+1);
            \draw[] (\i,\j) -- (\i+1,\j);
            \draw[] (\i+1,\j+1) -- (\i+ 1,\j);
            }}
    \draw (1,1) -- (1.5,0.5);
    \node at (1.35,0.83) {\scriptsize$s$};
    \draw [line width=1.2pt] (1,1) -- (2,1);
    \draw [line width=1.2pt] (2,1) -- (2,0);
    \draw [line width=1.2pt] (1,0) -- (2,0);
    \draw [line width=1.2pt] (1,0) -- (1,1);
    \draw [line width=1.2pt,dashed,dash pattern=on 5pt off 3pt] (1.5,0.5) -- (1.5,1.5);
    \draw [line width=1.2pt,dashed,dash pattern=on 5pt off 3pt] (0.5,0.5) -- (1.5,0.5);
    \draw [line width=1.2pt,dashed,dash pattern=on 5pt off 3pt] (0.5,0.5) -- (0.5,1.5);
    \draw [line width=1.2pt,dashed,dash pattern=on 5pt off 3pt] (0.5,1.5) -- (1.5,1.5);    \node at (0.8,0.8) {\scriptsize$1$};
    \node at (1.2,1.2) {\scriptsize $4$};
    \node at (1.8,0.8) {\scriptsize$3$};  \node at (1.2,0.2) {\scriptsize$2$}; \node at (0.2,1.2) {\scriptsize$6$};
    \node at (0.8,1.8) {\scriptsize$5$};
    \end{tikzpicture}\;.
\end{gathered}
\end{equation*}
To construct the measurement circuits, we first write $K^{R,C}_s$ in terms of projectors onto different charge and flux sectors more explicitly. For $C = C_1$, we have: 
\begin{equation} 
\begin{aligned}
    K^{[+],C_1}_s &= P^{0}_{A^\mu_v} P^{0}_{A^\sigma_v} B^e_p\;,\\
    K^{[-],C_1}_s &= P^{0}_{A^\mu_v} P^{1}_{A^\sigma_v} B^e_p\;,\\
    K^{[2],C_1}_s &= (P^{1}_{A^\mu_v} + P^{2}_{A^{\mu}_v}) B^e_p\;,
\end{aligned}
\end{equation}
where we use the same notation as in \eqnref{eqn:T_circuits}.
For $C = C_2$, we have:
\begin{equation}
\begin{aligned} 
    K^{[+],C_2}_s &= P^{0}_{A^{\sigma}_v}B_p^{\sigma} + P^{0}_{A^{\mu\sigma}_v} B_p^{\mu \sigma} + P^{0}_{A^{\mu^2\sigma}_v} B_p^{\mu^2 \sigma}\;,\\
    K^{[-],C_2}_s &= P^{1}_{A^{\sigma}_v}B_p^{\sigma} + P^{1}_{A^{\mu\sigma}_v} B_p^{\mu \sigma} + P^{1}_{A^{\mu^2\sigma}_v} B_p^{\mu^2 \sigma}\;.
\end{aligned}
\end{equation}
For $C = C_3$, we have:
\begin{equation}
\begin{aligned} 
    K^{1,C_3}_s &= P^{0}_{A^\mu_v} B^{\mu}_p + P^{0}_{A^{\mu^2}_v} B^{\mu^2}_p\;,\\
    K^{\omega,C_3}_s &= P^{1}_{A^\mu_v} B^{\mu}_p + P^{2}_{A^{\mu^2}_v} B^{\mu^2}_p\;,\\
    K^{\bar{\omega},C_3}_s &= P^{2}_{A^\mu_v} B^{\mu}_p + P^{1}_{A^{\mu^2}_v} B^{\mu^2}_p\;.
\end{aligned}
\end{equation}
It follows from these expressions that one can realize the measurement of $K^{R,C}_s$ by measuring $A^{g}_v$ and $B^h_p$ separately. This is consistent with the intuition that each anyon is a dyon in the quantum double model with its flux detected by the plaquette operators and charge by the star operators.

The measurement circuits for $B^h_p$ are shown in \figref{fig:KRC_circuit} (a,b), and that for $A^g_v$ are shown in \figref{fig:KRC_circuit} (c,d,e). As detailed in Appendix~\ref{appendix:circuit_K_RC_s}, the measurement of $B^h_p$ can be decomposed into adaptive measurements of the two operators $\hat{Z}_p^\pm$ and $Z_p$ (see \eqnref{eqn:Z_p}), so their corresponding circuits are shown \figref{fig:KRC_circuit} (a,b). 
These circuits are analogous to the typical syndrome measurement circuits for stabilizer codes~\cite{Fowler2012_SurfaceCode}, where the entangling gate maps the stabilizers of the ancillas to the tensor products with the stabilizers of the physical qudits. Then, by measuring the ancilla on the corresponding basis, we effectively realize the measurement of the physical degrees of freedom. 
Showing these circuits realize the desired measurement is a non-trivial task and we leave the details in Appendix~\ref{appendix:circuit_K_RC_s}.

Note that \figref{fig:KRC_circuit} (a) and (c) involve a new qutrit-to-qutrit controlled gate, that applies $\hat{I}$, $\hat{X}$ or $\hat{X}^\dag$ to the target qutrit if the control qutrit is in the state $\ket{\hat{0}}$, $\ket{\hat{1}}$ or $\ket{\hat{2}}$. This is the qutrit version of the qubit CNOT gate, as also appears in the circuits of $F^{R,C}_{\rho_{h/v}}$ (see Appendix~\ref{appendix:more_circuit_realization}). One can verify that it is also a Clifford gate
\begin{equation} \label{eqn:algebra_qutrit_CX}
\begin{aligned}
    &C\hat{X}(\hat{I}\otimes \hat{Z})C\hat{X}^\dagger = \hat{Z}^\dagger\otimes \hat{Z}\;,\\
    &C\hat{X}(\hat{X}\otimes \hat{I})C\hat{X}^\dagger = \hat{X}\otimes \hat{X}\;,
\end{aligned}
\end{equation}
where the first and second qutrits are the control and target, respectively.

Combining all these elements, we have the following adaptive protocol for the $K^{R,C}_s$ measurement:
\begin{enumerate}
    \item Measure $B^{h}_p$. We first measure $Z_p$. Depending on the outcome $l_h = 0$ or $l_h = 1$, we measure $\hat{Z}_p^+$ or $\hat{Z}_p^-$ and obtain $k_h$. These results determine $h = \mu^{k_h}\sigma^{l_h}$. 
    
    \item If $h = e \in C_1$, measure $A^{\mu}_v$ in the $\mathcal{M}_{A1} = \{\ket{\hat{0}_+},\ket{\hat{0}_+}^\perp\}$ basis. The $\ket{\hat{0}_+}^\perp$ outcome indicates the $C$ anyon. If the outcome is $\ket{\hat{0}_+}$, perform the $A^\sigma_v$ measurement, where the $+1$ or $-1$ outcome corresponds to the $A$ or $B$ anyon, respectively.
    
    \item If $h = \sigma,\mu\sigma,\mu^2\sigma \in C_2$, measure $A^{\sigma}_v, A^{\mu\sigma}_v, A^{\mu^2\sigma}_v$ correspondingly. The $+1$ or $-1$ outcome corresponds to the $D$ or $E$ anyon, respectively.
    
    \item If $h = \mu,\mu^2 \in C_3$, measure the $A^{\mu}_v,A^{\mu^2}_v$ in the $\mathcal{M}_{A2}=\{\ket{\hat{0}_+},\ket{\hat{1}_+},\ket{\hat{2}_+}\}$ basis correspondingly. The $1,\omega$ or $\bar{\omega}$ and $1,\bar{\omega}$ or $\omega$ measurement outcomes for $A^\mu_v$ and $A^{\mu^2}_v$ correspond to the $F,G$ or $H$ anyon, respectively.
\end{enumerate}

In summary, we have constructed circuits for applying the shortest ribbon operators and performing anyon type measurement, which are the basic ingredients for implementing universal quantum computation at the circuit level. 
The fact that there is only one non-Clifford gate simplifies the error control, as discussed in the following section.

\subsection{Experimental Feasibility}
As demonstrated in the previous sections, our circuit model relies solely on elementary gates and measurements involving qubits and qutrits, making it compatible with a range of near-term quantum platforms. One promising candidate for experimental realization is superconducting circuits. Current technology supports high-fidelity single- and multi-qubit gates~\cite{Arute2019_QuantumSupremacy,Satzinger_2021_toric_code,Bravyi_2024_Nature,2024_Google_Surface_Code}, and has already enabled the realization of non-Abelian topological orders such as the Ising and Fibonacci phases~\cite{Andersen2023Nature_Google_NonAbelian,Xu2024_Fibonacci,minev2024realizingstringnetcondensationfibonacci}. Beyond qubits, recent advances have demonstrated the feasibility of implementing high-fidelity single- and two-qutrit gates in superconducting platforms~\cite{Yurtalan_2020_qutrit,Morvan_2021_qutrit,Goss_2022_NC_qutrit,Luo_2023_PRL_qutrit_Ent,Goss2024_qutrit}, making them a natural fit for realizing the $\mathcal{D}(S_3)$ model. 
Furthermore, our proposal features geometrically local and low-weight syndrome measurements, which align well with the local connectivity constraints inherent in superconducting transmon architectures.

Another promising platform is the trapped-ion system~\cite{2015_Senko_Trapped_ion_Interger,2017_Bermudez_Trapped_Ion_FTQC,Ringbauer_2022_NP_trapped_ion_qudits,Moses_2023_PRX_trapped_ion_high_fid,decross2024computationalpowerrandom_trappedion,edmunds2024constructingspin1haldanephase}, which exhibits exceptionally high fidelities for both single- and two-qubit gates—surpassing $99.997\%$ and $99.8\%$, respectively~\cite{Moses_2023_PRX_trapped_ion_high_fid,decross2024computationalpowerrandom_trappedion}. In recent years, several proposals have demonstrated the feasibility of realizing qutrits in trapped-ion systems~\cite{Ringbauer_2022_NP_trapped_ion_qudits,edmunds2024constructingspin1haldanephase,Meth2025_NP_trapped_ion_gauge_field}, and the qutrit toric code—often considered a precursor to the $\mathcal{D}(S_3)$ quantum double—has already been experimentally implemented~\cite{iqbal2024qutrittoriccodeparafermions}. Furthermore, the first experimental realization of the non-Abelian quantum double $\mathcal{D}(D_4)$ was achieved on a trapped-ion platform~\cite{Iqbal2024Nature_Quantinuum_non_Abelian,Iqbal2024_Feedforward}, suggesting its viability for realizing $\mathcal{D}(S_3)$ as well. Notably, the all-to-all connectivity and high-fidelity measurements inherent to trapped-ion systems make them particularly well-suited to our protocol, which requires numerous ancilla qudits and repeated syndrome measurements.

Our proposal is also compatible with quantum computing platforms based on neutral atoms~\cite{Lukin2021SpinLiquid,Bluvstein2022,Bluvstein_2023,Ma_2021_PRX_Yb,Ma_2023_Nature_Yb_high_fid,ammenwerth2024realizationfasttriplemagicalloptical,evered2025probingtopologicalmatterfermion}. Early demonstrations of high-fidelity quantum operations with qubits in ${}^{87}$Rb~\cite{Lukin2021SpinLiquid,Bluvstein2022,Bluvstein_2023} have inspired ongoing efforts to explore alternative atomic species such as ${}^{171}$Yb~\cite{Ma_2021_PRX_Yb,Ma_2023_Nature_Yb_high_fid,Jia_2024_npj_QI_Yb_high_dim,Nakamura_2024_PRX_dual_Yb,muniz_2024_highfidelityuniversalgates171yb} and ${}^{88}$Sr~\cite{ammenwerth2024realizationfasttriplemagicalloptical}. These atoms possess multiple accessible energy levels, making them natural candidates for realizing qutrit-based protocols~\cite{Jia_2024_npj_QI_Yb_high_dim,ammenwerth2024realizationfasttriplemagicalloptical}. Furthermore, recent advances have demonstrated the feasibility of dual-species neutral atom arrays~\cite{Singh_2022_dual_Elem,Cesa_2023_PRL_Dual_Elem,Nakamura_2024_PRX_dual_Yb,Shen_2024_APS_dual,zhang2025dualtypedualelementatomarrays}, making it possible to simultaneously deploy ${}^{87}$Rb and ${}^{171}$Yb for qubit and qutrit encoding, respectively~\cite{Shen_2024_APS_dual,zhang2025dualtypedualelementatomarrays}. The scalability, intrinsic mobility, and short gate times of neutral atom systems are particularly advantageous for implementing the $\mathcal{D}(S_3)$ quantum double model.

Other platforms, such as ultracold molecules~\cite{Bao_2023_Science_Molecules,Vilas_2024_Molecules,Picard_2024_Nature} and bosonic systems~\cite{Gertler2021_Nature,Xu_2024_PRX_Bosonic,Putterman_2025_Nature}, have also demonstrated the capability to implement high-fidelity quantum gates, making them promising candidates for realizing our proposed architecture. In summary, the simplicity of our circuit model—relying only on low-dimensional systems such as qubits and qutrits—ensures compatibility with a wide range of current universal quantum platforms. We believe this experimental feasibility strengthens the potential of our proposal as a viable route toward realizing fault-tolerant universal quantum computation.

\section{\label{sec:S3_QEC} Remarks on quantum error correction}

During the computational process, errors can happen in each component of the quantum circuits, necessitating active quantum error correction (QEC). 
In this section, we explain how to actually incorporate QEC into the computation by using all the gadgets that we have provided so far.
The error threshold of the QEC algorithm for non-Abelian topological codes can be rather small.
Therefore, we also construct a concatenation encoding scheme by replacing the physical degree of freedom on each edge with another quantum error-correcting code to further suppress the effective physical error rate.

\subsection{Incorporating quantum error correction}
\label{sec:Anyon_Config_Picture}

There have been various algorithms for QEC in non-Abelian topological codes, that are shown to have a finite error threshold either numerically or analytically~\cite{Wotton2014PRX_QEC_S3,Brell_2014_Ising_QEC,Wotton2016PRA_ActiveEC,Hutter_2016_Continuous_Ising, Burton_2017_Fibonacci_QEC,Dauphinais_2017_FT_QEC_non_cyclic,Alexis2022PRX_Threshold_TVCode,schotte2022faulttolerant}.
Most of these discussions work at a phenomenological level, which assumes that errors incoherently generate short anyon pairs. 
Note that the proposed algorithms are also formulated in terms of how to pair up the error anyons as in the surface code~\cite{Fowler2012_SurfaceCode}.
Realistic errors, which occur at the circuit level, are usually represented by qubit and qutrit Pauli operators. These operators create more complicated excitations than in the phenomenological error model for non-Abelian topological codes.
Here, we want to design a syndrome measurement that can convert the circuit-level noise into the phenomenological error model, as it is easier to formulate the QEC algorithm for the latter.
Moreover, the computation algorithm ($U$-model) operates in the anyon configuration picture as well, allowing us to combine QEC and computation more coherently.

Such syndrome measurement is realized by the $\mathcal{M}_K$ measurement, namely we measure the $K_s^{R,C}$ at every site. 
As an example, suppose a single-qubit Pauli $X$ error occurs on a horizontal edge, the subsequent application of $\mathcal{M}_K$ produces non-trivial measurement outcomes at the neighboring vertices $s_1,s_2$ and $s_3$:
\begin{equation}
    \begin{tikzpicture}[baseline={(current bounding box.center)}, scale=0.9]
    \foreach \i in {0,1}
        {\foreach \j in {0,1}{
            \draw[thick] (\i,\j+1) -- (\i,\j);
            \draw[thick] (\i,\j+1) -- (\i+1,\j+1);
            \draw[thick] (\i,\j) -- (\i+1,\j);
            \draw[thick] (\i+1,\j+1) -- (\i+ 1,\j);
            }}
    \draw[ultra thick] (0,1) -- (1,1);
    \draw (0,2) -- (0.25,1.75);
    \draw (1,2) -- (1.25,1.75);
    \foreach \i in {0,1}{           
        \draw (\i,1) -- (0.25+\i,0.75);
    }
    \node at (0.4,1.6) {\scriptsize $s_1$};
    \node at (0.2,0.5) {\scriptsize $s_2$};
    \node at (1.2,0.5) {\scriptsize $s_3$};
    \node at (0.5,1.2) {\scriptsize $X$};
    \end{tikzpicture}
    \ \ \xrightarrow{\mathcal{M}_K}\ \ 
        \begin{tikzpicture}[baseline={(current bounding box.center)}, scale=0.9]
    \foreach \i in {0,1}
        {\foreach \j in {0,1}{
            \draw[thick] (\i,\j+1) -- (\i,\j);
            \draw[thick] (\i,\j+1) -- (\i+1,\j+1);
            \draw[thick] (\i,\j) -- (\i+1,\j);
            \draw[thick] (\i+1,\j+1) -- (\i+ 1,\j);
            }}
    \draw (0,2) -- (0.25,1.75);
    \draw (1,2) -- (1.25,1.75);
    \foreach \i in {0,1}{           
        \draw (\i,1) -- (0.25+\i,0.75);
    }
    \node at (0.4,1.6) {\scriptsize $D$};
    \node at (0.4,0.6) {\scriptsize $D$};
    \node at (1.4,0.6) {\scriptsize $C$};
    \node at (1.4,1.6) {\scriptsize $A$};
    \end{tikzpicture}\;.
\end{equation}
Similarly, for other Pauli errors occurring at the same edge, we have the following results under $\mathcal{M}_K$:
\begin{equation}
\begin{aligned}
    Z &\xrightarrow{\mathcal{M}_K} \{s_1:A, s_2:B, s_3:B\}\;, \\
    Y &\xrightarrow{\mathcal{M}_K} \{s_1:D, s_2:E, s_3:C\}\;, \\
    \hat{X} &\xrightarrow{\mathcal{M}_K} \{s_1:F, s_2:F, s_3:A\}\;, \\
    \hat{Z} &\xrightarrow{\mathcal{M}_K} \{s_1:A, s_2:C, s_3:C\}\;, \\
    \hat{X}\hat{Z} &\xrightarrow{\mathcal{M}_K} \{s_1:F, s_2:G, s_3:C\}\;,
\end{aligned}
\end{equation}
where $\{s_i:\alpha\}$ denotes that an $\alpha$ anyon occurs at $s_i$. 
Therefore, using $\calM_K$, we can map an arbitrary circuit-level noise into an anyon configuration picture and use the existing algorithm developed for the phenomenological error model to perform error correction.

Indeed, QEC in the anyon configuration is sufficient for correcting circuit-level noise. In $\mathcal{D}(S_3)$, ribbon operators generate all possible anyon excitations. Thus, in the anyon configuration picture, correcting anyon errors is equivalent to addressing errors induced by specific $F^{R,C;u,v}_{\rho}$ operators. In Appendix~\ref{appendix:Proof_completeness_anyon_ops}, we prove that the shortest ribbon operators form a complete orthonormal basis. Recall that an error-correcting code that can correct errors in one complete orthonormal basis can correct arbitrary circuit-level noise with the same support~\cite{Nielsen_Chuang_2010}. Therefore, we only need to operate the QEC algorithm in the anyon configuration picture. 

In all the computation tasks, each step of moving a single anyon requires local anyon charge measurement.
Therefore, we can naturally incorporate QEC by always performing the global $\mathcal{M}_K$ measurement.
This enables us to propose a general scheme for doing computation and simultaneously correcting circuit-level noise. 
Given an initial configuration $Q_0$ and the target configuration $Q$, we can iteratively apply the following protocol until reaching the final state
\begin{equation} \label{eqn:basic_step_active_EC}
    \begin{tikzpicture}[baseline={(current bounding box.center)}]
        \draw[-{Stealth}] (3.75,0)   -- (3.75,5);
        \node at (3.95,4.8) {$t$};
        \node at (0,-0.25) {$\cdots$};
        \draw[-{Stealth}] (0,0)   -- (0,0.75);
        \node at (0.65,0.375) {$\mathcal{E}(Q_0)$};
        \node[rectangle,draw,rounded corners,line width=0.8pt] (r) at (0,1) {Syndrome Measurement $\mathcal{M}_K$};
        \draw[-{Stealth}] (0,1.25)   -- (0,2);
        \node at (0.4,1.625) {$Q_1$};
        \node[rectangle,draw,rounded corners,line width=0.8pt] (r) at (0,2.25) {Decode $Q_1$};
        \draw[-{Stealth}] (0,2.5)   -- (0,3.33);
        \node at (1.6,2.9) {Movement pattern};
        \node[rectangle,draw,rounded corners,line width=0.8pt] (r) at (0,3.65) {Apply $F^{R,C}_{\rho_h}$ and $F^{R,C}_{\rho_v}$};
        \draw[-{Stealth}] (0,4.0)   -- (0,4.75);
        \node at (0.65,4.375) {$\mathcal{R}(Q_1)$};
        \node at (0,5) {$\cdots$};
    \end{tikzpicture}
\end{equation}
In each round, $\mathcal{M}_K$ maps the input state $\mathcal{E}(Q_0)$, corrupted by some error channel $\mathcal{E}$, into the anyon configuration $Q_1$. Then we apply the decoding algorithm to determine how to manipulate the anyons and apply $F^{R,C}_{\rho_{h/v}}$ afterward, represented as a recovery operation $\mathcal{R}$. 
The new state $\mathcal{R}(Q_1)$ is fed to the next round as the input. The QEC and computation are deemed successful if the configuration $Q$ is achieved without logical errors.

The actual error threshold for non-Abelian codes can be rather small even if we assume a perfect $\calM_K$ syndrome measurement and only consider errors occurring on idling qudits and in the circuits of the ribbon operators~\cite{Hutter_2016_Continuous_Ising,Dauphinais_2017_FT_QEC_non_cyclic}. Intuitively, since non-Abelian anyons cannot be moved instantaneously, additional errors may occur during this process. On the other hand, applying ribbon operators may propagate a single-site error to two neighboring sites and increase the bare error rate.
Therefore, it is necessary to suppress the physical error rate for each qubit-qutrit pair to a sufficiently small value in order to achieve the error threshold.

\subsection{\label{sec:ECC} Effective error rate suppression via concatenation}

To suppress the physical error rate, we can replace the qubit and qutrit on each edge with a small qubit and qutrit quantum error correcting code, which is called the ``local code".
Intuitively, we can regard this architecture as an analog of the concatenation code, where we concatenate the $\calD(S_3)$ quantum double topological code with local codes defined on each edge. 
We apply all the gates in the circuit construction at the logical level of these local codes.
This method reduces the effective physical error rate to the logical error rate of these codes, which is hopefully below the threshold.
In this section, we provide an explicit construction of the local code based on the qudit Calderbank-Shor-Steane (CSS) code~\cite{Shor1995PRA_QEC,Shor1996PRA_ShorCode,Steane1996PRL_SteaneCode}. 
In particular, it allows a transversal realization for all the necessary Clifford logical gates and a fault-tolerant realization of the controlled charge conjugation gate $C\hat{C}$.

For $p$-dimensional qudits, a $n$-qudit CSS code is defined by two $p$-ary parity check matrices $H_X$ and $H_Z$. 
The rows of $H_X$ and $H_Z$ specify the $X$ and $Z$ stabilizers and thus the code space. 
We use $\mathcal{H}_X$ and $\mathcal{H}_Z$ to denote the linear spaces spanned by the rows of $H_X$ and $H_Z$, respectively.
The code distance is the minimal Hamming weight of the nonzero vectors in $\mathcal{H}_Z^\perp$ and $\mathcal{H}_X^\perp$. 
The dimension of the code subspace is given by the dimension of the coset $\mathcal{H}_Z^\perp/\mathcal{H}_X$, and a basis of the code state reads
\begin{equation} \label{eqn:code_states_CSS}
    \ket{x+\mathcal{H}_X} \equiv  \sum_{y\in \mathcal{H}_X} \ket{x+y}\,,\, x\in \mathcal{H}_Z^\perp/\mathcal{H}_X\,.
\end{equation}
An $[[n,k,d]]$ qudit code uses $n$ physical qudits to encode $k$ logical qudits with a distance $d$.

One important example is the $[[9,1,3]]$ Shor code, which uses 9 qubits to protect a single qubit and is robust against any single-qubit errors~\cite{Shor1995PRA_QEC}. Specifically, the parity check matrices are 
\begin{equation}
    H_X = \begin{bmatrix}
       1 & 1 & 1 & 1& 1 & 1 & 0 & 0 & 0\\
       0 & 0 & 0 & 1 & 1 & 1 & 1 & 1 & 1
   \end{bmatrix}
\end{equation}
and
\begin{equation}
    H_Z = \left[
\begin{array}{c|c|c}
\begin{array}{ccc}
1 & 1 & 0 \\
0 & 1 & 1 \\
\end{array} & 
& \\ \hline &
\begin{array}{ccc}
1 & 1 & 0 \\
0 & 1 & 1 \\
\end{array} & \\\hline   & &
\begin{array}{ccc}
    1 & 1 & 0 \\
    0 & 1 & 1
\end{array}
\end{array}
\right]\;.
\end{equation}
We define $\tilde{0} \equiv 000, \tilde{1} \equiv 111$ and can write the two logical states as
\begin{equation}
\begin{aligned} \label{eqn:logical_states_Shor_code}
    &\ket{0}_L = \ket{\tilde{0}\tilde{0}\tilde{0}}+\ket{\tilde{1}\tilde{1}\tilde{0}}+\ket{\tilde{0}\tilde{1}\tilde{1}}+\ket{\tilde{1}\tilde{0}\tilde{1}}\\
    &\ket{1}_L  = \ket{\tilde{1}\tilde{1}\tilde{1}}+\ket{\tilde{0}\tilde{0}\tilde{1}}+\ket{\tilde{1}\tilde{0}\tilde{0}}+\ket{\tilde{0}\tilde{1}\tilde{0}}\;,
\end{aligned}
\end{equation}
with logical Pauli operators $X_L = X_1\ldots X_9,Z_L = Z_1\ldots Z_9$. 

As discussed in \secref{sec:quantum_double_circuit}, our circuits for anyon manipulation involve basic Pauli's, the qubit CNOT gate, the qutrit $C\hat{X}$ and $\hat{C}$ gates, and the non-Clifford $C\hat{C}$ gate. 
We must realize all these gates at the logical level in a fault-tolerant manner to suppress error propagation. 
It follows from the CSS structure that any qubit code supports transversal logical CNOT gates and any qutrit code supports transversal logical $\hat{C}$ and $C\hat{X}$ gates, which guarantees the fault tolerance automatically.
The remaining task is to properly design a pair of the qubit and qutrit code that also allows a fault-tolerant realization of the non-Clifford logical $C\hat{C}$ gate.

To gain an intuition of our final construction, let us start by examining the most naive idea of constructing a fault-tolerant logical $C\hat{C}$.
Specifically, we consider a qubit and qutrit code by $\mathrm{CSS}(X,H_X;Z,H_Z)$ and $\mathrm{CSS}(\hat{X},\hat{H}_{\hat{X}};\hat{Z},\hat{H}_{\hat{Z}})$ with the code parameters $[[n,1,d_1]]$ and $[[n,1,d_2]]$.
Namely, they only protect a single logical qubit and qutrit, and share the same length $n_1 = n_2 = n$.
In this case, the basis of their logical states reads
\begin{equation} \label{eqn:logical_states_qubit_qutrit_code}
\begin{aligned}
    \ket{\alpha}_L =& \sum_{y \in \mathcal{H}_X} \ket{\alpha x + y} \;,\quad  \alpha \in \bbF_2 \,,\\
    \ket{\hat{\beta}}_L =& \sum_{\hat{y} \in \hat{\mathcal{H}}_{\hat{X}}} \ket{\hat{\beta} \hat{x}+\hat{y}} \;,\quad \hat{\beta} \in \bbF_3 \,,
\end{aligned}
\end{equation}
where $x \in \mathcal{H}_Z^\perp / \mathcal{H}_X$, $\hat{x} \in  \hat{\mathcal{H}}_{\hat{Z}}^\perp / \hat{\mathcal{H}}_{\hat{X}}$ are the generators of the coset spaces.
On the one hand, the logical controlled-charge conjugation gate $C\hat{C}^L$ should act on the logical states as
\begin{equation} \label{eqn:CC_action_logical_states}
    C\hat{C}^L \ket{\alpha}_L \ket{\hat{\beta}}_L = \ket{\alpha}_L \ket{(1+\alpha)\cdot \hat{\beta}}_L\;,
\end{equation}
where the the multiplication between elements in $\mathbb{F}_2$ and $\mathbb{F}_3$ is defined as $\alpha \cdot \hat{\beta}\ \mathrm{mod}\ 3$.
If we naively apply $C\hat{C}$ at the physical level transversally, we have
\begin{equation}
\begin{aligned}
    & C\hat{C}^{\otimes n} \ket{\alpha}_L\ket{\hat{\beta}}_L \\ 
    = & \sum_{y,\hat{y}} \ket{\alpha x + y} \ket{(\hat{\beta}\hat{x} + \hat{y})+(\alpha x + y)\wedge (\hat{\beta}\hat{x} + \hat{y})}
\end{aligned}
\label{eq:naive application}
\end{equation}
where $\wedge$ denotes a component-wise multiplication in $\mathbb{F}_3$, known as the Schur product
\begin{equation}
    (u_1,\cdots,u_n) \wedge (v_1,\cdots,v_n) = (u_1v_1, \cdots,u_nv_n)\;.
\end{equation}
In order for the naive version \eqnref{eq:naive application} to match \eqnref{eqn:CC_action_logical_states}, it is necessary to have $(\alpha x + y)\wedge (\hat{\beta} \hat{x} + \hat{y})$ remains in the space $\hat{\mathcal{H}}_{\hat{X}} + (\alpha-1)\cdot \hat{\beta}\hat{x}$. 

However, this condition cannot be satisfied if we want $\mathrm{CSS}(\hat{X},\hat{H}_{\hat{X}};\hat{Z},\hat{H}_{\hat{Z}})$ to have a maximized code distance.
Indeed, we need $\hat{\calH}_{\hat{X}}$ and $\hat{\calH}_{\hat{Z}}^\perp$ to be the so-called \textit{indecomposable} linear vector spaces, meaning they cannot be expressed as a direct sum of subspaces; otherwise, the code distance is limited by the largest subspaces of them. 
For an indecomposable linear vector space $\mathcal{H} \subseteq \mathbb{F}^n_q$ and any vector $a\in \mathbb{F}^n_q$, $a\wedge \mathcal{H} \subseteq \mathcal{H}$ holds if and only if $a$ is a multiple of the identity vector~\cite{randriambololona2014products}.
Here, $(y+\alpha x)$ cannot be a multiple of the identity vector, as this would trivialize the qubit logical state $\ket{\alpha}_L$.  

Instead, we adjust the code parameters by considering a larger qubit code with $n_1>n_2$ and divide the physical qubits into small groups of $n_2$ qubits, allowing each group to serve as an independent control for the qutrit code.
Specifically, we have the following construction:
\begin{lemma} \label{lemma:logical_CC}
    For any odd integer $n$, an $[[n^2, 1, n]]$ qubit Shor code $\mathrm{CSS}(X,H_X;Z,H_Z)$ and an $[[n,1,d]]$ qutrit CSS code $\mathrm{CSS}(\hat{X},\hat{\mathcal{H}}_{\hat{X}};\hat{Z},\hat{\mathcal{H}}_{\hat{Z}})$ enable a fault-tolerant logical controlled $\hat{C}$ gate given by:
    \begin{equation}
        C\hat{C}^L = C\hat{C}^{\otimes n}_{1\sim n} \hat{\mathcal{R}} C\hat{C}^{\otimes n}_{n+1 \sim 2n} \cdots \hat{\mathcal{R}} C\hat{C}^{\otimes n}_{(n-1)n+1 \sim n^2}\;,
    \end{equation}
    where $C\hat{C}^{\otimes n}_{i\sim j}$ is the transversal $C\hat{C}$ gate controlled by the $i$ through $j$ (a total of $n$) physical qubits over the entire qutrit code, and $\hat{\mathcal{R}}$ represents the error correction operation on the qutrit code. 
\end{lemma}
\begin{proof}
It suffices to consider the case with $n=3$. The general case can be proved similarly.\footnote{It follows from the no-cloning theorem that there is no three-qutrit code with a distance larger than $1$. The $n=3$ case is merely a toy example and will not be useful in practice.}
In this case, a logical $C\hat{C}^L$ is
\begin{equation}
\label{eqn:logical_CC_Shor}
    C\hat{C}^L = C\hat{C}_{1\sim 3}^{\otimes 3}\, C\hat{C}^{\otimes 3}_{4\sim 6} \, C\hat{C}^{\otimes 3}_{7\sim 9}\;.
\end{equation}
To verify this claim, we first note that in the $[[9,1,3]]$ Shor code, the qubits in the group $1\sim 3$, $4\sim 6$ and $7\sim 9$ consistently take the bit string values $\tilde{0} = 000$ or $\tilde{1} = 111$, as shown in \eqref{eqn:logical_states_Shor_code}. 
Using $\tilde{\alpha} = \tilde{0}$ or $\tilde{1}$ to denote the restriction of the codeword to qubits $1\sim 3$, the qutrit state after applying $C\hat{C}^{\otimes 3}_{1\sim 3}$ reads
\begin{equation} \label{eqn:action_CC_3_qubits}
    \sum_{\hat{y}\in \hat{\mathcal{H}}_{\hat{X}}}\ket{(\tilde{1}+\tilde{\alpha})\wedge (\hat{y}+\hat{\beta}\hat{x})} = \begin{cases}
        \ket{\hat{\beta}}_L\;, & \tilde{\alpha} = \tilde{0}\\
        \ket{2\hat{\beta}}_L\;, & \tilde{\alpha} = \tilde{1}
    \end{cases}\;,
\end{equation}
where we have used the fact that $2\hat{y}$ is equivalent to $\hat{y}$ when summing over all $\hat{y}\in \hat{\mathcal{H}}_{\hat{X}}$. 
The results for applying $C\hat{C}^{\otimes}_{4\sim 6}$ and $C\hat{C}^{\otimes}_{7\sim 9}$ are similar.
Therefore, by applying $C\hat{C}^{\otimes 3}_{1\sim 3}\, C\hat{C}^{\otimes 3}_{4\sim 6}\, C\hat{C}^{\otimes 3}_{7\sim 9}$, each occurrence of $\tilde{1}$ in the codeword of Shor code induces a mapping $\ket{\hat{\beta}}_L \to \ket{2\hat{\beta}}_L$.   
Using the following arithmetic property
\begin{equation} \label{eqn:arithmetic_F2_F3}
    2^q \ \mathrm{mod}\ 3=\begin{cases}
        1,\quad q\ \mathrm{\ even} \\
        2,\quad q\ \mathrm{\ odd}
    \end{cases}\;,
\end{equation}
we obtain the desired action of $C\hat{C}^L$
\begin{equation}
\begin{aligned}
    C\hat{C}^{\otimes 3}_{1\sim 3}\, C\hat{C}^{\otimes 3}_{4\sim 6}\, C\hat{C}^{\otimes 3}_{7\sim 9} \ket{\alpha}_L\ket{\hat{\beta}}_L =\ket{\alpha}_L \ket{(1 + \alpha) \cdot \hat{\beta}}_L\,.
\end{aligned}
\end{equation}
Note that $C\hat{C}^L$ defined by \eqnref{eqn:logical_CC_Shor} is not transversal and thus not fault-tolerant. 
Fortunately, \eqnref{eqn:action_CC_3_qubits} shows that $C\hat{C}^{\otimes 3}$ preserves the codespace of the qutrit code.
Therefore, we can perform qutrit error correction $\hat{\mathcal{R}}$ on it after each round of $C\hat{C}^{\otimes 3}$ to prevent the single-qutrit errors from propagating. Namely, we should consider
\begin{equation}
    C\hat{C}^L = C\hat{C}_{1\sim 3}^{\otimes 3}\, \hat{\mathcal{R}} \, C\hat{C}^{\otimes 3}_{4\sim 6} \,  \hat{\mathcal{R}} \, C\hat{C}^{\otimes 3}_{7\sim 9}\;
\end{equation}
as our final logical $C\hat{C}^L$ which satisfies all the requirements.
\end{proof}

After the encoding, the effective physical error rate has an exponential scaling with the code distance, i.e., $p_{\mathrm{eff}} \propto \mathrm{poly}(d)p^{\lfloor d/2 \rfloor + 1}$. 
However, the logical $C\hat{C}^L$ gate has a circuit depth proportional to $n\sim \mathrm{poly}(d)$, which causes the fundamental time unit of any gate that involves $C\hat{C}$ to become $n$ times larger after the encoding.
Among the local codes allowed in Lemma~\ref{lemma:logical_CC}, the smallest one is $n=7$, corresponding to a pair of $[[49,1,7]]$ qubit Shor code and $[[7,1,3]]$ Steane qutrit code~\cite{Steane1996PRL_SteaneCode}.
We hope the exponential suppression of the error rate is worth the polynomial overhead. However, it remains to be tested whether this additional encoding is beneficial in practice.

\section{\label{sec:SummaryOutlook} Summary and outlook}
In this work, we propose a comprehensive blueprint for realizing a large-scale anyon-based quantum computer using the $\mathcal{D}(S_3)$ quantum double model. Previous studies have identified $\mathcal{D}(S_3)$ as the minimal non-Abelian topological order that balances preparability, computational power, and error correctability. Building on this foundation, our work integrates these aspects into a unified architecture for a non-Abelian quantum processor in two dimensions, requiring only local connectivity and avoiding magic state distillation. To this end, we address several central challenges, including coherent anyon braiding and measurement in a potentially fault-tolerant manner, experimentally viable circuit designs, and the incorporation of quantum error correction. Taken together, we believe these protocols represent a significant step toward constructing a fault-tolerant universal quantum computer.

From a fundamental standpoint, our braiding and measurement protocols establish a concrete link between quantum fault tolerance and the abstract framework of anyon theory. In particular, our coherent braiding scheme resolves the issue of error propagation along anyon paths—a common source of logical error in non-Abelian topological orders. Moreover, the fact that $S_3$ is a solvable group necessitates the use of fusion operations for universality, which can introduce logical errors if implemented naively. We address this by leveraging anyon interferometry to perform the required charge measurements while maintaining large separations between encoding anyons, thereby preserving the topological protection. These techniques are broadly applicable to other non-Abelian topological orders that rely on anyon braiding and charge measurements, and we believe they offer a valuable contribution to the field of topological quantum computation.

Building on these foundations, we demonstrate how the universal gate set proposed by Cui et al.\ can be implemented remotely using our braiding and measurement protocols. This remote realization preserves noise resilience by exploiting the intrinsic locality of non-Abelian anyons. To make our scheme experimentally accessible, we design a circuit model compatible with all required anyon operations. Crucially, the circuits consist of simple gates and measurements acting on qubits and qutrits, with low-weight syndrome extraction and strictly local connectivity. These features render our model particularly feasible for near-term quantum hardware platforms—including superconducting circuits, trapped ions, and neutral atoms—making it a strong candidate for the experimental realization of universal quantum computation with non-Abelian anyons.

In addition to simulating the $\mathcal{D}(S_3)$ topological order, our proposal establishes a viable design paradigm for anyon-based quantum computing by incorporating quantum error correction at the circuit level. We show that generic circuit-level noise can be mitigated using anyon manipulations, reinforcing the robustness of the system. To further suppress the effective error rate, we introduce a concatenation scheme that replaces physical degrees of freedom with local codes, enhancing the scalability of the circuit model.

In summary, this work presents a top-down architecture for an anyon-based quantum computer that offers several distinct advantages: universality, fault tolerance, geometric locality, and experimental feasibility. We anticipate that this proposal can be realized in practice and may serve as a foundation for the eventual construction of large-scale quantum computers. We hope it will stimulate interest across multiple communities—ranging from theory to experiment—and accelerate the development of a topological quantum computer.

One important future direction is to examine quantum error correction more carefully and to assess the resource requirements specific to this proposal.
For example, one can probe the optimal error threshold by mapping the anyon error correction problem onto a statistical physics model and analyzing its phase transition~\cite{Fan2023PRXQuantum,sala2024decoherencewavefunctiondeformationd4,sala2024stabilityloopmodelsdecohering}.
In parallel, developing efficient decoding algorithms, such as minimum-weight perfect matching tailored to $\mathcal{D}(S_3)$, and investigating the corresponding error thresholds will be important—topics we hope to explore in future work.
Building on these results, one can carry out comprehensive resource estimation, as has been done for the surface code~\cite{beverland2022assessingrequirementsscalepractical,Gidney_2021}, color code~\cite{Beverland_2021}, and Fibonacci qubits~\cite{Baraban_2010}, to evaluate the practical advantages of our proposal.

In \secref{sec:adaptive_movement}, we choose a linear-depth adaptive protocol to move anyons. Given the existence of proposals for constant-depth implementations of $\mathcal{D}(G)$ for solvable groups $G$ using adaptive circuits~\cite{verresen2022efficiently,Tantivasadakarn2023PRXQuantum_Hierarchy,bravyi2022adaptive}, it would be interesting to explore whether our linear-depth protocol can be improved to a constant-depth version~\cite{ren2024efficientpreparationsolvableanyons,lyons2024protocolscreatinganyonsdefects}
while maintaining its robustness against errors.
Achieving this would significantly reduce the intrinsic overhead associated with non-Abelian anyons.

In \secref{sec:ECC}, we construct a family of quantum error-correcting codes that incorporate a fault-tolerant, non-Clifford logical $C\hat{C}$ gate between a logical qubit and qutrit.
This raises the intriguing possibility of future exploration into magic-state distillation involving systems of unequal dimensions—extending beyond existing protocols~\cite{Bravyi2005PRA_MSD,Bravyi2012PRA_MSD_LowOverhead,Haah2018codesprotocols,Campbell2012PRX_MDS_PrimeDim,Hastings2018PRL_MSD_sublogrithmic,Campbell2017PRA_unified_MSD}.
Another practical question is the search for codes supporting a logical $C\hat{C}$ that require fewer physical qubits/qutrits, thereby reducing resource overhead, or that allow for a lower-depth implementation of $C\hat{C}$, potentially improving the error threshold.
Additionally, it is natural to generalize the structure of concatenating a topological code with local codes to other non-Abelian topological orders, which would require extending our constructions to more general qudit stabilizer codes. We leave this for future work.

Finally, it may be valuable to explore the circuit complexity of computation based on the $U$-model and to search for efficient quantum algorithms realizable within $\mathcal{D}(S_3)$. Optimizing compilation strategies tailored to this universal gate set is also an important direction for future research.

\section*{Author Contributions}
L.C. formulated the main protocols of universal quantum computation, conducted the primary proof of the qubit-qutrit error-correcting codes, and constructed the circuit realization of $\mathcal{D}(S_3)$. Y.R. contributed to the development of anyon ribbon error correction and probabilistic ribbon movement techniques. R.F. formulated the basic idea of the anyon interferometry and contributed to the design of quantum error correction. A.J. contributed through discussions and insights, particularly in the formulation of the protocols and refinement of the theoretical proof. All authors contributed to writing the manuscript.

\begin{acknowledgments}
We thank Kaifeng Bu, Shawn Xingshan Cui, Michael Freedman, Hong-Ye Hu, Seth Lloyd, Chiu Fan Bowen Lo, Anasuya Lyons,  Nathanan Tantivasadakarn, Ruben Verresen, and Ashvin Vishwanath for their insightful discussion. L.C. thanks Zhiyang He, Issac Kim, Quynh T. Nguyen, Zijian Song, Yi Tan, and Chen Zhao for an insightful discussion on the qutrit-qubit codes.
R.F. thanks Yimu Bao for helpful feedback on the manuscript.
This work was supported in part by the Army Research Office Grant W911NF-19-1-0302, and by the Army Research
Office MURI Grant W911NF-20-1-0082.
This material is partially based upon work supported by the U.S.
Department of Energy, Office of Science, National Quantum Information
Science Research Centers, Quantum Systems Accelerator, under Grant
number DOE DE-SC0012704. 
R.F. is supported by the Gordon and Betty Moore Foundation (Grant GBMF8688).
While completing this manuscript, we became aware of an independent upcoming work by C.F.B. Lo et al.~\cite{lo2024universal} demonstrating universal gate set in $S_3$ quantum double using a different logical encoding. We thank the authors for informing us their work.

\end{acknowledgments}

\bibliography{ref.bib}

\onecolumngrid
\newpage
\appendix
\setcounter{secnumdepth}{2}

\section{Quantum double algebra and lattice realization}\label{appendix:quantum_double}

\subsection{Drinfeld quantum double \texorpdfstring{$\mathcal{D}(G)$}{D(G)}} \label{appendix:Drinfeld_quantum_double}
\begin{table*}[tb]
\centering
    \begin{tabular}{|c|c|c|c|c|c||c||c|}
    \hline
        $C_k$ & $\{c: c \in C_k\}$ & $g_k \in C_k$ & $\{\tau_c:c\in C_k\}$ & $Z_k$ & $R$ & Basis of $V_{(R,C)}$ & Anyon Type \\ 
        \hline
        $C_1$ & $\{e\}$ & $e$ & $\{e\}$ & $S_3$ & $[+]$ & $\ket{+}$ & A  \\ 
        & & & & & $[-]$ & $\ket{-}$ & B \\
        & & & & & $[2]$ & $\ket{2_+},\ket{2_-}$ & C
        \\ 
        \hline
        $C_2$ & $\{(12),(31),(23)\}$ & $\sigma$ & $\{e,\mu^2,\mu\}$ & $\{e,\sigma\} \simeq \mathbb{Z}_2$ & $[+]$ & $\ket{\sigma},\ket{\mu\sigma},\ket{\mu^2\sigma}$& D \\
        & $\{\sigma,\mu\sigma,\mu^2\sigma\}$ & & & & $[-]$ & $\ket{\sigma,-},\ket{\mu\sigma,-},\ket{\mu^2\sigma,-}$ & E\\\hline
        $C_3$ & $\{(123),(132)\}$ & $\mu$ & $\{e,\sigma\}$ & $\{e,\mu,\mu^2\} \simeq \mathbb{Z}_3$ & $1$ & $\ket{\mu},\ket{\mu^2}$ & F \\
        & $\{\mu,\mu^2\}$ & & & & $\omega$ & $\ket{\mu,\omega},\ket{\mu^2,\omega}$ & G\\
        & & & & & $\bar{\omega}$ & $\ket{\mu,\bar{\omega}},\ket{\mu^2,\bar{\omega}}$ & H \\\hline
    \end{tabular}
    \caption{The conjugacy classes $C_k$ and centralizers $Z_k$ of $S_3$, along with the representation $(R,C)$ of $\mathcal{D}(S_3)$. For each $C_k$, the elements $c$ and $\tau_c$ (where $c = \tau_c g_k \bar{\tau}_c$) are listed in corresponding positions. The corresponding anyon types $A\sim H$ in the quantum double model $\mathcal{D}(S_3)$ are listed in the last column.}
    \label{tab:Rep_D_S_3}
\end{table*}

This section is a lightening review of the definition and representation of the Drinfeld quantum double $\mathcal{D}(G)$~\cite{bakalov2001lectures}. For a finite group $G$, the group algebra $A = \mathbb{C}G$ and its dual space $A^*$ are both Hopf algebras, where $A^*$ is spanned by the linear functional $\delta_g \in \mathrm{Hom}(A,\mathbb{C})$, with $\delta_g(h) = \delta_{g,h}$, and $\delta_{g,h}$ is the Kronecker delta for two group elements $g,h\in G$.  The Drinfeld quantum double $\mathcal{D}(G)$ endows the $A^* \otimes A$ with the structure of a Hopf algebra. $\mathcal{D}(G)$ has the $\mathbb{C}$-basis $\{\delta_h g: g,h \in G\}$, with the multiplication given by
\begin{equation} \label{eqn:quantum_double_mult}
    (\delta_{h_1} g_1)\cdot (\delta_{h_2}g_2) = \delta_{h_1,g_1h_2\bar{g}_1}\delta_{h_1}  g_1g_2\;,
\end{equation}
as well as the comultiplication $\Delta(\delta_h g) = \sum_{h_2h_1 = h} \delta_{h_1} g\otimes \delta_{h_2} g$, unit $(\sum_{h}\delta_h) e\in \mathcal{D}(G)$, counit $\epsilon(\delta_h g) = \delta_{h,e}$, and the antipode $S(\delta_h g) = \bar{g}\delta_{\bar{h}} =  \delta_{\bar{g}\bar{h}g} \bar{g}$. It is straightforward to verify that these operations satisfy the commutative diagrams of Hopf algebra \cite{bakalov2001lectures}. 

The group $G$ is partitioned into disjoint union of conjugacy classes as $G = \cup_k C_k$, where $C_k \equiv C(g_k) = \{hg_k\bar{h}: h \in G\}$, with $g_k \in G$ being the representative of the class $C_k$. For every $c \in C_k$, one can select an element $\tau_c \in G$ such that $\tau_c g_k \bar{\tau}_c = c$, and we choose $\tau_{g_k} = e$. For each $C_k$, the centralizer is denoted by $Z_k \equiv Z(C_k) = Z(g_k) = \{h\in G:hg_k = g_k h\}$. Notice that for any $g \in G$ and $c\in C_k$, we have $\bar{\tau}_{gc\bar{g}} g \tau_c \in Z_k$. 

The irreducible representation (irrep) of $\mathcal{D}(G)$ is given by a tuple $(R,C)$, where $C$ is one of the conjugacy classes of $G$ with representative $g_k$, and $R$ is an irrep of $Z_k$ over the vector space $V_R = \{\ket{j}, 1\leq j \leq |R|\}$. The Hilbert space corresponding to $(R,C)$ is given by $V_{(R,C)} = \mathbb{C}G \otimes V_R$, with the basis $\{\ket{c}\otimes \ket{j}: c\in C, 1 \leq j \leq |R|\}$. The elements of $\mathcal{D}(G)$ act on $V_{(R,C)}$ as:
\begin{align} \label{eqn:irrep_action}
    \delta_h \ket{c}\otimes \ket{j} &= \delta_{h,c}\ket{c}\otimes \ket{j}\;,\notag\\
    g\ket{c}\otimes \ket{j} &= \ket{gc\bar{g}}\otimes \sum_{i}\Gamma^{R}_{ij}(\bar{\tau}_{gc\bar{g}} g \tau_c) \ket{i}\;.
\end{align}

The conjugacy classes and centralizers of $S_3$, as well as the representation of $\mathcal{D}(S_3)$ are listed in Table~\ref{tab:Rep_D_S_3}. As defined in the main text (see Table~\ref{tab:conjugacy_irreps}), the anyon types $A\sim H$ and their local degrees of freedom $u=(c,j),v=(c^\prime,j^\prime)$ (see \eqref{eqn:F_RC}) are also listed in the last two columns of Table~\ref{tab:Rep_D_S_3}.

\subsection{Ribbon operator algebra}\label{appendix:ribbon_algebra}

In this section, we review some necessary ribbon algebra. A ribbon $\rho$ comprises a set of sites between the beginning and ending sites $s_0$ and $s_1$, as shown in Figure~\ref{fig:ops_quantum_double}(a). Equivalently, a ribbon is composed of a sequence of triangles. Each triangle is formed by two neighboring sites (dashed lines) and a third edge that connects the points not shared by the two sites. This third edge can be either a direct edge (solid line) or a dual edge (short-dashed line). Depending on the type of this edge, the triangle is classified as either a direct triangle $\tau^\prime$ or a dual triangle $\tau$. On each triangle, the ribbon operator acts as a $T$ or $L$ operator as follows:
\begin{equation} \label{eqn:ribbon_op_triangle}
    F^{h,g}_{\tau^\prime} \equiv T^{g}_{\tau^\prime},\quad F^{h,g}_\tau \equiv \delta_{e,g} L^{h}_\tau\;,
\end{equation}
following the rules illustrated in the right panel of Figure~\ref{fig:ops_quantum_double}(b). If a ribbon $\rho$ is separated by an intermediate site $s$, which is the ending site of $\rho_1$ and the starting site of $\rho_2$, we denote it as $\rho = \rho_1\rho_2$. When acting on $\rho = \rho_1\rho_2$, the ribbon operator $F^{h,g}_\rho$ obeys the comultiplication (or gluing) rule as
\begin{equation} \label{eqn:comult_rule}
    F^{h,g}_\rho = \sum_{m \in G} F^{h,m}_{\rho_1}F^{\bar{m}hm,\bar{m}g}_{\rho_2}\;.
\end{equation}
Therefore, for a long ribbon $\rho$, one can always write the ribbon operator $F^{h,g}_\rho$ by decomposing $\rho$ into a sequence of triangles using the comultiplication rule \eqref{eqn:comult_rule} recursively and applying \eqref{eqn:ribbon_op_triangle} for each triangle. As an example, consider the ribbon in Figure~\ref{fig:ops_quantum_double}(a) with the following triangles:
\begin{equation}
        \begin{tikzpicture}[baseline={(current bounding box.center)}]
    \foreach \i in {0,1,...,2}
        {\foreach \j in {0,1}{
            \draw[thick] (\i,\j+1) -- (\i,\j);
            \draw[thick] (\i,\j+1) -- (\i+1,\j+1);
            \draw[thick] (\i,\j) -- (\i+1,\j);
            \draw[thick] (\i+1,\j+1) -- (\i+ 1,\j);
            }}
    \foreach \i in {0,1}{
        \draw (0.5 + \i,0.5) -- (1+\i,1);
        \draw (0.5+ \i,0.5) -- (1.5+\i,0.5);
    }
    \foreach \i in {0,1,2}{           
        \draw (\i,1) -- (0.5+\i,0.5);
    }
    \draw (2.5,0.5) -- (2.5,1.5);
    \draw (2,1) -- (2.5,1.5);
    \draw (2,2) -- (2.5,1.5);
    \node at (0.2,0.5) {$s_0$};
    \node at (2.5,1.8) {$s_1$};
    \node at (2.7,0.8) {$\tau_5$};
    \node at (2.18,1.5) {$\tau_6^\prime$};
    \node at (1.8, 1.7) {$g_3$};
    \foreach \i in {1,2}{
        \pgfmathtruncatemacro{\j}{2*\i}
        \node at (\i-0.2, 0.3) {$\tau_\j$};
    }
    \foreach \i in {1,2}{
        \pgfmathtruncatemacro{\j}{2*\i-1}
        \node at (\i-0.53, 0.8) {$\tau_\j^\prime$};
    }
    \foreach \i in {1,2}{
    \node at (\i-0.5,1.2) {$g_\i$};
    }
    \end{tikzpicture} \notag
\end{equation}
Then the ribbon operator $F^{h,g}_\rho$ acts as
\begin{equation}
    \sum_{\{g_i\}} \delta_{g,g_1g_2\bar{g}_3} T^{g_1}_{\tau_1^\prime} L^{\bar{g}_1 h g_1}_{\tau_2}T^{g_2}_{\tau_3^\prime} L^{(\bar{g}_2\bar{g}_1) h(g_1g_2) }_{\tau_4} L^{(\bar{g}_2\bar{g}_1) h(g_1g_2) }_{\tau_5} T^{g_3}_{\tau_6^\prime}\;,
\end{equation}
as illustrated in Figure~\ref{fig:ops_quantum_double}(a). 

\section{Additional details on the universal quantum computation protocol based on \texorpdfstring{$\mathcal{D}(S_3)$}{D(S3)}} \label{appendix:UQC}

In this section, we provide supplementary mathematical derivations for the protocols developed in \secref{sec:UQC}. 

\subsection{Review of anyon diagrammatics}

Anyon theories are described using unitary modular tensor categories~\cite{Kitaev2006Annals_AnyonTheory}. Here, we review some necessary anyon diagrammatics for a self-contained derivation. 
For the remainder of this section, we assume that the fusion multiplicity is not larger than 1, i.e., for all $\alpha$, $\beta$, and $\gamma$, we have $N^\gamma_{\alpha\beta} \leq 1$, which is sufficient for our discussion of quantum double.
Our discussion involves the following key identities
\begin{itemize}
\item The first one is the resolution of the identity:
\begin{equation} \label{eqn:resolution_identity}
    \begin{tikzpicture}[baseline = {(current bounding box.center)}]
        \draw (-1,0) -- (-1,1.4);
        \draw (-0.8,0) node {$\alpha$};
        \draw (0,0) -- (0,1.4);
        \draw (0.2,0) node {$\beta$};
    \end{tikzpicture} = \sum_{\gamma} N^\gamma_{\alpha\beta} \sqrt{\frac{d_\gamma}{d_\alpha d_\beta}}\ 
    \begin{tikzpicture}[baseline = {(current bounding box.center)}]
        \draw (-0.8,1.4) -- (-0.4,1);
        \draw (0,1.4) -- (-0.4,1);
        \draw (-0.4,1) -- (-0.4,0.4);
        \draw (-0.4,0.4) -- (-0.8,0);
        \draw (-0.4,0.4) -- (0,0);
        \draw (-1,0) node {$\alpha$};
        \draw (0.2,0) node {$\beta$};
        \draw (-1,1.4) node {$\alpha$};
        \draw (0.2,1.4) node {$\beta$};
        \draw (-0.2,0.7) node {$\gamma$};
    \end{tikzpicture}\;,
\end{equation}
where $\gamma$ sums over all anyons in the theory, and the normalization factor ensures isotopic invariance \cite{Bonderson_2008_Interferometry}, i.e., invariant under topologically trivial deformation of the worldlines. This identity allows the insertion of an anyon worldline at any point in the diagram, facilitating subsequent calculations. 

\item The second identity is the charge conservation:
\begin{equation} \label{eqn:charge_conservation}
    \begin{tikzpicture}[baseline = {(current bounding box.center)}]
        \draw (0,0) -- (0,0.4);
        \draw (0,1.0) -- (0,1.4);
        \draw (0,0.4) .. controls (-0.3,0.5) and (-0.3,0.9) .. (0,1);
        \draw (0,0.4) .. controls (0.3,0.5) and (0.3,0.9) .. (0,1);
        \draw (0.25,0) node {$\gamma$};
        \draw (0.25,1.4) node {$\gamma^\prime$};
        \draw (-0.45,0.7) node {$\alpha$};
        \draw (0.45,0.7) node {$\beta$};
    \end{tikzpicture} = \delta_{\gamma\gamma^\prime} \sqrt{\frac{d_\alpha d_\beta}{d_\gamma}}\ \ \ 
    \begin{tikzpicture}[baseline = {(current bounding box.center)}]
        \draw (0,0) -- (0,1.4);
        \draw (0.2,0) node {$\gamma$};
    \end{tikzpicture}\;,
\end{equation}
where the normalization factor ensures isotopic invariance. Physically, this identity states that if a $\gamma$ anyon is split into an $\alpha$ and a $\beta$ anyon, fusing these two anyons will necessarily result in the original $\gamma$ anyon, reflecting the charge conservation.

\item The third identity is the $R$-move:
\begin{equation} \label{eqn:R_move}
    \begin{tikzpicture}[baseline={(current bounding box.center)}]
        \draw (0,-0.5) -- (0,-0.2);
        \draw (0.2,-0.5) node {$\gamma$};
        \draw (0,-0.2) .. controls (0.4,0.1) and (0.4,0.2).. (0.1,0.27);
        \draw (-0.15,0.35) -- (-0.5,0.5);
        \draw (-0.7,0.5) node {$\beta$};
        \draw (0,-0.2) .. controls (-0.4,0.1) and (-0.5,0.2) .. (0.5,0.5);
        \draw (0.7,0.5) node {$\alpha$};
    \end{tikzpicture}
    =R_{\gamma}^{\alpha\beta}
    \begin{tikzpicture}[baseline={(current bounding box.center)}]
        \draw (0,-0.5) -- (0,0);
        \draw (0.2,-0.5) node {$\gamma$}; 
        \draw (0,0) -- (-0.5,0.5);
        \draw (0,0) -- (0.5,0.5);
        \draw (-0.7,0.5) node {$\beta$};
        \draw (0.7,0.5) node {$\alpha$};
    \end{tikzpicture}\;,
\end{equation}
which indicates that a topologically non-trivial braiding in the anyon diagram can be resolved by introducing a phase factor, known as the $R$-symbol $R^{\alpha\beta}_\gamma$. A unitary theory should satisfy $\left(R^{\alpha\beta}_{\gamma}\right)^{-1} = \left(R^{\alpha\beta}_{\gamma}\right)^*$.

\item The fourth identity is the $F$-move, which can be expressed in two equivalent forms:
\begin{equation} \label{eqn:F_move}
    \begin{tikzpicture}[baseline={(current bounding box.center)}]
        \draw (0,0.1) -- (0,0.5);
        \draw (0,0.5) -- (-0.4,0.9);
        \draw (-0.4,0.9) -- (-0.8,1.3);
        \draw (-0.4,0.9) -- (0,1.3);
        \draw (0,0.5) -- (0.8,1.3);
        \draw (0.2,0.1) node {$\eta$};
        \draw (-0.8,1.5) node {$\alpha$};
        \draw (0,1.5) node {$\beta$};
        \draw (0.8,1.5) node {$\gamma$};
        \draw (-0.35,0.5) node {$\mu$};
    \end{tikzpicture} = \sum_{\nu} [F^{\alpha\beta \gamma}_\eta]_{\mu\nu} 
    \begin{tikzpicture}[baseline={(current bounding box.center)}]
        \draw (0,0.1) -- (0,0.5);
        \draw (0,0.5) -- (0.4,0.9);
        \draw (0.4,0.9) -- (0.8,1.3);
        \draw (0.4,0.9) -- (0,1.3);
        \draw (0,0.5) -- (-0.8,1.3);
        \draw (0.2,0.1) node {$\eta$};
        \draw (-0.8,1.5) node {$\alpha$};
        \draw (0,1.5) node {$\beta$};
        \draw (0.8,1.5) node {$\gamma$};
        \draw (0.35,0.5) node {$\nu$};
    \end{tikzpicture}\;,\quad
    \begin{tikzpicture}[baseline={(current bounding box.center)}]
        \draw (-0.6,0.6) -- (-0.4,0);
        \draw (-0.6,-0.6) -- (-0.4,0);
        \draw (0.6,0.6) -- (0.4,0);
        \draw (0.6,-0.6) -- (0.4,0);
        \draw (-0.4,0) -- (0.4,0);
        \draw (-0.8,-0.6) node {$\alpha$};
        \draw (-0.8,0.6) node {$\beta$};
        \draw (0.8,-0.6) node {$\eta$};
        \draw (0.8,0.6) node {$\gamma$};
        \draw (0,0.2) node {$\mu$};
    \end{tikzpicture} = \sum_{\nu} [F^{\alpha\beta \gamma}_\eta]_{\mu\nu}     \begin{tikzpicture}[baseline={(current bounding box.center)}]
        \draw (-0.6,0.6) -- (0,0.4);
        \draw (-0.6,-0.6) -- (0,-0.4);
        \draw (0.6,0.6) -- (0,0.4);
        \draw (0.6,-0.6) -- (0,-0.4);
        \draw (0,-0.4) -- (0,0.4);
        \draw (-0.8,-0.6) node {$\alpha$};
        \draw (-0.8,0.6) node {$\beta$};
        \draw (0.8,-0.6) node {$\eta$};
        \draw (0.8,0.6) node {$\gamma$};
        \draw (0.2,0) node {$\nu$};
    \end{tikzpicture}\;.
\end{equation}
The $F$-move describes the basis transformation of internal states within a fusion tree. 
In a unitary theory, $[F^{\alpha\beta\gamma}_{\eta}]_{\mu\nu}$ must be a unitary matrix, and it should satisfy the pentagon and hexagon equation together with the $R$-move.
\end{itemize}

Given any fusion tree state, the quantum trace $\widetilde{\mathrm{Tr}}$ is defined by connecting the outgoing lines with their corresponding incoming lines, thereby closing the diagram by anyon loops.
It follows from setting $\gamma^\prime = \gamma = 1$ (the trivial anyon) in \eqref{eqn:charge_conservation} that
\begin{equation} \label{eqn:anyon_loop}
    \widetilde{\mathrm{Tr}}[\begin{tikzpicture}[baseline = {(current bounding box.center)}]
        \draw (-0.5,0.4) .. controls (-0.2,0) and (0.2,0) .. (0.5,0.4);
        \draw (-0.5,-0.4) .. controls (-0.2,0) and (0.2,0) .. (0.5,-0.4);
        \draw (-0.7,0.4) node {$\alpha$};
        \draw (0.7,0.4) node {$\alpha$};
        \draw (-0.7,-0.4) node {$\alpha$};
        \draw (0.7,-0.4) node {$a$};
    \end{tikzpicture}] = \begin{tikzpicture}[baseline = {(current bounding box.center)}]
        \draw (-0.5,0.4) .. controls (-0.2,0) and (0.2,0) .. (0.5,0.4);
        \draw (-0.5,-0.4) .. controls (-0.2,0) and (0.2,0) .. (0.5,-0.4);
        \draw (-0.5,0.4) .. controls (-0.6,0.6) and (-0.73, 0.6) .. (-0.75,0) .. controls (-0.73, -0.6) and (-0.6,-0.6) .. (-0.5,-0.4);
        \draw (0.5,0.4) .. controls (0.6,0.6) and (0.73, 0.6) .. (0.75,0) .. controls (0.73, -0.6) and (0.6,-0.6)  .. (0.5,-0.4);
        \draw (-0.25,0.4) node {$\alpha$};
    \end{tikzpicture} =
    \begin{tikzpicture}[baseline = {(current bounding box.center)}]
        \draw[fill=none](0,0) circle (0.4);
        \draw (-0.6, 0) node {$\alpha$};
    \end{tikzpicture} = d_\alpha\;.
\end{equation}
Therefore the actual trace $\mathrm{Tr}$ of a state should be defined by
\begin{equation} \label{eqn:example_trace}
    \mathrm{Tr}[\begin{tikzpicture}[baseline = {(current bounding box.center)}]
        \draw (-0.5,0.4) .. controls (-0.2,0) and (0.2,0) .. (0.5,0.4);
        \draw (-0.5,-0.4) .. controls (-0.2,0) and (0.2,0) .. (0.5,-0.4);
        \draw (-0.7,0.4) node {$\alpha$};
        \draw (0.7,0.4) node {$\alpha$};
        \draw (-0.7,-0.4) node {$\alpha$};
        \draw (0.7,-0.4) node {$\alpha$};
    \end{tikzpicture}] = \frac{1}{d_\alpha} \widetilde{\mathrm{Tr}}[\begin{tikzpicture}[baseline = {(current bounding box.center)}]
        \draw (-0.5,0.4) .. controls (-0.2,0) and (0.2,0) .. (0.5,0.4);
        \draw (-0.5,-0.4) .. controls (-0.2,0) and (0.2,0) .. (0.5,-0.4);
        \draw (-0.7,0.4) node {$\alpha$};
        \draw (0.7,0.4) node {$\alpha$};
        \draw (-0.7,-0.4) node {$\alpha$};
        \draw (0.7,-0.4) node {$a$};
    \end{tikzpicture}] =\frac{1}{d_\alpha}\ \begin{tikzpicture}[baseline = {(current bounding box.center)}]
        \draw (-0.5,0.4) .. controls (-0.2,0) and (0.2,0) .. (0.5,0.4);
        \draw (-0.5,-0.4) .. controls (-0.2,0) and (0.2,0) .. (0.5,-0.4);
        \draw (-0.5,0.4) .. controls (-0.6,0.6) and (-0.73, 0.6) .. (-0.75,0) .. controls (-0.73, -0.6) and (-0.6,-0.6) .. (-0.5,-0.4);
        \draw (0.5,0.4) .. controls (0.6,0.6) and (0.73, 0.6) .. (0.75,0) .. controls (0.73, -0.6) and (0.6,-0.6)  .. (0.5,-0.4);
        \draw (-0.25,0.4) node {$\alpha$};
    \end{tikzpicture} = 1\;,
\end{equation}
where the factor $1/d_\alpha$ accounts for the normalization properly.

\subsection{The protocol for moving \texorpdfstring{$D$}{D} and \texorpdfstring{$E$}{E} anyons} \label{appendix:movement_D_E}

\begin{figure*}
    \centering
    \includegraphics[width = 0.95\textwidth]{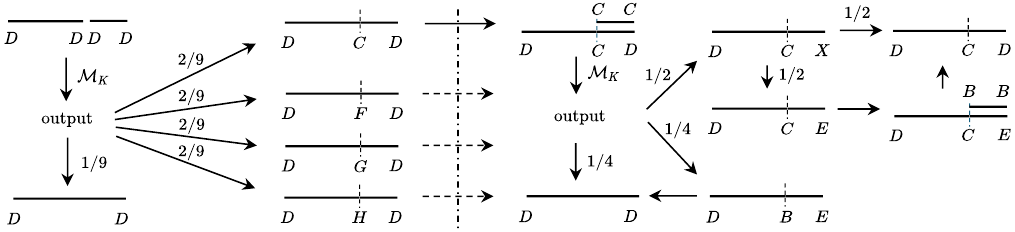}
    \caption{Protocol for moving a type $D$ (non-Abelian) anyon one site to the right.}
    \label{fig:anyon_movement_D}
\end{figure*}

Since $D$ and $E$ anyons have similar fusion rules, we explicitly demonstrate the protocol for moving a $D$ anyon one site to the right, as shown in Figure~\ref{fig:anyon_movement_D}. The corresponding protocol for the $E$ anyon is developed by exchanging $D$ and $E$ labels. The protocol consists of two subroutines, separated by the vertical dashed line in the figure.

We start at the uppermost left part of the first subroutine (left panel), where we extend the long $D$-line by appending the shortest horizontal $D$ ribbon to move the $D$ anyon located at its right end. According to the fusion rule $D \times D = A + C + F + G + H$ and the quantum dimensions $d_A = 1$, $d_C = d_F = d_G = d_H = 2$, the $\mathcal{M}_K$ measurement has five possible outcomes with the corresponding probabilities labeled in the left panel of the figure.

If the outcome is $A$, with probability $1/9$, we successfully move the $D$ anyon. If the outcome is $C$, with probability $2/9$, we proceed to the second subroutine (right panel), where we try to annihilate the intermediate $C$ anyon by applying the shortest horizontal $C$ ribbon. The $\mathcal{M}_K$ measurement has three outcomes: if the outcome is $A$ or $B$, i.e., the Abelian anyons, the configuration can always be corrected to the desired one, as shown in the lower part of the right panel. When the outcome is $C$, with probability $1/2$, the configuration can be corrected to start over the right subroutine, as shown in the upper part of the right panel.

If the fusion outcome of the first subroutine is $F$, $G$, or $H$, one can proceed to the second subroutine by replacing all $C$ labels with $F$, $G$, or $H$ labels, indicated by the dashed arrows in the figure. Therefore, for $n$ rounds of $\mathcal{M}_K$ measurement, the failure probability is $(8/9) \cdot (1/2)^{n-1}$, leading to a success probability of $1 - (8/9) \cdot (1/2)^{n-1}$ as described in the main text.

\subsection{Calculation of the amplitude \texorpdfstring{$I_{x;z,w}$}{Ixyz}}
\label{appendix:Anyon_Interferometry}
We derive the amplitude $I_{x;z,w}$ for the anyon interferometry protocol shown in Figure~\ref{fig:FT_measurements}(a). By isotopic invariance, we deform the $x$ anyon worldlines in the second subfigure to obtain:
\begin{equation}
    \begin{tikzpicture}[baseline = {(current bounding box.center)}]
        \draw (-0.5,1.2) .. controls (-0.7,1.1) and (-1.0,0.9) .. (-1.0,0.6) .. controls (-1.0,0.3) and (-0.7,0.1) .. (-0.5,0) .. controls (-0.3,0.1) and (0.1,0.4) .. (-0.11,0.78);
        \draw (-0.25,0.92) .. controls (-0.28,1.0) and (-0.38,1.1) .. (-0.5,1.2);
        \draw (-0.5,1.2) -- (-0.5,1.5);
        \draw (0.5,1.2) .. controls (0.5,0.5) and (0.3,0) .. (0,0) .. controls (-0.1,0.1) .. (-0.16,0.16);
        \draw (-0.28,0.28) .. controls (-0.5,0.6) and (-0.22,0.85).. (-0.18,0.85) .. controls (-0.05,0.95) and (0,1.1) .. (0.05,1.2);
        \draw (-1.2,0.7) node {$z$};
        \draw (0.7,0.7) node {$x$};
        \draw (-0.7,1.5) node {$w$};
    \end{tikzpicture}
    =\begin{tikzpicture}[baseline = {(current bounding box.center)}]
        \draw (-0.5,1.2) .. controls (-0.7,1.1) and (-1.0,0.9) .. (-1.0,0.6) .. controls (-1.0,0.3) and (-0.7,0.1) .. (-0.5,0) .. controls (-0.4,0.1) .. (-0.4,0.2);
        \draw (-0.4,0.2) -- (-0.4,0.5);
        \draw (-0.5,1.2) -- (-0.5,1.5);
        \draw (0.5,1.2) .. controls (0.5,0.5) and (0.3,0) .. (0,0) .. controls (-0.1,0.1) .. (-0.1,0.2);
        \draw (-0.1,0.2) -- (-0.1,0.5);
        \draw (-0.4,0.5) .. controls (0,0.7) and (0,0.9) .. (-0.18,0.98);
        \draw (-0.1,0.5) .. controls (-0.12,0.51) .. (-0.16,0.56);
        \draw (-0.3,0.66) .. controls (-0.5,0.75) and (-0.5,0.95) .. (0.05,1.2);
        \draw (-0.32,1.09) .. controls (-0.35,1.1) .. (-0.5,1.2);
        \draw (-1.2,0.7) node {$z$};
        \draw (0.7,0.7) node {$x$};
        \draw (-0.7,1.5) node {$w$};
        \draw (-0.6,0.5) node {$z$};
        \draw (0.1,0.5) node {$x$};
    \end{tikzpicture}
    =\sum_{w^\prime}N^{w^\prime}_{zx} \sqrt{\frac{d_{w^\prime}}{d_zd_x}} 
        \begin{tikzpicture}[baseline = {(current bounding box.center)}]
        \draw (-0.5,1.2) .. controls (-0.7,1.1) and (-1.0,0.9) .. (-1.0,0.6) .. controls (-1.0,0.3) and (-0.7,0.1) .. (-0.5,0) .. controls (-0.4,0.03) .. (-0.25,0.15);
        \draw (-0.25,0.15) -- (-0.25,0.4);
        \draw (-0.5,1.2) -- (-0.5,1.5);
        \draw (0.5,1.2) .. controls (0.5,0.5) and (0.3,0) .. (0,0) .. controls (-0.1,0.03) .. (-0.25,0.15);
        \draw (-0.4,0.5) .. controls (0,0.7) and (0,0.9) .. (-0.18,0.98);
        \draw (-0.1,0.5) .. controls (-0.12,0.51) .. (-0.16,0.56);
        \draw (-0.3,0.66) .. controls (-0.5,0.75) and (-0.5,0.95) .. (0.05,1.2);
        \draw (-0.32,1.09) .. controls (-0.35,1.1) .. (-0.5,1.2);
        \draw (-0.25,0.4) -- (-0.4,0.5);
        \draw (-0.25,0.4) -- (-0.1,0.5);
        \draw (-1.0,0.1) node {$z$};
        \draw (0.5,0.1) node {$x$};
        \draw (-0.7,1.5) node {$w$};
        \draw (-0.6,0.6) node {$z$};
        \draw (0.1,0.6) node {$x$};
        \draw (0,0.27) node {$w^\prime$};
    \end{tikzpicture} = \sum_{w^\prime} N^{w^\prime}_{zx} \sqrt{\frac{d_{w^\prime}}{d_zd_x}} R^{zx}_{w^\prime} R^{xz}_{w^\prime} 
    \begin{tikzpicture}[baseline = {(current bounding box.center)}]
        \draw (-0.5,1.2) .. controls (-0.7,1.1) and (-1.0,0.9) .. (-1.0,0.6) .. controls (-1.0,0.3) and (-0.7,0.1) .. (-0.5,0) .. controls (-0.4,0.03) .. (-0.25,0.15);
        \draw (-0.25,0.15) -- (-0.25,1.1);
        \draw (-0.5,1.2) -- (-0.5,1.5);
        \draw (0.5,1.2) .. controls (0.5,0.5) and (0.3,0) .. (0,0) .. controls (-0.1,0.03) .. (-0.25,0.15);
        \draw (-0.25,1.1) .. controls (-0.35,1.15) .. (-0.5,1.2);
        \draw (-0.25,1.1) -- (0,1.2);
        \draw (-0.45,1) node {$z$};
        \draw (-0.05,1) node {$x$};
        \draw (-0.7,1.5) node {$w$};
        \draw (-1.0,0.1) node {$z$};
        \draw (0.5,0.1) node {$x$};
        \draw (0,0.6) node {$w^\prime$};
    \end{tikzpicture}\;.
\end{equation}
In the first equality, we apply isotopic invariance to deform the diagram. The second equality utilizes the resolution of the identity \eqref{eqn:resolution_identity}, while the third equality applies the $R$-move \eqref{eqn:R_move} twice. Since the diagram contains a $z$ loop, we introduce a factor of $1/d_z$ in accordance with the relation between the actual trace and the quantum trace (see \eqref{eqn:example_trace}). Applying the $F$-move \eqref{eqn:F_move}, we obtain:
\begin{equation}
    \sum_{w^\prime,  w^{\prime\prime}}\frac{1}{d_z}  N^{w^\prime}_{zx} \sqrt{\frac{d_{w^\prime}}{d_zd_x}} R^{zx}_{w^\prime} R^{xz}_{w^\prime}[F^{zxx}_z]_{w^\prime w^{\prime\prime}}\  
    \begin{tikzpicture}[baseline = {(current bounding box.center)}]
        \draw (0.5,1.2) .. controls (0.4,0.2) and (0.1,0.2) .. (0,1.2);
        \draw (0.1,0.6) .. controls (-0.1,0.5) and (-0.3,0.5) .. (-0.3,0.7);
        \draw (-0.5,1.0) .. controls (-0.55,1.1) and (-0.58,1.15) .. (-0.6,1.2);
        \draw (-0.5,1.0) .. controls (-0.58,0.85) and (-0.55,0.67) .. (-0.3,0.7);
        \draw (-0.3,0.7) .. controls (-0.2,0.75) and (-0.25,1.05) .. (-0.5,1.0);
        \draw (0.5,1.4) node {$x$};
        \draw (0,1.4) node {$x$};
        \draw (-0.6,1.4) node {$w$};
        \draw (-0.5,0.55) node {$z$};
        \draw (-0.3,1.15) node {$z$};
        \draw (-0.1,0.35) node {$w^{\prime\prime}$};
    \end{tikzpicture} = \sum_{w^\prime} N^{w^\prime}_{zx} \sqrt{\frac{d_{w^\prime}}{d_z d_x d_w}} R^{zx}_{w^\prime} R^{xz}_{w^\prime} [F^{zxx}_z]_{w^\prime w}   \begin{tikzpicture}[baseline = {(current bounding box.center)}]
        \draw (0.5,1.2) .. controls (0.4,0.2) and (0.1,0.2) .. (0,1.2);
        \draw (0.03,0.9) .. controls (-0.1,0.5) and (-0.3,0.5) .. (-0.4,1.2);
        \draw (0.5,1.4) node {$x$};
        \draw (0,1.4) node {$x$};
        \draw (-0.5,1.4) node {$w$};
    \end{tikzpicture}\;,
\end{equation}
where we employ charge conservation \eqref{eqn:charge_conservation} to resolve the $z$ loop. Finally, fusing the $w$ anyon with the $x$ anyon and applying \eqref{eqn:charge_conservation} again results in:
\begin{equation}
    \begin{tikzpicture}[baseline = {(current bounding box.center)}]
        \draw (-0.5,1.2) .. controls (-0.7,1.1) and (-1.0,0.9) .. (-1.0,0.6) .. controls (-1.0,0.3) and (-0.7,0.1) .. (-0.5,0) .. controls (-0.3,0.1) and (0.1,0.4) .. (-0.11,0.78);
        \draw (-0.25,0.92) .. controls (-0.28,1.0) and (-0.38,1.1) .. (-0.5,1.2);
        \draw (-0.5,1.2) -- (-0.5,1.5);
        \draw (0.5,1.2) .. controls (0.5,0.5) and (0.3,0) .. (0,0) .. controls (-0.1,0.1) .. (-0.16,0.16);
        \draw (-0.28,0.28) .. controls (-0.5,0.6) and (-0.22,0.85).. (-0.18,0.85) .. controls (-0.05,0.95) and (0,1.1) .. (0.05,1.2);
        \draw (-1.2,0.7) node {$z$};
        \draw (0.7,0.7) node {$x$};
        \draw (-0.7,1.5) node {$w$};
    \end{tikzpicture}
    \to \ 
    \sum_{w^\prime} N^{w^\prime}_{zx} \sqrt{\frac{d_{w^\prime}}{d_{z}d_x}} R^{zx}_{w^\prime} R^{xz}_{w^\prime} [F^{zxx}_z]_{w^\prime w}\ \begin{tikzpicture}[baseline = {(current bounding box.center)}]
        \draw (0.5,1.2) .. controls (0.4,0) and (-0.1,0) .. (-0.2,1.2);
        \draw (0.5,1.4) node {$x$};
        \draw (-0.2,1.4) node {$x$};
    \end{tikzpicture}\;,
\end{equation}
which yields the amplitude in Figure~\ref{fig:FT_measurements}(a):
\begin{equation} \label{eqn:I_xzw}
    I_{x;z,w} = \sum_{w^\prime} N^{w^\prime}_{zx} \sqrt{\frac{d_{w^\prime}}{d_{z}d_x}} R^{zx}_{w^\prime} R^{xz}_{w^\prime} [F^{zxx}_z]_{w^\prime w}\;.
\end{equation}

\subsection{The phase factor from local anyon interferometer errors}
\label{appendix:Phase_Factor_Interferometry}
In this section, we derive the phase factor $I_{z,(a,b,c),w}$ that arises from local anyon interferometer errors, as introduced in \eqref{eqn:local_interferometry}. By applying the $F$-move \eqref{eqn:F_move}, followed by the resolution of identity \eqref{eqn:resolution_identity} and the $R$-move \eqref{eqn:R_move}, we obtain:
\begin{equation}
    \begin{tikzpicture}[baseline = {(current bounding box.center)}]
        \draw (0,0) -- (0,0.21); 
        \draw (0,0.36) -- (0,2.0); 
        \draw (-0.5,0.88).. controls (-0.75,0.83) and (-0.9,0.6) .. (-0.7,0.45) ..controls (-0.4,0.22) and (0.4, 0.22) .. (0.7,0.45) .. controls (0.9,0.6) and (0.75,0.9) .. (0.08,0.95);
        \draw (-0.08,0.95) .. controls (-0.15,0.95) and (-0.3,0.94) .. (-0.38,0.91);
        \draw (-0.7,0.86) .. controls (-0.65,1.22) and (-0.45,1.05) .. (-0.43,0.85) .. controls (-0.45, 0.55) and (-0.62,0.63).. (-0.65,0.73);
        \draw (-0.59,1.09) .. controls (-0.6,1.20) and (-0.65,1.36) .. (-0.8,1.45);
        \draw (-0.80,0.54) .. controls (-1.1,0.9) and (-1,1.3) .. (-0.8,1.45);
        \draw (-0.8,1.45) .. controls (-0.8,1.60) and (-0.75,1.8) .. (-0.7,1.9);
        \node at (0.2,0) {\(x\)};
        \node at (0.2,2.0) {\(x\)};
        \node at (-0.5,0.15) {\(z\)};
        \node at (-0.32,0.60) {\(a\)};
        \node at (-0.47,1.27) {\(b\)};
        \node at (-1.13,1.1) {\(c\)};
        \node at (-0.5,2) {\(w\)};
    \end{tikzpicture} = \sum_{b^\prime} \frac{[F^{aac}_w]_{bb^\prime}}{d_a}\ 
    \begin{tikzpicture}[baseline = {(current bounding box.center)}]
        \draw (0,0) -- (0,0.21); 
        \draw (0,0.36) -- (0,2.0); 
        \draw (-0.38,0.90).. controls (-0.75,0.83) and (-0.9,0.6) .. (-0.7,0.45) ..controls (-0.4,0.22) and (0.4, 0.22) .. (0.7,0.45) .. controls (0.9,0.6) and (0.75,0.9) .. (0.08,0.95);
        \draw (-0.08,0.95) .. controls (-0.12,0.95) and (-0.23,0.94) .. (-0.28,0.92);
        \draw (-0.80,0.54) .. controls (-1.1,0.8) and (-1,1.1) .. (-0.8,1.25);
        \draw (-0.45,1.45) .. controls (-0.45,1.60) and (-0.4,1.8) .. (-0.3,2);
        \draw (-0.45,1.45) .. controls (-0.35,1.4) and (-0.2,0.6) .. (-0.54,0.78);
        \draw (-0.64,0.88) .. controls (-0.72,0.92) and (-0.78,1.05) .. (-0.8, 1.25);
        \draw (-0.8,1.25) .. controls (-0.65,1.4) and (-0.55,1.4) .. (-0.45,1.45);
        \node at (0.2,0) {\(x\)};
        \node at (0.2,2.0) {\(x\)};
        \node at (-0.5,0.15) {\(z\)};
        \node at (-0.65,1.6) {$b^\prime$};
        \node at (-0.35,0.6) {$a$};
        \node at (-1.13,1.1) {\(c\)};
        \node at (-0.6,2) {\(w\)};
    \end{tikzpicture} = \sum_{b^\prime,c^\prime}  \frac{N^{c^\prime}_{az}}{d_a} \sqrt{\frac{d_{c^\prime}}{d_a d_z}} R^{az}_{c^\prime} R^{za}_{c^\prime}[F^{aac}_w]_{bb^\prime} \ 
    \begin{tikzpicture}[baseline = {(current bounding box.center)}]
        \draw (0,0) -- (0,0.21); 
        \draw (0,0.36) -- (0,2.0); 
        \draw (-0.08,0.95).. controls (-0.75,0.9) and (-0.9,0.6) .. (-0.7,0.45) ..controls (-0.4,0.22) and (0.4, 0.22) .. (0.7,0.45) .. controls (0.9,0.6) and (0.75,0.9) .. (0.08,0.95);
        \draw (-0.80,0.54) .. controls (-1.1,0.8) and (-1,1.1) .. (-0.8,1.25);
        \draw (-0.45,1.45) .. controls (-0.45,1.60) and (-0.4,1.8) .. (-0.3,2);
        \draw (-0.45,1.45) .. controls (-0.35,1.4) and (-0.2,1.2) .. (-0.2,0.95);
        \draw (-0.8,1.25) .. controls (-0.75,1.1) and (-0.7,0.9) .. (-0.6,0.84);
        \draw (-0.8,1.25) .. controls (-0.65,1.4) and (-0.55,1.4) .. (-0.45,1.45);
        \node at (0.2,0) {\(x\)};
        \node at (0.2,2.0) {\(x\)};
        \node at (-0.5,0.15) {\(z\)};
        \node at (-0.65,1.6) {$b^\prime$};
        \node at (-0.59,1.11) {$a$};
        \node at (-0.15,1.35) {$a$};
        \node at (-0.33,0.73) {$c^\prime$};
        \node at (-0.60,0.61) {$z$};
        \node at (-1.13,1.1) {\(c\)};
        \node at (-0.6,2) {\(w\)};
    \end{tikzpicture} \;.
\end{equation}
By applying the $F$-move twice on the two $a$-lines, this expression transforms into:
\begin{equation}
    \sum_{b^\prime,c^\prime,d^\prime, e^\prime} \frac{N^{c^\prime}_{az}}{d_a} \sqrt{\frac{d_{c^\prime}}{d_a d_z}} R^{az}_{c^\prime} R^{za}_{c^\prime} [F^{aac}_w]_{bb^\prime} [F^{cb^\prime c^\prime}_{z}]_{ad^\prime} [F^{b^\prime w z}_{c^\prime}]_{ae^\prime}\  
    \begin{tikzpicture}[baseline = {(current bounding box.center)}]
        \draw (0,0) -- (0,0.21); 
        \draw (0,0.36) -- (0,2.0); 
        \draw (-0.08,1.3) .. controls (-0.95,1.2) and (-1.2,0.75) .. (-0.8,0.45) .. controls (-0.5,0.2) and (0.5,0.2) .. (0.8,0.45) .. controls (1.2,0.75) and (0.95,1.2) .. (0.08,1.3); 
        \draw (-0.95,0.6) .. controls (-1.25, 0.75) and (-1.2,1.1) .. (-0.93,0.95);
        \draw (-0.65,1.15) .. controls (-0.75,1.6) and (-0.5,1.45) .. (-0.45,1.23);
        \draw (-0.4,2) .. controls (-0.27,1.7) and (-0.25,1.6) .. (-0.18,1.28);
        \node at (0.2,0) {\(x\)};
        \node at (0.2,2.0) {\(x\)};
        \node at (-0.7,0.15) {\(z\)};
        \node at (-0.65,1.6) {$b^\prime$};
        \node at (-0.5,1.0) {$c^\prime$};
        \node at (-0.85,1.25) {$d^\prime$};
        \node at (-1.3,1) {$c$};
        \node at (-0.8,0.75) {$z$};
        \node at (-0.25,1.05) {$e^\prime$};
        \node at (-0.6,2) {\(w\)};
    \end{tikzpicture}\;.
\end{equation}
Applying the charge conservation condition \eqref{eqn:charge_conservation} twice gives the final result:
\begin{equation}
    \begin{tikzpicture}[baseline = {(current bounding box.center)}]
        \draw (0,0) -- (0,0.21); 
        \draw (0,0.36) -- (0,2.0); 
        \draw (-0.5,0.88).. controls (-0.75,0.83) and (-0.9,0.6) .. (-0.7,0.45) ..controls (-0.4,0.22) and (0.4, 0.22) .. (0.7,0.45) .. controls (0.9,0.6) and (0.75,0.9) .. (0.08,0.95);
        \draw (-0.08,0.95) .. controls (-0.15,0.95) and (-0.3,0.94) .. (-0.38,0.91);
        \draw (-0.7,0.86) .. controls (-0.65,1.22) and (-0.45,1.05) .. (-0.43,0.85) .. controls (-0.45, 0.55) and (-0.62,0.63).. (-0.65,0.73);
        \draw (-0.59,1.09) .. controls (-0.6,1.20) and (-0.65,1.36) .. (-0.8,1.45);
        \draw (-0.80,0.54) .. controls (-1.1,0.9) and (-1,1.3) .. (-0.8,1.45);
        \draw (-0.8,1.45) .. controls (-0.8,1.60) and (-0.75,1.8) .. (-0.7,1.9);
        \node at (0.2,0) {\(x\)};
        \node at (0.2,2.0) {\(x\)};
        \node at (-0.5,0.15) {\(z\)};
        \node at (-0.32,0.60) {\(a\)};
        \node at (-0.47,1.27) {\(b\)};
        \node at (-1.13,1.1) {\(c\)};
        \node at (-0.5,2) {\(w\)};
    \end{tikzpicture}
    \ \to \ 
    \sum_{b^\prime, c^\prime} N^{c^\prime}_{az} N^{z}_{b^\prime c^\prime} \sqrt{\frac{d_{b^\prime}d_c}{d_a}} \frac{d_{c^\prime}}{d_z d_a} R^{az}_{c^\prime} R^{za}_{c^\prime} [F^{aac}_w]_{bb^\prime} [F^{cb^\prime c^\prime}_z]_{az} [F^{b^\prime w z}_{c^\prime}]_{az}\  \begin{tikzpicture}[baseline = {(current bounding box.center)}]
        \draw (0,0) -- (0,0.21); 
        \draw (0,0.36) -- (0,2.0); 
        \draw (-0.08,0.95).. controls (-0.75,0.9) and (-0.9,0.6) .. (-0.7,0.45) ..controls (-0.4,0.22) and (0.4, 0.22) .. (0.7,0.45) .. controls (0.9,0.6) and (0.75,0.9) .. (0.08,0.95);
        \draw (-0.65,0.81) .. controls (-0.9,1.00) and (-0.95,1.6) .. (-0.7,1.9);
        \node at (0.2,0) {\(x\)};
        \node at (0.2,2.0) {\(x\)};
        \node at (-0.5,0.15) {\(z\)};
        \node at (-0.5,2) {\(w\)};
    \end{tikzpicture}\;,
\end{equation}
which yields the $x$-independent phase factor in \eqref{eqn:local_interferometry}:
\begin{equation}
    I_{z,(a,b,c),w} = \sum_{b^\prime, c^\prime} N^{c^\prime}_{az} N^{z}_{b^\prime c^\prime} \sqrt{\frac{d_{b^\prime}d_c}{d_a}} \frac{d_{c^\prime}}{d_z d_a} R^{az}_{c^\prime} R^{za}_{c^\prime} [F^{aac}_w]_{bb^\prime} [F^{cb^\prime c^\prime}_z]_{az} [F^{b^\prime w z}_{c^\prime}]_{az}\;.
\end{equation}

\subsection{The remote measurements}
\label{appendix:FT_Measurements}
This section provides a detailed examination of the anyon interferometry protocols depicted in Figures~\ref{fig:FT_measurements}(b) and \ref{fig:FT_measurements}(c). The fusion rules, governed by the multiplicity $N^{\gamma}_{\alpha\beta}$, along with the associated $F$-symbols and $R$-symbols for $\mathcal{D}(S_3)$, are comprehensively outlined in \cite{Cui2015QIP_UQC_Weakly}. As demonstrated in the supplementary Mathematica notebook, we confirm that these categorical data satisfy both the Pentagon and Hexagon consistency equations. Readers may also refer to the notebook for the amplitudes discussed in this section.

Now, we compute the amplitude for the anyon interferometer that implements $\mathcal{M}_A$, as depicted in Figure~\ref{fig:FT_measurements}(b). 
The interferometer is performed using a pair of $D$ anyons.
By steps analogous to those in Figure~\ref{fig:FT_measurements}(a) and Appendix~\ref{appendix:Anyon_Interferometry}, we resolve the $D$-loop by introducing the amplitude $I_{x;D,w}$ expressed as:
\begin{equation}
    \begin{tikzpicture}[baseline = {(current bounding box.center)}]
         \draw (0,0) -- (0,0.5);
        \draw (0,0.5) -- (-0.15,0.65);
        \draw (-0.25,0.75) -- (-0.6,1.1);
        \draw (-0.6,1.1) -- (-1,1.5);
        \draw (-0.6,1.1) -- (-0.2,1.5);
        \draw (0,0.5) -- (0.6,1.1);
        \draw (0.6,1.1) -- (1,1.5);
        \draw (0.6,1.1) -- (0.2,1.5);
        \draw (-0.48,0.89) .. controls (-0.82,0.67) and (-0.96,0.46) .. (-0.75,0.42) .. controls (-0.55,0.4) and (-0.33,0.53) .. (-0.16,0.72) .. controls (-0.03,0.87) and (0.05,1.1) .. (-0.35,0.95);
        \draw (-0.85,0.5) .. controls (-1.4,0.75) and (-1.3,1.1) .. (-1.3,1.1);
        \draw (-0.2,0.4) node {$x$};
        \draw (0.5,0.7) node {$y$};
        \draw (0.2,0) node {$G$};
        \draw (-1,1.75) node {$D$};
        \draw (-0.3,1.75) node {$D$};
        \draw (0.3,1.75) node {$D$};
        \draw (1,1.75) node {$D$};
        \draw (-0.75,0.15) node {$D$};
        \draw (-1.35,1.3) node {$w$};
    \end{tikzpicture} = \frac{I_{x;D,w}}{\sqrt{d_w}}\ 
    \begin{tikzpicture}[baseline = {(current bounding box.center)}]
        \draw (0,0) -- (0,0.5);
        \draw (0,0.5) -- (-0.6,1.1);
        \draw (-0.6,1.1) -- (-1,1.5);
        \draw (-0.6,1.1) -- (-0.2,1.5);
        \draw (0,0.5) -- (0.6,1.1);
        \draw (0.6,1.1) -- (1,1.5);
        \draw (0.6,1.1) -- (0.2,1.5);
        \draw (-0.3,0.8) .. controls (-0.7,0.7) and (-1.05,0.85) .. (-1.2,1.1);
        \draw (-0.3,0.5) node {$x$};
        \draw (0.5,0.7) node {$y$};
        \draw (-1.3,1.3) node {$w$};
        \draw (0.2,0) node {$G$};
        \draw (-1,1.75) node {$D$};
        \draw (-0.3,1.75) node {$D$};
        \draw (0.3,1.75) node {$D$};
        \draw (1,1.75) node {$D$};
    \end{tikzpicture}\;.
\end{equation}
By referring to \eqref{eqn:I_xzw}, and considering $x=A,G$, the non-zero amplitudes are:
\begin{equation}
    I_{A;D,A} = 1,\quad I_{G;D,G} = \omega^2\;.
\end{equation}
This implies that when two $D$ anyons are fused following the braiding of one with $x$, the outcome $w=A$ or $w=G$ indicates whether $x=A$ or $x=G$. As illustrated in the final panel of Figure~\ref{fig:FT_measurements}(a), the next step is to fuse the $w$ anyon with an appropriate anyon to return the fusion tree to $V^{DDDD}_G$, ensuring that the logical information stored in the system remains unaffected. If $w=A$, no further action is required. If $w=G$, which implies $x=G$, we fuse the $w$ anyon with the upper left-most $D$ anyon in the tree. The result of this fusion can be computed by performing an $F$-move, as follows:
\begin{equation}
    \begin{tikzpicture}[baseline = {(current bounding box.center)}]
        \draw (0,0) -- (0,0.5);
        \draw (0,0.5) -- (-0.8,1.3);
        \draw (-0.8,1.3) -- (-1.2,1.7);
        \draw (-0.6,1.1) -- (-0.3,1.4);
        \draw (-0.3,1.4) -- (-0.6,1.7);
        \draw (-0.3,1.4) -- (0,1.7);
        \draw (0,0.5) -- (0.8,1.3);
        \draw (0.8,1.3) -- (0.4,1.7);
        \draw (0.8,1.3) -- (1.2,1.7);
        \draw (-0.5,0.7) node {$G$};
        \draw (0.5,0.7) node {$y$};
        \draw (-0.25,1.1) node {$G$};
        \draw (-1.3,1.95) node {$G$};
        \draw (0.2,0) node {$G$};
        \draw (-0.6,1.95) node {$D$};
        \draw (0,1.95) node {$D$};
        \draw (0.4,1.95) node {$D$};
        \draw (1.2,1.95) node {$D$};
    \end{tikzpicture} = \sum_{x^\prime} [(F^{GDD}_G)^{-1}]_{Gx^\prime} \ 
    \begin{tikzpicture}[baseline = {(current bounding box.center)}]
        \draw (0,0) -- (0,0.5);
        \draw (0,0.5) -- (-0.8,1.3);
        \draw (-0.8,1.3) -- (-1.2,1.7);
        \draw (-0.6,1.1) -- (-0.3,1.4);
        \draw (-0.9,1.4) -- (-0.6,1.7);
        \draw (-0.3,1.4) -- (0,1.7);
        \draw (0,0.5) -- (0.8,1.3);
        \draw (0.8,1.3) -- (0.4,1.7);
        \draw (0.8,1.3) -- (1.2,1.7);
        \draw (-0.5,0.7) node {$G$};
        \draw (0.5,0.7) node {$y$};
        \draw (-0.95,1.1) node {$x^\prime$};
        \draw (-1.3,1.95) node {$G$};
        \draw (0.2,0) node {$G$};
        \draw (-0.6,1.95) node {$D$};
        \draw (0,1.95) node {$D$};
        \draw (0.4,1.95) node {$D$};
        \draw (1.2,1.95) node {$D$};
    \end{tikzpicture} = \frac{1}{\sqrt{2}}\ \begin{tikzpicture}[baseline = {(current bounding box.center)}]
        \draw (0,0) -- (0,0.5);
        \draw (0,0.5) -- (-0.8,1.3);
        \draw (-0.8,1.3) -- (-1.2,1.7);
        \draw (-0.6,1.1) -- (-0.3,1.4);
        \draw (-0.9,1.4) -- (-0.6,1.7);
        \draw (-0.3,1.4) -- (0,1.7);
        \draw (0,0.5) -- (0.8,1.3);
        \draw (0.8,1.3) -- (0.4,1.7);
        \draw (0.8,1.3) -- (1.2,1.7);
        \draw (-0.5,0.7) node {$G$};
        \draw (0.5,0.7) node {$y$};
        \draw (-0.95,1.1) node {$D$};
        \draw (-1.3,1.95) node {$G$};
        \draw (0.2,0) node {$G$};
        \draw (-0.6,1.95) node {$D$};
        \draw (0,1.95) node {$D$};
        \draw (0.4,1.95) node {$D$};
        \draw (1.2,1.95) node {$D$};
    \end{tikzpicture} - \frac{1}{\sqrt{2}}\ \begin{tikzpicture}[baseline = {(current bounding box.center)}]
        \draw (0,0) -- (0,0.5);
        \draw (0,0.5) -- (-0.8,1.3);
        \draw (-0.8,1.3) -- (-1.2,1.7);
        \draw (-0.6,1.1) -- (-0.3,1.4);
        \draw (-0.9,1.4) -- (-0.6,1.7);
        \draw (-0.3,1.4) -- (0,1.7);
        \draw (0,0.5) -- (0.8,1.3);
        \draw (0.8,1.3) -- (0.4,1.7);
        \draw (0.8,1.3) -- (1.2,1.7);
        \draw (-0.5,0.7) node {$G$};
        \draw (0.5,0.7) node {$y$};
        \draw (-0.95,1.1) node {$E$};
        \draw (-1.3,1.95) node {$G$};
        \draw (0.2,0) node {$G$};
        \draw (-0.6,1.95) node {$D$};
        \draw (0,1.95) node {$D$};
        \draw (0.4,1.95) node {$D$};
        \draw (1.2,1.95) node {$D$};
    \end{tikzpicture}\;,
\end{equation}
where we have used the fact that every $F^{\alpha\beta\gamma}_{\eta}$ is a unitary matrix. The non-zero matrix elements of $F^{GDD}_G$ are:
\begin{equation}
    [F^{GDD}_G]_{DG} = \frac{1}{\sqrt{2}},\quad [F^{GDD}_G]_{EG} = -\frac{1}{\sqrt{2}}\;.
\end{equation}
Therefore, when the $G$ anyon is fused with the $D$ anyon, the resulting fusion trees, along with their corresponding probabilities, are given as follows:
\begin{equation}
    \mathrm{Prob} = \frac{1}{2}:\ \begin{tikzpicture}[baseline = {(current bounding box.center)}]
        \draw (0,0) -- (0,0.5);
        \draw (0,0.5) -- (-0.6,1.1);
        \draw (-0.6,1.1) -- (-1,1.5);
        \draw (-0.6,1.1) -- (-0.2,1.5);
        \draw (0,0.5) -- (0.6,1.1);
        \draw (0.6,1.1) -- (1,1.5);
        \draw (0.6,1.1) -- (0.2,1.5);
        \draw (-0.5,0.7) node {$G$};
        \draw (0.5,0.7) node {$y$};
        \draw (0.2,0) node {$G$};
        \draw (-1,1.75) node {$D$};
        \draw (-0.3,1.75) node {$D$};
        \draw (0.3,1.75) node {$D$};
        \draw (1,1.75) node {$D$};
    \end{tikzpicture}\;,\quad \mathrm{Prob} = \frac{1}{2}:\ \begin{tikzpicture}[baseline = {(current bounding box.center)}]
        \draw (0,0) -- (0,0.5);
        \draw (0,0.5) -- (-0.6,1.1);
        \draw (-0.6,1.1) -- (-1,1.5);
        \draw (-0.6,1.1) -- (-0.2,1.5);
        \draw (0,0.5) -- (0.6,1.1);
        \draw (0.6,1.1) -- (1,1.5);
        \draw (0.6,1.1) -- (0.2,1.5);
        \draw (-0.5,0.7) node {$G$};
        \draw (0.5,0.7) node {$y$};
        \draw (0.2,0) node {$G$};
        \draw (-1,1.75) node {$E$};
        \draw (-0.3,1.75) node {$D$};
        \draw (0.3,1.75) node {$D$};
        \draw (1,1.75) node {$D$};
    \end{tikzpicture}\;.
\end{equation}
If the fusion outcome is $D$ (the first fusion tree), we return to the original fusion tree space $V^{DDDD}_G$. However, if the outcome is $E$ (the second fusion tree), according to the fusion rules $B \times E \to D$ and $B \times G \to G$, we introduce a pair of $B$ anyons. We then fuse one $B$ anyon with the $E$ anyon and the other with the $G$ anyon, leading to the following relation:
\begin{equation} \label{eqn:B_anyon_correction_MA}
    \begin{tikzpicture}[baseline = {(current bounding box.center)}]
        \draw (0,0) -- (0,0.5);
        \draw (0,0.5) -- (-0.6,1.1);
        \draw (-0.6,1.1) -- (-1,1.5);
        \draw (-0.6,1.1) -- (-0.2,1.5);
        \draw (0,0.5) -- (0.6,1.1);
        \draw (0.6,1.1) -- (1,1.5);
        \draw (0.6,1.1) -- (0.2,1.5);
        \draw (-0.9,1.4) .. controls (-1.3,1.1) and (-1.1,0.4) .. (0,0.2);
        \draw (-0.5,0.7) node {$G$};
        \draw (0.5,0.7) node {$y$};
        \draw (0.2,0) node {$G$};
        \draw (0.2,0.4) node {$G$};
        \draw (-0.9,1.05) node {$E$};
        \draw (-1.1,0.5) node {$B$};
        \draw (-1,1.75) node {$D$};
        \draw (-0.3,1.75) node {$D$};
        \draw (0.3,1.75) node {$D$};
        \draw (1,1.75) node {$D$};
    \end{tikzpicture} = [F^{BDD}_G]_{EG} [F^{BGy}_{G}]_{GG}\ \begin{tikzpicture}[baseline = {(current bounding box.center)}]
        \draw (0,0) -- (0,0.5);
        \draw (0,0.5) -- (-0.6,1.1);
        \draw (-0.6,1.1) -- (-1,1.5);
        \draw (-0.6,1.1) -- (-0.2,1.5);
        \draw (0,0.5) -- (0.6,1.1);
        \draw (0.6,1.1) -- (1,1.5);
        \draw (0.6,1.1) -- (0.2,1.5);
        \draw (-0.5,0.7) node {$G$};
        \draw (0.5,0.7) node {$y$};
        \draw (0.2,0) node {$G$};
        \draw (-1,1.75) node {$D$};
        \draw (-0.3,1.75) node {$D$};
        \draw (0.3,1.75) node {$D$};
        \draw (1,1.75) node {$D$};
    \end{tikzpicture}\;,
\end{equation}
where we apply $F$-moves twice and invoke charge conservation to eliminate the $B$ worldline. Since $[F^{BDD}_G]_{EG} = 1$ and $[F^{BGy}_G]_{GG} = 1$ or $-1$ for $y=A$ or $y=G$, respectively, this process introduces a relative minus sign between the basis states $\ket{GG}$ and $\ket{GA}$ in the subspace $A^\prime = \mathrm{span}\{\ket{GG},\ket{GA}\}$. To resolve this, we continue the interferometry and fusion protocol until another $E$ outcome is obtained, thereby canceling the relative minus sign, as in the procedure for $\mathcal{M}_U$ in \secref{sec:fault_tolerant_meas}. The probability of eventually obtaining another $E$ outcome is exponentially close to $1$ with respect to the number of measurement rounds. Thus, the protocol for realizing $\mathcal{M}_A$ via anyon interferometry is successfully implemented, as shown in Figure~\ref{fig:FT_measurements}(b).

Next, we compute the amplitude for the anyon interferometry that realizes $\tilde{\mathcal{M}}_U$, as depicted in Figure~\ref{fig:FT_measurements}(c). By using a pair of $H$ anyons to perform the interferometry, we resolve the $H$-loop to obtain the following amplitude:
\begin{equation}
    \begin{tikzpicture}[baseline = {(current bounding box.center)}]
         \draw (0,0) -- (0,0.5);
        \draw (0,0.5) -- (-0.15,0.65);
        \draw (-0.25,0.75) -- (-0.6,1.1);
        \draw (-0.6,1.1) -- (-1,1.5);
        \draw (-0.6,1.1) -- (-0.2,1.5);
        \draw (0,0.5) -- (0.6,1.1);
        \draw (0.6,1.1) -- (1,1.5);
        \draw (0.6,1.1) -- (0.2,1.5);
        \draw (-0.48,0.89) .. controls (-0.82,0.67) and (-0.96,0.46) .. (-0.75,0.42) .. controls (-0.55,0.4) and (-0.33,0.53) .. (-0.16,0.72) .. controls (-0.03,0.87) and (0.05,1.1) .. (-0.35,0.95);
        \draw (-0.85,0.5) .. controls (-1.4,0.75) and (-1.3,1.1) .. (-1.3,1.1);
        \draw (-0.2,0.4) node {$x$};
        \draw (0.5,0.7) node {$y$};
        \draw (0.2,0) node {$G$};
        \draw (-1,1.75) node {$D$};
        \draw (-0.3,1.75) node {$D$};
        \draw (0.3,1.75) node {$D$};
        \draw (1,1.75) node {$D$};
        \draw (-0.75,0.15) node {$H$};
        \draw (-1.35,1.3) node {$w$};
    \end{tikzpicture} = \frac{I_{x;H,w}}{\sqrt{d_w}}\ 
    \begin{tikzpicture}[baseline = {(current bounding box.center)}]
        \draw (0,0) -- (0,0.5);
        \draw (0,0.5) -- (-0.6,1.1);
        \draw (-0.6,1.1) -- (-1,1.5);
        \draw (-0.6,1.1) -- (-0.2,1.5);
        \draw (0,0.5) -- (0.6,1.1);
        \draw (0.6,1.1) -- (1,1.5);
        \draw (0.6,1.1) -- (0.2,1.5);
        \draw (-0.3,0.8) .. controls (-0.7,0.7) and (-1.05,0.85) .. (-1.2,1.1);
        \draw (-0.3,0.5) node {$x$};
        \draw (0.5,0.7) node {$y$};
        \draw (-1.3,1.3) node {$w$};
        \draw (0.2,0) node {$G$};
        \draw (-1,1.75) node {$D$};
        \draw (-0.3,1.75) node {$D$};
        \draw (0.3,1.75) node {$D$};
        \draw (1,1.75) node {$D$};
    \end{tikzpicture}\;.
\end{equation}
The non-zero values of $I_{x;H,w}$ are:
\begin{equation}
\begin{aligned}
    &I_{A;H,A} = I_{B;H,A} = I_{G;H,A} = 1\;,\quad I_{D;H,H} = \omega^2\;,\quad I_{E;H,H} = -\omega^2\;,\quad I_{C;H,A} = H_{F;H,A} = I_{H;H,A} = -\frac{1}{2}\;.\\
    &I_{C;H,B} = I_{H;H,B} = -\frac{i\sqrt{3}}{2},\quad I_{F;H,B} = \frac{i\sqrt{3}}{2}\;.
\end{aligned}
\end{equation}
This implies that the fusion outcome $w$ is always Abelian. Thus, according to $w\times G \to G$, we can fuse $w$ with $G$ to return the fusion tree to $V^{DDDD}_G$. The result of this fusion is computed by performing an $F$-move as follows:
\begin{equation} \label{eqn:w_anyon_fusion_MU}
    \begin{tikzpicture}[baseline = {(current bounding box.center)}]
        \draw (0,0) -- (0,0.5);
        \draw (0,0.5) -- (-0.6,1.1);
        \draw (-0.6,1.1) -- (-1,1.5);
        \draw (-0.6,1.1) -- (-0.2,1.5);
        \draw (0,0.5) -- (0.6,1.1);
        \draw (0.6,1.1) -- (1,1.5);
        \draw (0.6,1.1) -- (0.2,1.5);
        \draw (-0.3,0.8) -- (-0.8,0.3);
        \draw (-0.25,0.5) node {$x$};
        \draw (-0.6,0.85) node {$x$};
        \draw (0.5,0.7) node {$y$};
        \draw (-1,0.3) node {$w$};
        \draw (0.2,0) node {$G$};
        \draw (-1,1.75) node {$D$};
        \draw (-0.3,1.75) node {$D$};
        \draw (0.3,1.75) node {$D$};
        \draw (1,1.75) node {$D$};
    \end{tikzpicture} 
    = [F^{wxy}_{G}]_{xG}
    \begin{tikzpicture}[baseline = {(current bounding box.center)}]
        \draw (0,0) -- (0,0.5);
        \draw (0,0.5) -- (-0.6,1.1);
        \draw (-0.6,1.1) -- (-1,1.5);
        \draw (-0.6,1.1) -- (-0.2,1.5);
        \draw (0,0.5) -- (0.6,1.1);
        \draw (0.6,1.1) -- (1,1.5);
        \draw (0.6,1.1) -- (0.2,1.5);
        \draw (0,0.2) -- (-0.8,0.2);
        \draw (-0.5,0.7) node {$x$};
        \draw (0.5,0.7) node {$y$};
        \draw (-1,0.2) node {$w$};
        \draw (0.2,0) node {$G$};
        \draw (0.2,0.4) node {$G$};
        \draw (-1,1.75) node {$D$};
        \draw (-0.3,1.75) node {$D$};
        \draw (0.3,1.75) node {$D$};
        \draw (1,1.75) node {$D$};
    \end{tikzpicture}\;.
\end{equation}
By fusing $w$ with $G$, we realize the following mapping of the fusion tree:
\begin{equation}
    \begin{tikzpicture}[baseline = {(current bounding box.center)}]
        \begin{scope}[shift={(0,0)}]
        \draw (0,0) -- (0,0.5);
        \draw (0,0.5) -- (-0.15,0.65);
        \draw (-0.25,0.75) -- (-0.6,1.1);
        \draw (-0.6,1.1) -- (-1,1.5);
        \draw (-0.6,1.1) -- (-0.2,1.5);
        \draw (0,0.5) -- (0.6,1.1);
        \draw (0.6,1.1) -- (1,1.5);
        \draw (0.6,1.1) -- (0.2,1.5);
        \draw (-0.48,0.89) .. controls (-0.82,0.67) and (-0.96,0.46) .. (-0.75,0.42) .. controls (-0.55,0.4) and (-0.33,0.53) .. (-0.16,0.72) .. controls (-0.03,0.87) and (0.05,1.1) .. (-0.35,0.95);
        \draw (-0.75,0.42) .. controls (-0.8,0.2) and (-0.6,0.1) .. (0,0.2);
        \draw (-0.35,1.1) node {$x$};
        \draw (0.5,0.7) node {$y$};
        \draw (0.2,0) node {$G$};
        \draw (-1,1.75) node {$D$};
        \draw (-0.3,1.75) node {$D$};
        \draw (0.3,1.75) node {$D$};
        \draw (1,1.75) node {$D$};
        \draw (-0.95,0.8) node {$H$};  \draw (-0.92,0.18) node {$w$};
        \draw (-0.2,0.4) node {$G$};
    \end{scope}

    \draw[->] (1.3,0.7) -- (1.7,0.7);

    \begin{scope}[shift={(5.4,0)}]
        \draw (-2.1,0.7) node {$I_{x;H,w} [F^{wxy}_G]_{xG}$};
        \draw (0,0) -- (0,0.5);
        \draw (0,0.5) -- (-0.6,1.1);
        \draw (-0.6,1.1) -- (-1,1.5);
        \draw (-0.6,1.1) -- (-0.2,1.5);
        \draw (0,0.5) -- (0.6,1.1);
        \draw (0.6,1.1) -- (1,1.5);
        \draw (0.6,1.1) -- (0.2,1.5);
        \draw (-0.5,0.7) node {$x$};
        \draw (0.5,0.7) node {$y$};
        \draw (0.2,0) node {$G$};
        \draw (-1,1.75) node {$D$};
        \draw (-0.3,1.75) node {$D$};
        \draw (0.3,1.75) node {$D$};
        \draw (1,1.75) node {$D$};
    \end{scope}
    \end{tikzpicture}\;,
\end{equation}
which yields the amplitude for $\tilde{\mathcal{M}}_U$, as shown in Figure~\ref{fig:FT_measurements}(c):
\begin{equation}
    I^U_{(x,y);z,w} = I_{x;H,w}[F^{wxy}_G]_{xG}\;.
\end{equation}
By substituting the categorical data from $\mathcal{D}(S_3)$, we obtain the values of $I^{U}_{(x,y);z,w}$ as listed in \secref{sec:fault_tolerant_meas}.

\subsection{Fusion and splitting of the fusion trees}
\label{appendix:multi_qutrit_gate}
In this section, we provide a detailed derivation of the fusion and splitting protocols described in Figure~\ref{fig:fusion_G_anyons}. Specifically, in the second subroutine of Figure~\ref{fig:fusion_G_anyons}(a), the outcome of the anyon interferometry protocol can be computed by applying an $F$-move to the upper middle panel, as follows:
\begin{equation}
    \begin{tikzpicture}[baseline = {(current bounding box.center)}]
        \filldraw (-0.6,0.6) circle (3pt); 
        \filldraw (0.6,0.6) circle (3pt);
        \draw (-0.6,0.6) -- (0,0.4);
        \draw (0.6,0.6) -- (0,0.4);
        \draw (0,0.4) -- (0,-0.4);
        \draw (0,-0.4) -- (-0.6,-0.6);
        \draw (0,-0.4) -- (0.6,-0.6);
        \draw (-0.9,0.6) node {$G$};
        \draw (0.9,0.6) node {$G$};
        \draw (-0.8,-0.6) node {$G$};
        \draw (0.8,-0.6) node {$G$};
        \draw (0.5,0) node {$A/B$};
    \end{tikzpicture} = \sum_{X} [(F^{GGG}_G)^{-1}]_{(A/B)X}
    \begin{tikzpicture}[baseline = {(current bounding box.center)}]
        \filldraw (-0.6,0.6) circle (3pt); 
        \filldraw (0.6,0.6) circle (3pt);
        \draw (-0.6,0.6) -- (-0.4,0);
        \draw (0.6,0.6) -- (0.4,0);
        \draw (0.4,0) -- (-0.4,0);
        \draw (-0.4,0) -- (-0.6,-0.6);
        \draw (0.4,0) -- (0.6,-0.6);
        \draw (-0.9,0.6) node {$G$};
        \draw (0.9,0.6) node {$G$};
        \draw (-0.8,-0.6) node {$G$};
        \draw (0.8,-0.6) node {$G$};
        \draw (0,0.2) node {$X$};
    \end{tikzpicture}
    = \frac{1}{2}\begin{tikzpicture}[baseline = {(current bounding box.center)}]
        \filldraw (-0.6,0.6) circle (3pt); 
        \filldraw (0.6,0.6) circle (3pt);
        \draw (-0.6,0.6) -- (-0.4,0);
        \draw (0.6,0.6) -- (0.4,0);
        \draw (0.4,0) -- (-0.4,0);
        \draw (-0.4,0) -- (-0.6,-0.6);
        \draw (0.4,0) -- (0.6,-0.6);
        \draw (-0.9,0.6) node {$G$};
        \draw (0.9,0.6) node {$G$};
        \draw (-0.8,-0.6) node {$G$};
        \draw (0.8,-0.6) node {$G$};
        \draw (0,0.2) node {$A$};
    \end{tikzpicture} + \frac{1}{2}\begin{tikzpicture}[baseline = {(current bounding box.center)}]
        \filldraw (-0.6,0.6) circle (3pt); 
        \filldraw (0.6,0.6) circle (3pt);
        \draw (-0.6,0.6) -- (-0.4,0);
        \draw (0.6,0.6) -- (0.4,0);
        \draw (0.4,0) -- (-0.4,0);
        \draw (-0.4,0) -- (-0.6,-0.6);
        \draw (0.4,0) -- (0.6,-0.6);
        \draw (-0.9,0.6) node {$G$};
        \draw (0.9,0.6) node {$G$};
        \draw (-0.8,-0.6) node {$G$};
        \draw (0.8,-0.6) node {$G$};
        \draw (0,0.2) node {$B$};
    \end{tikzpicture} 
    \pm \frac{1}{\sqrt{2}}
    \begin{tikzpicture}[baseline = {(current bounding box.center)}]
        \filldraw (-0.6,0.6) circle (3pt); 
        \filldraw (0.6,0.6) circle (3pt);
        \draw (-0.6,0.6) -- (-0.4,0);
        \draw (0.6,0.6) -- (0.4,0);
        \draw (0.4,0) -- (-0.4,0);
        \draw (-0.4,0) -- (-0.6,-0.6);
        \draw (0.4,0) -- (0.6,-0.6);
        \draw (-0.9,0.6) node {$G$};
        \draw (0.9,0.6) node {$G$};
        \draw (-0.8,-0.6) node {$G$};
        \draw (0.8,-0.6) node {$G$};
        \draw (0,0.2) node {$G$};
    \end{tikzpicture}\;,
\end{equation}
where $+$ or $-$ sign corresponds to the presence of the $A$ or $B$ anyon on the left-hand side of the equation. In this calculation, we have used the matrix representation of $F^{GGG}_G$ as follows:
\begin{equation}
    F^{GGG}_G = \begin{pmatrix}
        1/2 & 1/2 & 1/\sqrt{2}\\
        1/2 & 1/2 & -1/\sqrt{2}\\
        1/\sqrt{2} & -1/\sqrt{2} & 0
    \end{pmatrix}\;,
\end{equation}
where the columns and rows are labeled in the order $A, B, G$. According to the following non-zero amplitudes:
\begin{equation}
    I_{A;D,A} = 1\;,\quad I_{B;D,A} = -1\;,\quad I_{G;D,G} = \omega^2\;,
\end{equation}
when using a pair of $D$ anyons to perform the interferometry, the resulting fusion tree states, along with the corresponding fusion outcomes $w$, are: 
\begin{equation}
    w=A:\ \frac{1}{\sqrt{2}}\begin{tikzpicture}[baseline = {(current bounding box.center)}]
        \filldraw (-0.6,0.6) circle (3pt); 
        \filldraw (0.6,0.6) circle (3pt);
        \draw (-0.6,0.6) -- (-0.4,0);
        \draw (0.6,0.6) -- (0.4,0);
        \draw (0.4,0) -- (-0.4,0);
        \draw (-0.4,0) -- (-0.6,-0.6);
        \draw (0.4,0) -- (0.6,-0.6);
        \draw (-0.9,0.6) node {$G$};
        \draw (0.9,0.6) node {$G$};
        \draw (-0.8,-0.6) node {$G$};
        \draw (0.8,-0.6) node {$G$};
        \draw (0,0.2) node {$A$};
    \end{tikzpicture} - \frac{1}{\sqrt{2}}\begin{tikzpicture}[baseline = {(current bounding box.center)}]
        \filldraw (-0.6,0.6) circle (3pt); 
        \filldraw (0.6,0.6) circle (3pt);
        \draw (-0.6,0.6) -- (-0.4,0);
        \draw (0.6,0.6) -- (0.4,0);
        \draw (0.4,0) -- (-0.4,0);
        \draw (-0.4,0) -- (-0.6,-0.6);
        \draw (0.4,0) -- (0.6,-0.6);
        \draw (-0.9,0.6) node {$G$};
        \draw (0.9,0.6) node {$G$};
        \draw (-0.8,-0.6) node {$G$};
        \draw (0.8,-0.6) node {$G$};
        \draw (0,0.2) node {$B$};
    \end{tikzpicture} \;,\quad w=G:\ \begin{tikzpicture}[baseline = {(current bounding box.center)}]
        \filldraw (-0.6,0.6) circle (3pt); 
        \filldraw (0.6,0.6) circle (3pt);
        \draw (-0.6,0.6) -- (-0.4,0);
        \draw (0.6,0.6) -- (0.4,0);
        \draw (0.4,0) -- (-0.4,0);
        \draw (-0.4,0) -- (-0.6,-0.6);
        \draw (0.4,0) -- (0.6,-0.6);
        \draw (0,0) -- (0,-0.6);
        \draw (-0.9,0.6) node {$G$};
        \draw (0.9,0.6) node {$G$};
        \draw (-0.8,-0.6) node {$G$};
        \draw (0.8,-0.6) node {$G$};
        \draw (-0.2,0.2) node {$G$};
        \draw (0.2,0.2) node {$G$};
        \draw (0.2,-0.6) node {$G$};
    \end{tikzpicture}\;.
\end{equation}
These present the two possible outcomes, as illustrated in the second subroutine of Figure~\ref{fig:fusion_G_anyons}(a).

Next, we explain why fusing the $G$ anyons in these fusion trees will always result in a single $G$ anyon, as shown in the lower right panel of Figure~\ref{fig:fusion_G_anyons}(a). In the first case, where $w=A$, applying an $F$-move gives:
\begin{equation}
    \frac{1}{\sqrt{2}}\begin{tikzpicture}[baseline = {(current bounding box.center)}]
        \filldraw (-0.6,0.6) circle (3pt); 
        \filldraw (0.6,0.6) circle (3pt);
        \draw (-0.6,0.6) -- (-0.4,0);
        \draw (0.6,0.6) -- (0.4,0);
        \draw (0.4,0) -- (-0.4,0);
        \draw (-0.4,0) -- (-0.6,-0.6);
        \draw (0.4,0) -- (0.6,-0.6);
        \draw (-0.9,0.6) node {$G$};
        \draw (0.9,0.6) node {$G$};
        \draw (-0.8,-0.6) node {$G$};
        \draw (0.8,-0.6) node {$G$};
        \draw (0,0.2) node {$A$};
    \end{tikzpicture} - \frac{1}{\sqrt{2}}\begin{tikzpicture}[baseline = {(current bounding box.center)}]
        \filldraw (-0.6,0.6) circle (3pt); 
        \filldraw (0.6,0.6) circle (3pt);
        \draw (-0.6,0.6) -- (-0.4,0);
        \draw (0.6,0.6) -- (0.4,0);
        \draw (0.4,0) -- (-0.4,0);
        \draw (-0.4,0) -- (-0.6,-0.6);
        \draw (0.4,0) -- (0.6,-0.6);
        \draw (-0.9,0.6) node {$G$};
        \draw (0.9,0.6) node {$G$};
        \draw (-0.8,-0.6) node {$G$};
        \draw (0.8,-0.6) node {$G$};
        \draw (0,0.2) node {$B$};
    \end{tikzpicture} = \begin{tikzpicture}[baseline = {(current bounding box.center)}]
        \filldraw (-0.6,0.6) circle (3pt); 
        \filldraw (0.6,0.6) circle (3pt);
        \draw (-0.6,0.6) -- (0,0.4);
        \draw (0.6,0.6) -- (0,0.4);
        \draw (0,0.4) -- (0,-0.4);
        \draw (0,-0.4) -- (-0.6,-0.6);
        \draw (0,-0.4) -- (0.6,-0.6);
        \draw (-0.9,0.6) node {$G$};
        \draw (0.9,0.6) node {$G$};
        \draw (-0.8,-0.6) node {$G$};
        \draw (0.8,-0.6) node {$G$};
        \draw (0.2,0) node {$G$};
    \end{tikzpicture}\;,
\end{equation}
showing that fusing the two lower $G$ anyons results in a single $G$ anyon. In the second case, where $w=G$, we can first fuse the left two $G$ anyons by performing the following $F$-move:
\begin{equation}
    \begin{tikzpicture}[baseline = {(current bounding box.center)}]
        \filldraw (-0.6,0.6) circle (3pt); 
        \filldraw (0.6,0.6) circle (3pt);
        \draw (-0.6,0.6) -- (-0.4,0);
        \draw (0.6,0.6) -- (0.4,0);
        \draw (0.4,0) -- (-0.4,0);
        \draw (-0.4,0) -- (-0.6,-0.6);
        \draw (0.4,0) -- (0.6,-0.6);
        \draw (0,0) -- (0,-0.6);
        \draw (-0.9,0.6) node {$G$};
        \draw (0.9,0.6) node {$G$};
        \draw (-0.8,-0.6) node {$G$};
        \draw (0.8,-0.6) node {$G$};
        \draw (-0.2,0.2) node {$G$};
        \draw (0.2,0.2) node {$G$};
        \draw (0.2,-0.6) node {$G$};
    \end{tikzpicture} = \frac{1}{\sqrt{2}} \begin{tikzpicture}[baseline = {(current bounding box.center)}]
        \filldraw (-0.6,0.6) circle (3pt); 
        \filldraw (0.6,0.6) circle (3pt);
        \draw (-0.6,0.6) -- (-0.4,0);
        \draw (0.6,0.6) -- (0.4,0);
        \draw (0.4,0) -- (-0.4,0);
        \draw (-0.4,0) -- (-0.6,-0.6);
        \draw (0.4,0) -- (0.6,-0.6);
        \draw (-0.5,-0.3) -- (-0.3,-0.6);
        \draw (-0.9,0.6) node {$G$};
        \draw (0.9,0.6) node {$G$};
        \draw (-0.8,-0.6) node {$G$};
        \draw (0.8,-0.6) node {$G$};
        \draw (0,0.2) node {$G$};
        \draw (-0.1,-0.6) node {$G$};
        \draw (-0.7,-0.2) node {$A$};
    \end{tikzpicture}
    -\frac{1}{\sqrt{2}}
    \begin{tikzpicture}[baseline = {(current bounding box.center)}]
        \filldraw (-0.6,0.6) circle (3pt); 
        \filldraw (0.6,0.6) circle (3pt);
        \draw (-0.6,0.6) -- (-0.4,0);
        \draw (0.6,0.6) -- (0.4,0);
        \draw (0.4,0) -- (-0.4,0);
        \draw (-0.4,0) -- (-0.6,-0.6);
        \draw (0.4,0) -- (0.6,-0.6);
        \draw (-0.5,-0.3) -- (-0.3,-0.6);
        \draw (-0.9,0.6) node {$G$};
        \draw (0.9,0.6) node {$G$};
        \draw (-0.8,-0.6) node {$G$};
        \draw (0.8,-0.6) node {$G$};
        \draw (0,0.2) node {$G$};
        \draw (-0.1,-0.6) node {$G$};
        \draw (-0.7,-0.2) node {$B$};
    \end{tikzpicture}\;.
\end{equation}
Thus, after fusing the two $G$ anyons, the resulting fusion trees and their corresponding probabilities are:
\begin{equation}
    \mathrm{Prob} = \frac{1}{2}:\begin{tikzpicture}[baseline = {(current bounding box.center)}]
        \filldraw (-0.6,0.6) circle (3pt); 
        \filldraw (0.6,0.6) circle (3pt);
        \draw (-0.6,0.6) -- (-0.4,0);
        \draw (0.6,0.6) -- (0.4,0);
        \draw (0.4,0) -- (-0.4,0);
        \draw (-0.4,0) -- (-0.6,-0.6);
        \draw (0.4,0) -- (0.6,-0.6);
        \draw (-0.9,0.6) node {$G$};
        \draw (0.9,0.6) node {$G$};
        \draw (-0.8,-0.6) node {$A$};
        \draw (0.8,-0.6) node {$G$};
        \draw (0,0.2) node {$G$};
    \end{tikzpicture}\;,\quad \mathrm{Prob} = \frac{1}{2}:\begin{tikzpicture}[baseline = {(current bounding box.center)}]
        \filldraw (-0.6,0.6) circle (3pt); 
        \filldraw (0.6,0.6) circle (3pt);
        \draw (-0.6,0.6) -- (-0.4,0);
        \draw (0.6,0.6) -- (0.4,0);
        \draw (0.4,0) -- (-0.4,0);
        \draw (-0.4,0) -- (-0.6,-0.6);
        \draw (0.4,0) -- (0.6,-0.6);
        \draw (-0.9,0.6) node {$G$};
        \draw (0.9,0.6) node {$G$};
        \draw (-0.8,-0.6) node {$B$};
        \draw (0.8,-0.6) node {$G$};
        \draw (0,0.2) node {$G$};
    \end{tikzpicture}\;.
\end{equation}
In both cases, fusing the Abelian anyons $A$ or $B$ with the $G$ anyon results in a single $G$ anyon. Therefore, regardless of the fusion outcome $w$, we can always obtain the two-qutrit fusion tree shown in the lower right panel of Figure~\ref{fig:fusion_G_anyons}(a) by fusing the lower $G$ anyons.  

Finally, we detail the splitting protocol as described in Figure~\ref{fig:fusion_G_anyons}(b). The $F$-move in the final step is straightforward, so we will focus on demonstrating why the internal $B$ anyon does not interfere with universal quantum computation. 

For single-qutrit operations, we begin by considering braiding. Since the worldlines of the $D$ anyons do not have topologically non-trivial intersections with the $B$ worldline during braiding, the presence of the $B$ anyon does not affect any braiding operations. For the $\mathcal{M}_A$ measurement, non-trivial intersections of worldlines occur when a pair of $B$ anyons is used to correct the undesired fusion outcome of an $E$ anyon, as in \eqref{eqn:B_anyon_correction_MA}. The correction is given by:
\begin{equation}
\begin{aligned}
    &\begin{tikzpicture}[baseline = {(current bounding box.center)}]
        \draw (0,-0.4) -- (0,0.5);
        \draw (0,0.5) -- (-0.6,1.1);
        \draw (-0.6,1.1) -- (-1,1.5);
        \draw (-0.6,1.1) -- (-0.2,1.5);
        \draw (0,0.5) -- (0.6,1.1);
        \draw (0.6,1.1) -- (1,1.5);
        \draw (0.6,1.1) -- (0.2,1.5);
        \draw (-0.9,1.4) .. controls (-1.3,1.1) and (-1.1,0.4) .. (0,-0.2);
        \draw (0,0.2) -- (-0.352,0.104);
        \draw (-0.528,0.056) -- (-1.1,-0.1);
        \draw (-0.5,0.7) node {$G$};
        \draw (0.5,0.7) node {$y$};
        \draw (0.2,-0.4) node {$G$};
        \draw (0.2,0.4) node {$G$};
        \draw (0.2,0) node {$G$};
        \draw (-0.9,1.05) node {$E$};
        \draw (-1.2,0.5) node {$B$};
        \draw (-1.1,-0.3) node {$B$};
        \draw (-1,1.75) node {$D$};
        \draw (-0.3,1.75) node {$D$};
        \draw (0.3,1.75) node {$D$};
        \draw (1,1.75) node {$D$};
    \end{tikzpicture} = [F^{BDD}_G]_{EG} [F^{BGy}_{G}]_{GG}\ \begin{tikzpicture}[baseline = {(current bounding box.center)}]
        \draw (0,-0.4) -- (0,0.5);
        \draw (0,0.5) -- (-0.6,1.1);
        \draw (-0.6,1.1) -- (-1,1.5);
        \draw (-0.6,1.1) -- (-0.2,1.5);
        \draw (0,0.5) -- (0.6,1.1);
        \draw (0.6,1.1) -- (1,1.5);
        \draw (0.6,1.1) -- (0.2,1.5);
        \draw (0,-0.2) .. controls (-0.4,-0.15) and (-0.4,0.35) .. (0,0.4);
        \draw (0,0.1) -- (-0.22,0.056);
        \draw (-0.352,0.036) -- (-1.1,-0.1);
        \draw (-0.5,0.7) node {$G$};
        \draw (0.5,0.7) node {$y$};
        \draw (-0.5,0.3) node {$B$};
        \draw (-1.1,-0.3) node {$B$};
        \draw (0.2,0.25) node {$G$};
        \draw (0.2,-0.4) node {$G$};
        \draw (-1,1.75) node {$D$};
        \draw (-0.3,1.75) node {$D$};
        \draw (0.3,1.75) node {$D$};
        \draw (1,1.75) node {$D$};
    \end{tikzpicture} \\
    &= [F^{BDD}_G]_{EG} [F^{BGy}_{G}]_{GG} [(F^{BBG}_G)^{-1}]_{GA}
    \begin{tikzpicture}[baseline = {(current bounding box.center)}]
        \draw (0,-0.4) -- (0,0.5);
        \draw (0,0.5) -- (-0.6,1.1);
        \draw (-0.6,1.1) -- (-1,1.5);
        \draw (-0.6,1.1) -- (-0.2,1.5);
        \draw (0,0.5) -- (0.6,1.1);
        \draw (0.6,1.1) -- (1,1.5);
        \draw (0.6,1.1) -- (0.2,1.5);
        \draw (0,-0.2) .. controls (-1.3,-0.15) and (-0.7,0.9) .. (-0.22,0.056);
        \draw (-0.22,0.056) -- (-0.55,0);
        \draw (-0.77,-0.04) -- (-1.1,-0.1);
        \draw (-0.5,0.7) node {$G$};
        \draw (0.5,0.7) node {$y$};
        \draw (-1.0,0.3) node {$B$};
        \draw (-1.1,-0.3) node {$B$};
        \draw (0.2,-0.4) node {$G$};
        \draw (-1,1.75) node {$D$};
        \draw (-0.3,1.75) node {$D$};
        \draw (0.3,1.75) node {$D$};
        \draw (1,1.75) node {$D$};
    \end{tikzpicture} = [F^{BDD}_G]_{EG} [F^{BGy}_{G}]_{GG}\begin{tikzpicture}[baseline = {(current bounding box.center)}]
        \draw (0,-0.4) -- (0,0.5);
        \draw (0,0.5) -- (-0.6,1.1);
        \draw (-0.6,1.1) -- (-1,1.5);
        \draw (-0.6,1.1) -- (-0.2,1.5);
        \draw (0,0.5) -- (0.6,1.1);
        \draw (0.6,1.1) -- (1,1.5);
        \draw (0.6,1.1) -- (0.2,1.5);
        \draw (0,0.1) -- (-1.1,-0.1);
        \draw (-0.5,0.7) node {$G$};
        \draw (0.5,0.7) node {$y$};
        \draw (-1.1,-0.3) node {$B$};
        \draw (0.2,0.25) node {$G$};
        \draw (0.2,-0.4) node {$G$};
        \draw (-1,1.75) node {$D$};
        \draw (-0.3,1.75) node {$D$};
        \draw (0.3,1.75) node {$D$};
        \draw (1,1.75) node {$D$};
    \end{tikzpicture}\;,
\end{aligned}
\end{equation}
where we have used $[F^{BBG}_G]_{AG} = 1$ and $(R^{BB}_A)^2 = 1$. This result returns to the correction described in \eqref{eqn:B_anyon_correction_MA}. For more complex fusion trees involving multiple $B$ worldlines from external fusion trees, a similar analysis leads to the same conclusion. Therefore, the internal $B$ anyon worldlines do not affect the $\mathcal{M}_A$ measurement.

For the measurement $\tilde{\mathcal{M}}_U$, we consider the fusion between the $w$ anyon and the $G$ anyon, as originally described in \eqref{eqn:w_anyon_fusion_MU}, but now with an additional $B$ anyon worldline: 
\begin{equation}
\begin{aligned}
    &\begin{tikzpicture}[baseline = {(current bounding box.center)}]
        \draw (0,0) -- (0,0.5);
        \draw (0,0.5) -- (-0.6,1.1);
        \draw (-0.6,1.1) -- (-1,1.5);
        \draw (-0.6,1.1) -- (-0.2,1.5);
        \draw (0,0.5) -- (0.6,1.1);
        \draw (0.6,1.1) -- (1,1.5);
        \draw (0.6,1.1) -- (0.2,1.5);
        \draw (-0.3,0.8) .. controls (-0.5,0.6) and (-0.6,0.3) .. (-0.7,0);
        \draw (0,0.25) -- (-0.55,0.25);
        \draw (-0.65,0.25) -- (-1,0.25);
        \draw (-0.25,0.5) node {$x$};
        \draw (-0.6,0.85) node {$x$};
        \draw (0.5,0.7) node {$y$};
        \draw (-0.5,0) node {$w$};
        \draw (0.2,0) node {$G$};
        \draw (-1,1.75) node {$D$};
        \draw (-0.3,1.75) node {$D$};
        \draw (0.3,1.75) node {$D$};
        \draw (1,1.75) node {$D$};
        \draw (-1,0.45) node {$B$};
    \end{tikzpicture} = [F^{wxy}_{G}]_{xG} 
    \begin{tikzpicture}[baseline = {(current bounding box.center)}]
        \draw (0,0) -- (0,0.5);
        \draw (0,0.5) -- (-0.6,1.1);
        \draw (-0.6,1.1) -- (-1,1.5);
        \draw (-0.6,1.1) -- (-0.2,1.5);
        \draw (0,0.5) -- (0.6,1.1);
        \draw (0.6,1.1) -- (1,1.5);
        \draw (0.6,1.1) -- (0.2,1.5);
        \draw (0,0.4) -- (-1,0.4);
        \draw (0,0.2) -- (-1,0.2);
        \draw (-0.5,0.7) node {$x$};
        \draw (0.5,0.7) node {$y$};
        \draw (0.2,0) node {$G$};
        \draw (-1,1.75) node {$D$};
        \draw (-0.3,1.75) node {$D$};
        \draw (0.3,1.75) node {$D$};
        \draw (1,1.75) node {$D$};
        \draw (-1,0.6) node {$w$};
        \draw (-1,0) node {$B$};
    \end{tikzpicture}\\
    &= [F^{wxy}_{G}]_{xG}[(F^{BwG}_G)^{-1}]_{Gx^\prime} (R^{Bw}_{x^\prime})^* [(F^{wBG}_G)]_{x^\prime G} \begin{tikzpicture}[baseline = {(current bounding box.center)}]
        \draw (0,0) -- (0,0.5);
        \draw (0,0.5) -- (-0.6,1.1);
        \draw (-0.6,1.1) -- (-1,1.5);
        \draw (-0.6,1.1) -- (-0.2,1.5);
        \draw (0,0.5) -- (0.6,1.1);
        \draw (0.6,1.1) -- (1,1.5);
        \draw (0.6,1.1) -- (0.2,1.5);
        \draw (0,0.4) -- (-1,0.4);
        \draw (0,0.2) -- (-1,0.2);
        \draw (-0.5,0.7) node {$x$};
        \draw (0.5,0.7) node {$y$};
        \draw (0.2,0) node {$G$};
        \draw (-1,1.75) node {$D$};
        \draw (-0.3,1.75) node {$D$};
        \draw (0.3,1.75) node {$D$};
        \draw (1,1.75) node {$D$};
        \draw (-1,0.6) node {$B$};
        \draw (-1,0) node {$w$};
    \end{tikzpicture}
    =[F^{wxy}_{G}]_{xG} \begin{tikzpicture}[baseline = {(current bounding box.center)}]
        \draw (0,0) -- (0,0.5);
        \draw (0,0.5) -- (-0.6,1.1);
        \draw (-0.6,1.1) -- (-1,1.5);
        \draw (-0.6,1.1) -- (-0.2,1.5);
        \draw (0,0.5) -- (0.6,1.1);
        \draw (0.6,1.1) -- (1,1.5);
        \draw (0.6,1.1) -- (0.2,1.5);
        \draw (0,0.4) -- (-1,0.4);
        \draw (0,0.2) -- (-1,0.2);
        \draw (-0.5,0.7) node {$x$};
        \draw (0.5,0.7) node {$y$};
        \draw (0.2,0) node {$G$};
        \draw (-1,1.75) node {$D$};
        \draw (-0.3,1.75) node {$D$};
        \draw (0.3,1.75) node {$D$};
        \draw (1,1.75) node {$D$};
        \draw (-1,0.6) node {$B$};
        \draw (-1,0) node {$w$};
    \end{tikzpicture}
\end{aligned} 
\end{equation}
which returns to \eqref{eqn:w_anyon_fusion_MU}. In this equation, we have used
\begin{equation}
    \begin{cases}
        w=A,\ x^\prime = B: &\quad [(F^{BAG}_G)^{-1}]_{GB} =  (R^{BA}_{B})^* =  [(F^{ABG}_G)]_{BG} = 1\\
        w=B,\ x^\prime = A: &\quad [(F^{BBG}_G)^{-1}]_{GA} =  (R^{BB}_{A})^* =  [(F^{BBG}_G)]_{AG} = 1
    \end{cases}\;.
\end{equation}
By extending this analysis to more complex fusion trees involving multiple $B$ worldlines, we can conclude that the internal $B$ anyon worldlines do not affect the $\tilde{\mathcal{M}}_U$ measurement.

For the multi-qutrit operations, the $B$ worldline will only have topologically non-trivial effect when two fusion trees in $V^{GGGG}_D$ are swapped. However, one can verify that this swap always result in a trivial phase $R^{BG}_G R^{GB}_G = 1$. Therefore, the $B$ worldline does not affect the multi-qutrit operation either.

In summary, we have demonstrated that, all operations required to realize the universal gate set are unaffected by the internal $B$ worldline, confirming that the splitting protocol is successful.

\section{The circuit realization of \texorpdfstring{$\mathcal{D}(S_3)$}{D(S3)}} \label{appendix:more_circuit_realization}

\subsection{The circuit realization of \texorpdfstring{$F^{R,C}_{\rho_h}$}{FRCrhoh}}
\label{appendix:circuit_F_RCrhoh}

For later convenience, we restate the expressions for $\rho_h$ and $F^{R,C;u,v}_{\rho_h}$ as follows:
\begin{equation} \label{eqn:F_RC_rho_h_expression}
    F^{R,C;u,v}_{\rho_h} = \frac{|R|}{|Z(C)|} \sum_{n\in Z(C)} \Gamma^R_{jj^\prime}(n) (T^{\tau_c n \bar{\tau}_{c^\prime}}_{+})_1 (L^{c^\prime}_{+})_2\;,\quad     \rho_h = \ \begin{tikzpicture}[baseline={(current bounding box.center)}]
    \foreach \i in {0,1}
        {\foreach \j in {0,1}{
            \draw[thick] (\i,\j+1) -- (\i,\j);
            \draw[thick] (\i,\j+1) -- (\i+1,\j+1);
            \draw[thick] (\i,\j) -- (\i+1,\j);
            \draw[thick] (\i+1,\j+1) -- (\i+ 1,\j);
            }}
    \draw (0.5,0.5)--(1,1);
    \draw (0.5,0.5)-- (1.5,0.5);
    \foreach \i in {0,1}{           
        \draw (\i,1) -- (0.5+\i,0.5);
    }
    \node at (0.2,0.5) {$s$};
    \node at (1.6,0.8) {$s^\prime$};
    \node at (0.8,0.3) {$2$};
    \node at (0.47,1.2) {$1$};
    \end{tikzpicture}
\end{equation}

Next, we explore the circuit realization of $F^{R,C_2}_{\rho_h}$. For $C_2 = \{\sigma,\mu\sigma,\mu^2\sigma\}$, we have $Z(C_2) = \{e,\sigma\}\cong \mathbb{Z}_2$, with one-dimensional irreps $R=[+],[-]$, indicating $j=j^\prime = 1$. Therefore, we use $u,v = 0,1,2$ to represent $c,c^\prime=\{\sigma,\mu\sigma,\mu^2\sigma\} \in C_2$, respectively. For $n\in Z(C_2)$, one can obtain the $\tau_c n \bar{\tau}_{c^\prime}$ for $T$ operators as listed in Table~\ref{tab:C_2_transformation}.
\begin{table}[h]
\centering
\begin{tabular}{|c|cc|cc|cc|}
\hline
\diagbox[dir=SE]{$c \ (\tau_{c})$}{$c^\prime \ (\tau_{c^\prime})$} & \multicolumn{2}{c|}{$\sigma\ (e)$}                          & \multicolumn{2}{c|}{$\mu \sigma \ (\mu^2)$}                          & \multicolumn{2}{c|}{$\mu^2 \sigma\ (\mu)$}                          \\ \hline
$\sigma\ (e)$ & \multicolumn{1}{c|}{$e$} & $\sigma$                      & \multicolumn{1}{c|}{$\mu$} & $\mu^2\sigma$                      & \multicolumn{1}{c|}{$\mu^2$} & $\mu\sigma$                      \\ \hline
$\mu \sigma\ (\mu^2)$ & \multicolumn{1}{c|}{$\mu^2$} & $\mu^2 \sigma$                      & \multicolumn{1}{c|}{$e$} & $\mu\sigma$                      & \multicolumn{1}{c|}{$\mu$} & $\sigma$                      \\ \hline
$\mu^2 \sigma \ (\mu)$ & \multicolumn{1}{c|}{$\mu$} & $\mu\sigma$ & \multicolumn{1}{c|}{$\mu^2$} & $\sigma$ & \multicolumn{1}{c|}{$e$} & $\mu^2\sigma$ \\ \hline
\end{tabular}
\caption{The $\tau_{c} n \bar{\tau}_{c^\prime}$ for $C_2$, for $n=e$ and $n=\sigma$ in the left and right sub-cell, respectively.}
\label{tab:C_2_transformation}
\end{table}

\begin{table}[h]
\centering
{\renewcommand{\arraystretch}{1.3}
\begin{tabular}{|c|c|c|c|}
\hline
\diagbox[dir=SE]{$u$}{$v$} & $0$ & $1$& $2$ \\ \hline
$0$ & $\ket{\hat{0}0}\pm \ket{\hat{0}1}$ & $\ket{\hat{1}0}\pm\ket{\hat{2}1}$ & $\ket{\hat{2}0}\pm\ket{\hat{1}1}$  \\ \hline
$1$ & $\ket{\hat{2}0}\pm\ket{\hat{2}1}$ & $\ket{\hat{0}0}\pm\ket{\hat{1}1}$ & $\ket{\hat{1}0}\pm\ket{\hat{0}1}$ \\ \hline
$2$ & $\ket{\hat{1}0}\pm\ket{\hat{1}1}$ & $\ket{\hat{2}0}\pm\ket{\hat{0}1}$ & $\ket{\hat{0}0}\pm\ket{\hat{2}1}$ \\ \hline
\end{tabular}}
\caption{The basis $\ket{u,v}_{\pm}$ for the measurement on edge $1$ by $F^{R,C_2;u,v}_{\rho_h}$, where the $\pm$ corresponds to $R=[+]/[-]$, respectively.}
\label{tab:C_2_measurement}
\end{table}

Given $u,v$, the summation of $n\in Z(C_2)$ in \eqref{eqn:F_RC_rho_h_expression} for $T$ operators results in the projection of edge $1$ onto the states $\ket{u,v}_{\pm}$ as listed in Table~\ref{tab:C_2_measurement}. The $L^{c^\prime}_{+}$ operators act on edge $2$ according to \eqref{eqn:L_circuits}. Therefore, the circuit realization of $F^{R,C_2;u,v}_{\rho_h}$ is given by:
\begin{align}
    & F^{[+],C_2;u,v}_{\rho_h} = (\ket{u,v}\bra{u,v}_+)_1 \otimes (\hat{X}^v \hat{C}\otimes X)_2\;,\\
    & F^{[-],C_2;u,v}_{\rho_h} = (\ket{u,v}\bra{u,v}_-)_1 \otimes (\hat{X}^v \hat{C}\otimes X)_2\;.
\end{align}

One can prepare the state $\ket{u,v}_{\pm}$ by applying $(\hat{X}^Z)^v (\hat{C}\otimes I)$ on the $\ket{\hat{u}}\otimes \ket{\pm}$ state. Thus, the projection to $\ket{u,v}_{\pm}$ can be realized by applying a $(\hat{C}\otimes I)\hat{X}^{-vZ}$ on edge $1$ and measuring it in the Pauli $\hat{Z}$ basis of the qutrit and the Pauli $X$ basis of the qubit. If the $X$ measurement yields the opposite result (i.e. $\ket{\mp}$ when we expect $\ket{\pm}$), we have effectively applied a ribbon with the opposite charge (i.e. applying $D$ when we intend $E$ and vice versa). This can be corrected by applying a Pauli $Z$ on edge $1$, according to the fusion rule $D\times B = E$ and $E\times B = D$. Following the protocol for the $C$ anyon (see the paragraph below \eqref{eqn:F_2_C1_uv_circuit}), we make $v$ maximally mixed by employing a $\ket{\hat{0}_+}$ ancilla qutrit as the control of the $\hat{X}^v$ and $\hat{X}^{-vZ}$ gates and tracing it out. Next, we make $u$ maximally mixed by projecting $\ket{\hat{u}}$ into a Bell pair with an ancilla qutrit in the maximally mixed state $\hat{I}_u$ and tracing out the ancilla. In conclusion, the circuit realization for $F^{[\pm],C_2}_{\rho_h}$ is given by:
\begin{equation}
\begin{gathered}
    \includegraphics[width=10cm]{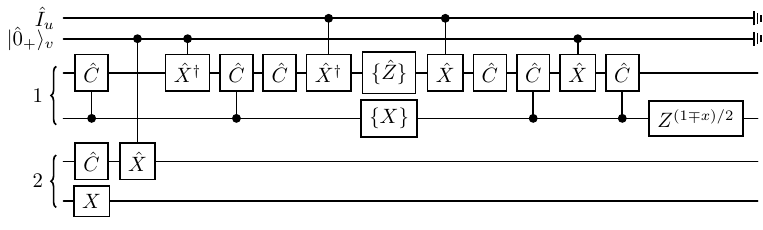}
\end{gathered}
\end{equation}
where $\{\hat{Z}\}$ and $\{X\}$ denote measurements in the Pauli $\hat{Z}$ and $X$ bases, respectively. The $Z^{(1\mp x)/2}$ represents the adaptive application of $Z$ based on the $\{X\}$ measurement outcome $x=1$ or $x=-1$, i.e. the state $\ket{+}$ or $\ket{-}$. For $([+],C_2)$/$([-],C_2)$, we apply $Z$ when $x=-1$/$x=1$, respectively.

Finally, we construct the circuit for $F^{R,C_3}_{\rho_h}$. For $C_3 = \{\mu,\mu^2\}$, we have $Z(C_3) = \{e,\mu,\mu^2\}$, with one-dimensional irreps $R=[1],[\omega],[\bar{\omega}]$. We use $u,v = 0,1$ to denote $c,c^\prime = \{\mu,\mu^2\} \in C_3$, respectively. Then one can obtain the $\tau_c n \bar{\tau}_{c^\prime}$ as listed in Table~\ref{tab:C_3_transformation}.

\begin{table}[h]
\centering
\begin{tabular}{|c|ccc|ccc|}
\hline
\diagbox[dir=SE]{$c \ (\tau_{c})$}{$c^\prime \ (\tau_{c^\prime})$} & \multicolumn{3}{c|}{$\mu\ (e)$}                          & \multicolumn{3}{c|}{$\mu^2 \ (\sigma)$}                           \\ \hline
$\mu\ (e)$ & \multicolumn{1}{c|}{$e$} & \multicolumn{1}{c|}{$\mu$} & $\mu^2$ & \multicolumn{1}{c|}{$\sigma$}                      & \multicolumn{1}{c|}{$\mu\sigma$} & $\mu^2\sigma$                      \\ \hline
$\mu^2\ (\sigma)$ & \multicolumn{1}{c|}{$\sigma$} & \multicolumn{1}{c|}{$\mu^2 \sigma$} & $\mu\sigma$ & \multicolumn{1}{c|}{$e$} & \multicolumn{1}{c|}{$\mu^2$} & $\mu$ \\ \hline
\end{tabular}
\caption{The $\tau_{c} n \bar{\tau}_{c^\prime}$ for $C_3$, for $n=e$, $n=\mu$ and $n=\mu^2$ in the left, middle and right sub-cell, respectively.}
\label{tab:C_3_transformation}
\end{table}

Given $u$ and $v$, the summation of $n\in Z(C_3)$ in \eqref{eqn:F_RC_rho_h_expression} for $T$ operators applies an $\hat{I}$, $\hat{Z}^{u+1}$ or $\hat{Z}^{-(u+1)}$ gate on edge $1$ for the $[1],[\omega]$ and $[\bar{\omega}]$ representation, respectively, followed by a projection of the qubit to the state $\ket{u+v}$. Therefore, combining the action of $L^{c^\prime}_+$ on edge 2, the circuit for $F^{R,C_3;u,v}_{\rho_h}$ is given by:
\begin{align}
    F^{[1],C_3;u,v}_{\rho_h} &= (\hat{I}\otimes \ket{u+v}\bra{u+v}) \otimes (\hat{X}^{v+1} \otimes I)_2\;, \\
    F^{[\omega],C_3;u,v}_{\rho_h} &= (\hat{Z}^{u+1}\otimes \ket{u+v}\bra{u+v}) \otimes (\hat{X}^{v+1} \otimes I)_2\;, \\
    F^{[\bar{\omega}],C_3;u,v}_{\rho_h} &= (\hat{Z}^{-(u+1)}\otimes \ket{u+v}\bra{u+v}) \otimes (\hat{X}^{v+1} \otimes I)_2\;.
\end{align}

Using similar methods, we construct the circuit realization of $F^{R,C_3}_{\rho_h}$ as follows:
\begin{equation}
\begin{gathered}
\includegraphics[width=8cm]{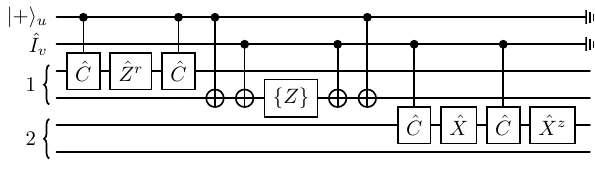}
\end{gathered}
    \;,
\end{equation}
where $r = 0,1,-1$ corresponds to the representations $[1],[\omega],[\bar{\omega}]$, respectively, and $z=0,1$ represents the measurement outcomes of $\{Z\}$ for $\ket{0},\ket{1}$, respectively. The adaptive application of $\hat{X}$ ensures the correct implementation of $F^{R,C_3}_{\rho_h}$ regardless of the measurement result of $\{Z\}$.

\subsection{The circuit realization of \texorpdfstring{$F^{R,C}_{\rho_v}$}{FRCrhov}}
\label{appendix:circuit_F_RCrhov}

For later convenience, we restate the expressions for $\rho_v$ and $F^{R,C;u,v}_{\rho_v}$ as follows:
\begin{equation}
    F^{R,C;u,v}_{\rho_v} = \frac{|R|}{|Z(C)|}\sum_{n\in Z(C)} \Gamma^{R}_{jj^\prime}(n) (L^{c}_{+})_1 (T^{\tau_c n\bar{\tau}_{c^\prime}}_{-})_2\;,\quad \rho_v = \ \begin{tikzpicture}[baseline={(current bounding box.center)}]
    \foreach \i in {0,1}
        {\foreach \j in {0,1}{
            \draw[thick] (\i,\j+1) -- (\i,\j);
            \draw[thick] (\i,\j+1) -- (\i+1,\j+1);
            \draw[thick] (\i,\j) -- (\i+1,\j);
            \draw[thick] (\i+1,\j+1) -- (\i+ 1,\j);
            }}
    \draw (0,1)--(0.5,1.5);
    \draw (0.5,0.5)-- (0.5,1.5);
    \foreach \i in {0,1}{           
        \draw (0,\i+1) -- (0.5,\i+0.5);
    }
    \node at (0.2,0.5) {$s$};
    \node at (0.55,1.8) {$s^\prime$};
    \node at (0.17,1.51) {$2$};
    \node at (0.7,0.75) {$1$};
    \end{tikzpicture}\;. 
\end{equation}

For the conjugacy class $C_1$, the ribbon operators for anyons $A=(\{+\},C_1)$ and $B=(\{-\},C_1)$ are given by:
\begin{equation}
    F^{[+],C_1}_{\rho_v} = (\hat{I}\otimes I)_1 \otimes (\hat{I}\otimes I)_2\;,\quad F^{[-],C_1}_{\rho_v} = (\hat{I}\otimes I)_1 \otimes (\hat{I}\otimes Z)_2\;.
\end{equation}

For the anyon $C=([2],C_1)$, we obtain the circuit realization of $F^{[2],C_1;u,v}_{\rho_v}$ as follows:
\begin{equation}
    F^{[2],C_1;u,v}_{\rho_v} = (\hat{I}\otimes I)_1\otimes (\hat{Z}^{-(v+1)} \otimes \ket{u+v}\bra{u+v})_2\;.
\end{equation}

To apply the $F^{[2],C_1}_{\rho_v}$ operator, we use the following circuit:
\begin{equation}
\begin{gathered}
\includegraphics{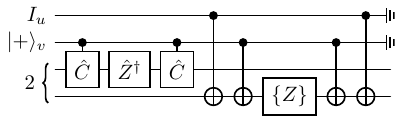}
\end{gathered}\;.
\end{equation}

\begin{table}[ht]
\centering
{\renewcommand{\arraystretch}{1.3}
\begin{tabular}{|c|c|c|c|}
\hline
\diagbox[dir=SE]{$u$}{$v$} & $0$ & $1$& $2$ \\ \hline
$0$ & $\ket{\hat{0}0}\pm\ket{\hat{0}1}$ & $\ket{\hat{2}0}\pm\ket{\hat{2}1}$ & $\ket{\hat{1}0}\pm\ket{\hat{1}1}$  \\ \hline
$1$ & $\ket{\hat{1}0}\pm\ket{\hat{2}1}$ & $\ket{\hat{0}0}\pm\ket{\hat{1}1}$ & $\ket{\hat{2}0}\pm\ket{\hat{0}1}$ \\ \hline
$2$ & $\ket{\hat{2}0}\pm\ket{\hat{1}1}$ & $\ket{\hat{1}0}\pm\ket{\hat{0}1}$ & $\ket{\hat{0}0}\pm\ket{\hat{2}1}$ \\ \hline
\end{tabular}}
\caption{The basis $\ket{u,v}_{\pm}$ for the measurement on edge $2$ by $F^{R,C_2;u,v}_{\rho_v}$, where the $\pm$ corresponds to $R=[+]/[-]$, respectively.}
\label{tab:C_2_measurement_v}
\end{table}

Next, we construct the circuit realization of $F^{R,C_{2}}_{\rho_v}$. By taking the inverse of $\tau_c n \bar{\tau}_{c^\prime}$ in Table~\ref{tab:C_2_transformation}, we can derive the basis $\ket{u,v}_{\pm}$ for the measurement on edge $2$, as listed in Table~\ref{tab:C_2_measurement_v}. The basis can be constructed by applying $(\hat{X}^Z)^u(\hat{C}\otimes I)$ on $\ket{\hat{v}}\otimes \ket{+}$. Equivalently, we apply $(\hat{C}\otimes I)\hat{X}^{-uZ}$ on edge $2$ and measure the qutrit and qubit in the Pauli $\hat{Z}$ and $X$ bases, respectively. Thus, we have the following circuit for $F^{R,C_2;u,v}_{\rho_v}$:
\begin{align}
    F_{\rho_v}^{[+],C_2;u,v} &= (\hat{X}^u \hat{C} \otimes X)_1 \otimes (\ket{u,v}\bra{u,v}_{+})_2 \;,\\ F_{\rho_v}^{[-],C_2;u,v} &= (\hat{X}^u \hat{C} \otimes X)_1 \otimes (\ket{u,v}\bra{u,v}_{-})_2\;.
\end{align}
Therefore, the circuit for $F^{[\pm],C_2}_{\rho_v}$ is given by:
\begin{equation}
\begin{gathered}
\includegraphics[width=11cm]{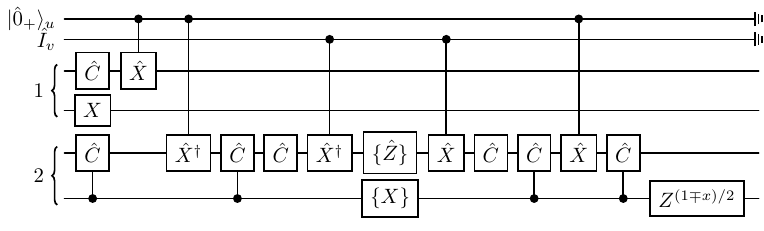}
\end{gathered}
    \;,
\end{equation}
where $Z^{(1\mp x)/2}$ is the adaptive application of $Z$ based on the $\{X\}$ measurement outcome $x=1$ or $x=-1$.

To apply $F^{R,C_3}_{\rho_v}$, one can derive the measurement on edge $2$ by taking the inverse of $\tau_c n \bar{\tau}_{c^\prime}$ in Table~\ref{tab:C_3_transformation}. Thus, for $u,v = 0,1$, the circuit for $F^{R,C_3;u,v}_{\rho_v}$ is given by:
\begin{align}
    F^{[1],C_3;u,v}_{\rho_v} &= (\hat{X}^{u+1} \otimes I)_1\otimes (\hat{I} \otimes \ket{u+v}\bra{u+v})_2\;,\\
    F^{[\omega],C_3;u,v}_{\rho_v} &= (\hat{X}^{u+1} \otimes I)_1\otimes (\hat{Z}^{-(v+1)} \otimes \ket{u+v}\bra{u+v})_2\;,\\
    F^{[\bar{\omega}],C_3;u,v}_{\rho_v} &= (\hat{X}^{u+1} \otimes I)_1\otimes (\hat{Z}^{v+1} \otimes \ket{u+v}\bra{u+v})_2\;.
\end{align}
We construct the circuit realization of $F^{R,C_3}_{\rho_v}$ as follows:
\begin{equation}
\begin{gathered}
\includegraphics[width=10cm]{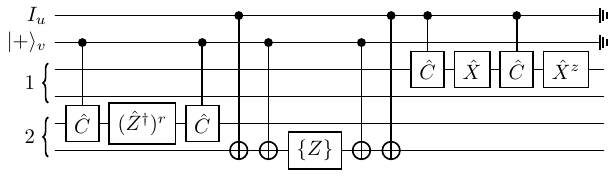}
\end{gathered}\;,
\end{equation}
where $r = 0,1,-1$ corresponds to the representations $[1],[\omega],[\bar{\omega}]$, respectively, and $z=0,1$ represents the measurement outcomes of $\{Z\}$.

\subsection{The circuit realization of \texorpdfstring{$K^{R,C}_{s}$}{KRCs}} \label{appendix:circuit_K_RC_s}

In this section, we explain how to construct the circuit for $K^{R,C}_{s}$ measurement. We will often omit identity operators for brevity.

Let us start with explaining the $B_p^h$ measurement shown in \figref{fig:KRC_circuit}(a) and (b).
The function $\delta_{h,g_1g_2\bar{g}_3\bar{g}_4}$ in $B^h_p$ projects the four edges of $p$ to the subspace satisfying $g_1g_2\bar{g}_3\bar{g}_4 = h$ (see \eqref{eqn:A_B_ops_DG}). By representing $g_i = \mu^{k_i}\sigma^{l_i}$ and using the relation $\sigma\mu\sigma = \bar{\mu}$, we can reexpress the function in the qutrit and qubit basis as:
\begin{equation} \label{eqn:delta_func_B}
    \delta_{k_h,k_1+(-1)^{l_1}k_2-(-1)^{l_h + l_4}k_3 - (-1)^{l_h}k_4}\delta_{l_h,l_1+l_2+l_3+l_4}\;.
\end{equation}
Since the Pauli $Z$ operator has eigenvalue $(-1)^l$ when acting on the $\ket{l}$ state, one can verify that the qubit delta function $\delta_{l_h,l_1+l_2+l_3+l_4}$ projects the state to the $(-1)^{l_h}$ eigenstate of the $Z_1\otimes Z_2\otimes Z_3 \otimes Z_4$ operator. 

A similar analysis applies to the qutrit delta function, where $(-1)^{l_1} k_2$ and $-(-1)^{l_4}k_3$ correspond to the eigenvalues $\omega^{(-1)^{l_1}k_2}$ and $\omega^{-(-1)^{l_4}k_3}$ of the operators $\hat{Z}_2^{Z_1}$ and $\hat{Z}_3^{-Z_4}$, respectively. Note that the $(-1)^{l_h}$ in the qutrit delta function indicates that the measurement circuit for $k_h$ depends on the qubit measurement outcome. Therefore, by defining the following operators:
\begin{equation} \label{eqn:Z_p}
\begin{aligned}
    \hat{Z}^+_p &\coloneqq \hat{Z}_1 \otimes \hat{Z}_2^{Z_1}\otimes \hat{Z}_3^{-Z_4} \otimes \hat{Z}_4^\dagger\;,\\
    \hat{Z}^-_p &\coloneqq \hat{Z}_1 \otimes \hat{Z}_2^{Z_1}\otimes \hat{Z}_3^{Z_4} \otimes \hat{Z}_4\;,\\
    Z_p &\coloneqq Z_1\otimes Z_2\otimes Z_3 \otimes Z_4\;,
\end{aligned}
\end{equation}
we can express $B^{h}_p$ in the qutrit and qubit basis as follows:
\begin{equation} \label{eqn:B_p_projectors}
    B^{\mu^{k_h}}_p = P^{k_h}_{\hat{Z}^+_p} P^{0}_{Z_p}\;,\quad B^{\mu^{k_h}\sigma}_p = P^{k_h}_{\hat{Z}^-_p} P^{1}_{Z_p}\;.
\end{equation}
More concretely, to perform the $B^h_p$ measurement, we first measure $Z_p$. Depending on the outcome, $l_h=0$ or $l_h = 1$, we measure $\hat{Z}_p^+$ or $\hat{Z}_p^-$ to obtain $k_h$, respectively. The outcome of the $B^h_p$ measurement is thus $h = \mu^{k_h} \sigma^{l_h}$. 

Now, we explain the circuits in Figure~\ref{fig:KRC_circuit}(a) and (b) for the $\hat{Z}^{\pm}_p$ and $Z_p$ measurements. Consider the $Z_p$ measurement circuit in Figure~\ref{fig:KRC_circuit}(b), where a qubit ancilla in the state $\ket{0}$ is introduced, with the stabilizer $Z_a$. Utilizing the algebraic relations between the controlled-NOT (CNOT) gate and qubit Pauli's:
\begin{equation} \label{eqn:algebra_CNOT}
    CX(I \otimes Z)CX = Z \otimes Z\;, \  CX(X \otimes I)CX = X \otimes X\;,
\end{equation}
the stabilizer $Z_a$ is mapped to $Z_a \otimes Z_p$ by the entangling gates in the circuit. Consequently, the Pauli-$Z$ measurement on the ancilla qubit is equivalent to the $Z_p$ measurement on the physical qubits at edges $1$ through $4$.

The measurement circuit for $\hat{Z}_p^{\pm}$ employs the qutrit-qutrit controlled-$\hat{X}$ ($C\hat{X}$) gate. This gate applies $\hat{I}, \hat{X}$, or $\hat{X}^\dagger$ to the target qutrit if the control qutrit is in the state $\ket{\hat{0}},\ket{\hat{1}}$, or $\ket{\hat{2}}$, and has the following algebra with the qutrit Pauli operators:
\begin{equation} \label{app:algebra_qutrit_CX}
\begin{aligned}
    &C\hat{X}(\hat{I}\otimes \hat{Z})C\hat{X}^\dagger = \hat{Z}^\dagger\otimes \hat{Z}\;,\\
    &C\hat{X}(\hat{X}\otimes \hat{I})C\hat{X}^\dagger = \hat{X}\otimes \hat{X}\;,
\end{aligned}
\end{equation}
where the first and second qutrits are the control and target, respectively. By combining \eqref{app:algebra_qutrit_CX} with the algebra of $C\hat{C}$ in \eqref{eqn:algebra_CC}, one can verify that the entangling gates in \figref{fig:KRC_circuit}(a) with a $+$ or $-$ sign map the stabilizer $\hat{Z}_a$ of the $\ket{\hat{0}}$ ancilla qutrit into $\hat{Z}_a\otimes \hat{Z}^+_p$ or $\hat{Z}_a\otimes \hat{Z}^-_p$, respectively. Therefore, by measuring the ancilla qutrit in the Pauli $\hat{Z}$ basis, we achieve the measurement of $\hat{Z}^{\pm}_p$. The above arguments complete the $B^h_p$ measurement circuits.

Now, we explain the measurement circuits for $A^\mu_v$, $A^\sigma_v$, $A^{\mu \sigma}_v$ and $A^{\mu^2 \sigma}_v$, as illustrated in \figref{fig:KRC_circuit}(c)(d) and (e).
Combining \eqref{eqn:A_B_ops_DG} and \eqref{eqn:L_circuits}, we have the following circuit realizations for $A^\mu_v$ and $A^\sigma_v$:
\begin{align}
    A^{\mu}_v &=\hat{X}_1\otimes \hat{X}_4 \otimes \hat{X}_5^{-Z_5} \otimes \hat{X}_6^{-Z_6}\;,\\
    A^{\sigma}_v &=(\hat{C}\otimes X)_1\otimes (\hat{C}\otimes X)_4 \otimes X_5 \otimes X_6\;.
\end{align}
For the measurement of $A^{\mu}_v$, we introduce a qutrit ancilla in the state $\ket{\hat{0}_+}$ with the stabilizer $\hat{X}_a$. The entangling gates in Figure~\ref{fig:KRC_circuit}(c) map the stabilizer $\hat{X}_a$ to $\hat{X}_a \otimes A^\mu_v$. By measuring the ancilla in the basis $\mathcal{M}_{A1}=\{\ket{\hat{0}_+},\ket{\hat{0}_+}^\perp\}$, we realize the $A^\mu_v$ measurement in $K^{R,C_1}_s$. Similarly, by measuring the ancilla in the basis $\mathcal{M}_{A2}=\{\ket{\hat{0}_+},\ket{\hat{1}_+},\ket{\hat{2}_+}\}$ (i.e. $\hat{X}$ basis), we realize the $A^\mu_v$ measurement in $K^{R,C_3}_s$. Since $A^{\mu^2}_v = (A^{\mu}_v)^\dagger$, the circuit can be constructed correspondingly. The measurement circuit for $A^\sigma_\mu$ is obtained by introducing a qubit ancilla, as illustrated in \figref{fig:KRC_circuit}(d).

We use an ancilla qubit in state $\ket{+}$ to construct the measurement circuits as shown in Figure~\ref{fig:KRC_circuit}(e), where the $+$ or $-$ sign corresponds to the measurement of $A^{\mu\sigma}_v$ or $A^{\mu^2\sigma}_v$, respectively. The previous measurement circuits for $A^{\mu}$ and $A^{\mu^2}$, which are based on an ancilla qutrit, cannot be used here. The following algbera relations are essential in constructing the entangling circuit:
\begin{align}
    C\hat{C}_{i\to j} (\hat{X}_j^{Z_j}) C\hat{C}_{i\to j} &= (C_i\hat{X}_j^{Z_j}) \hat{X}_j^{Z_j} \;,\\
    C\hat{C}_{i\to j} (\hat{X}_j) C\hat{C}_{i\to j} &= (C_i\hat{X}_j)\hat{X}_j\;,
\end{align}
where $C\hat{C}_{i\to j}$ is the $C\hat{C}$ gate with the $i$-th qubit and the $j$-th qutrit being the control and target, respectively. The $C_i\hat{X}_j^{Z_j}$ and $C_i\hat{X}_j$ gates apply $\hat{X}_j^{Z_j}$ and $\hat{X}_j$ to the qutrit and qubit on edge $j$, controlled by the $i$-th qubit. It can be readily verified that the entangling gates with $+$ or $-$ sign in Figure~\ref{fig:KRC_circuit}(e) map the stabilizer $X_a$ of the ancilla qubit to $X_a \otimes A^{\mu\sigma}$ or $X_a \otimes A^{\mu^2\sigma}$, respectively. Consequently, by measuring the ancilla in the $X$ basis, we realize the $A^{\mu\sigma}$ or $A^{\mu^2\sigma}$ measurement in $K^{R,C_2}_s$.

\section{Proof of the completeness of anyon errors}
\label{appendix:Proof_completeness_anyon_ops}
In this section, we provide proof for the completeness of anyon errors, as stated by the proposition:
\begin{proposition}
\label{lemma:completeness_anyon_errors}
    The $F^{R,C;u,v}_{\rho_h}$ and $F^{R,C;u,v}_{\rho_v}$ form a complete orthonormal basis for all operators acting on the physical qudits of $\mathcal{D}(S_3)$ on a two-dimensional square lattice.
\end{proposition}
\begin{proof}[Proof of Proposition~\ref{lemma:completeness_anyon_errors}]
To demonstrate that the shortest ribbon operators form an orthonormal basis for the linear operators acting on the $|G|^{|L|}$-dimensional Hilbert space, where $|G|=6$ is the dimension of the local Hilbert space $\mathcal{H} = \mathbb{C}G$, and $|L|$ is the number of edges in the two-dimensional directed square lattice, we need to prove the following points:
\begin{enumerate}
    \item \textbf{Cardinality:} There are $|G|^{2|L|}$ shortest ribbon operators.
    \item \textbf{Orthonormality:} These operators are orthonormal under the operator inner product defined by the trace.
\end{enumerate}

To prove the cardinality, we utilize the fact that each site $s=(v,p)$ corresponds one-to-one with its vertex $v$. Consequently, each $\rho_h$ or $\rho_v$ corresponds to a horizontal or vertical edge connecting neighboring sites, respectively. Therefore, the total number of $\rho_h$ and $\rho_v$ is $|L|$. We determine the total number of ribbon operators $F^{R,C;u,v}_{\rho}$ supported on an arbitrary ribbon $\rho$ by summing over all degrees of freedom:
\begin{equation}
        \sum_{R,C;u,v}1 =\sum_{R,C}|C|^2|R|^2 = \sum_{C}|C|^2|Z(C)| = |G|\sum_{C}|C|= |G|^2\;,
\end{equation}
where in the second equality, we use the fact that for each irrep $(R,C)$, the basis is $\{\ket{c}\otimes \ket{j}: c\in C, 1 \leq j \leq |R|\}$. In the third equality, we use the fact that the multiplicity of an irrep $R$ in the fundamental representation of the group $Z(C)$ is $|R|$, so $\sum_{R}|R|^2 = |Z(C)|$. Then in the fourth equality, we used $|C||Z(C)| = |G|$. Therefore, each $\rho_h$ or $\rho_v$ supports $|G|^2$ ribbon operators, resulting in a total of $|G|^{2|L|}$ shortest ribbon operators.

In the remainder of this section, we will focus exclusively on the shortest ribbons $\rho,\rho_1,\rho_2 \in \{\rho_h\}\cup \{\rho_v\}$. We will prove the orthonormality by demonstrating that the shortest ribbon operators obey the following orthogonality relation:
\begin{equation} \label{eqn:ortho_ribbon_ops_1}
    \mathrm{Tr}\left(F^{R_1,C_1;u_1,v_1\dagger}_{\rho_1} F^{R_2,C_2;u_2,v_2\dagger}_{\rho_2}\right) = \frac{|R_1|}{|Z(C_1)||G|} \delta_{\rho_1,\rho_2}\delta_{R_1,R_2}\delta_{C_1,C_2} \delta_{u_1,u_2}\delta_{v_1,v_2} \mathrm{Tr}(\bm{1})\;,
\end{equation}
where $\Tr(\bm{1}) = |G|^{2|L|}$ is the trace over the entire Hilbert space.

Indeed, by Equation (B68) of \cite{Bombin2008PRB_NonAbelian}, when $\rho_1=\rho_2=\rho$, we have:
\begin{equation}
    \mathrm{Tr}\left(F^{R_1,C_1;u_1,v_1\dagger}_{\rho} F^{R_2,C_2;u_2,v_2\dagger}_{\rho}\right) = \frac{|R_1|}{|Z(C_1)||G|}\delta_{R_1,R_2}\delta_{C_1,C_2} \delta_{u_1,u_2}\delta_{v_1,v_2} \mathrm{Tr}(\bm{1})\;.
\end{equation}
Thus it is sufficient to prove that for $\rho_1\neq \rho_2$, and when $(R_1,C_1)$ and $(R_2,C_2)$ are not both the trivial irrep (i.e. $([+],\{e\})$), the ribbon operators are orthogonal:
\begin{equation} \label{eqn:ortho_ribbon_ops_2}
    \mathrm{Tr}\left(F^{R_1,C_1;u_1,v_1\dagger}_{\rho_1} F^{R_2,C_2;u_2,v_2\dagger}_{\rho_2}\right) = 0\;.
\end{equation}
We first prove \eqref{eqn:ortho_ribbon_ops_2} for the case when $\rho_1$ and $\rho_2$ do not overlap. From the definition \eqref{eqn:LT_ops} of the $L$ and $T$ operators, we have
\begin{equation}
    \Tr(L^g_{\tau}) = \delta_{g,e}\Tr(\bm{1})\;,\quad |G|\Tr(T^h_{\tau^\prime}) = \Tr(\bm{1})\;,
\end{equation}
which leads us to claim that, when $(R,C)\neq ([+],\{e\})$:
\begin{equation} \label{eqn:trace_zero_ribbon_ops}
    \Tr(F^{R,C;u,v}_{\rho}) = 0\;.
\end{equation}
Indeed, for $C\neq \{e\}$, from the definitions \eqnref{eqn:shortest_F} of $F^{R,C;u,v}_{\rho}$, there is a nontrivial $L^{c}_{\tau}$ operator for $c \neq e$ in the expression of $F^{R,C;u,v}_{\rho}$, which has zero trace, indicating \eqref{eqn:trace_zero_ribbon_ops}. For $C=\{e\}$ and $R\neq [+]$, we have
\begin{equation}
    \Tr(F^{R,C;u,v}_{\rho}) = \frac{|R|}{|Z(C)|} \sum_{h\in G} \Gamma^{R}_{jj^\prime}(h)\Tr(T^{h}_{\tau^\prime}) = \frac{|R|}{|Z(C)||G|} \Tr(\bm{1}) \sum_{h \in G} \Gamma^{R}_{jj^\prime}(h) = 0\;,
\end{equation}
where in the last equality we use the orthogonality relations for matrix elements of irreducible representations of a group $G$:
\begin{equation}
    \sum_{h\in G}\Gamma^{\alpha}_{ij}(\bar{h})\Gamma^{\beta}_{kl}(h) = \frac{|G|}{|\alpha|} \delta_{i,l}\delta_{j,k}\delta_{\alpha,\beta}\;,
\end{equation}
for $\alpha = [+]$ and $\beta = R$. Therefore, we complete the proof of \eqref{eqn:trace_zero_ribbon_ops}, which immediately implies \eqref{eqn:ortho_ribbon_ops_2} when $\rho_1$ and $\rho_2$ do not overlap, because we can take the trace separately for the two ribbon operators, and then the nontrivial $(R_1,C_1)$ or $(R_2,C_2)$ gives the zero trace.

Next, we prove \eqref{eqn:ortho_ribbon_ops_2} for the case when $\rho_1$ and $\rho_2$ overlap. There are only two possible overlapping patterns between a horizontal and a vertical shortest ribbon as follows:
\begin{equation}
    \mathrm{\RNum{1}}:\ \begin{tikzpicture}[baseline={(current bounding box.center)}]
    \foreach \i in {0,1}
        {\foreach \j in {0,1}{
            \draw[thick] (\i,\j+1) -- (\i,\j);
            \draw[thick] (\i,\j+1) -- (\i+1,\j+1);
            \draw[thick] (\i,\j) -- (\i+1,\j);
            \draw[thick] (\i+1,\j+1) -- (\i+ 1,\j);
            }}
    \draw (0.5,0.5)--(1,1);
    \draw[ultra thick] (0.5,0.5)-- (1.5,0.5);
    \foreach \i in {0,1}{           
        \draw[ultra thick] (\i,1) -- (0.5+\i,0.5);
    }
    \draw[ultra thick] (0,1) -- (1,1);
    \draw (0,1)--(0.5,1.5);
    \draw[ultra thick] (0.5,0.5)-- (0.5,1.5);
    \foreach \i in {0,1}{           
        \draw[ultra thick] (0,\i+1) -- (0.5,\i+0.5);
    }
    \draw[ultra thick] (0,1)--(0,2);
    \node at (1.5,0.25) {$\rho_h$};
    \node at (0.75,1.75) {$\rho_v$};
    \node at (0.15,1.5) {$1$};
    \node at (0.75,1.25) {$2$};
    \node at (0.75,0.25) {$3$};
    \end{tikzpicture}\;,\qquad 
    \mathrm{\RNum{2}}:\ \begin{tikzpicture}[baseline={(current bounding box.center)}]
    \foreach \i in {0,1}
        {\foreach \j in {0,1}{
            \draw[thick] (\i,\j+1) -- (\i,\j);
            \draw[thick] (\i,\j+1) -- (\i+1,\j+1);
            \draw[thick] (\i,\j) -- (\i+1,\j);
            \draw[thick] (\i+1,\j+1) -- (\i+ 1,\j);
            }}
    \draw (0.5,1.5)--(1,2);
    \draw[ultra thick] (0.5,1.5)-- (1.5,1.5);
    \foreach \i in {0,1}{           
        \draw[ultra thick] (\i,2) -- (0.5+\i,1.5);
    }
    \draw[ultra thick] (0,2) -- (1,2);
    \draw (1,1)--(1.5,1.5);
    \draw[ultra thick] (1.5,0.5)-- (1.5,1.5);
    \foreach \i in {0,1}{           
        \draw[ultra thick] (1,\i+1) -- (1.5,\i+0.5);
    }
    \draw[ultra thick] (1,1)--(1,2);
    \node at (1.75,0.5) {$\rho_v$};
    \node at (0.5,1.25) {$\rho_h$};
    \end{tikzpicture}\;,
\end{equation}
where we label the edges of case \RNum{1} by $1$ through $3$ for later convenience. 

In case \RNum{1}, we select $\rho_1 = \rho_h$ and $\rho_2 = \rho_v$. Through a case-by-case analysis, we will show:
\begin{equation}
    \mathrm{Tr}\left(F^{R_1,C_1;u_1,v_1\dagger}_{\rho_h} F^{R_2,C_2;u_2,v_2\dagger}_{\rho_v}\right) = 0\;.
\end{equation}
Indeed, if $C_1\neq \{e\}$, the $L^{c}_{\tau_3}$ for $c\in C_1$ has a zero trace, yielding the desired result. For $C_1 = \{e\}$, if $C_2 \neq \{e\}$ for the vertical ribbon $\rho_v$, the operator $T^{h}_{\tau_2^\prime}L^{c}_{\tau_2}$ acting on edge $2$ for $c\in C_2,h\in G$ has a zero trace, as the $L$ operator is off-diagonal while the $T$ operator is diagonal. If $C_1=C_2=\{e\}$, but not both $R_1$ and $R_2$ are $[+]$, the two ribbon operators are non-overlapping, thus we can use \eqref{eqn:trace_zero_ribbon_ops} to obtain the zero trace. A similar analysis for case \RNum{2} also yields the desired result.

Collecting all the above arguments, we have proved the orthogonality condition \eqref{eqn:ortho_ribbon_ops_1} for the shortest ribbon operators. By appropriately normalizing these ribbon operators, we obtain a complete orthonormal basis for all operators on the Hilbert space of $\mathcal{D}(S_3)$, thus completing the proof of the proposition.
\end{proof}

\end{document}